\documentclass[11pt,leqno]{report}

\usepackage[english]{babel}
\usepackage[colorlinks,allcolors=black]{hyperref}
\usepackage{amsfonts}
\usepackage{amsmath}
\usepackage{amssymb}
\usepackage{amsthm}
\usepackage{vmargin}
\usepackage{graphicx}
\usepackage{booktabs}
\usepackage{float}
\usepackage{algorithm}
\usepackage{algorithmicx}
\usepackage{algpseudocode}

\usepackage{parskip}
\usepackage{dsfont}
\usepackage{subcaption}
\usepackage[bb=boondox]{mathalfa}
\usepackage{listingsutf8}
\usepackage{tikz}
\usepackage{chngcntr}
\usepackage{xcolor}
\definecolor{lightgray}{gray}{0.7} % 0 = nero, 1 = bianco
\usepackage{mdframed}
\newmdenv[
	skipabove=10pt,
	skipbelow=10pt,
	leftline=true,
	topline=false,
	rightline=false,
	bottomline=false,
	linecolor=lightgray,
	linewidth=2pt,
	leftmargin=10pt,
	innerleftmargin=10pt
]{leftbarquote}
\usepackage{etoolbox}
\counterwithin{figure}{section}

\DeclareMathOperator{\supp}{supp}
\DeclareMathOperator{\col}{col}
\DeclareMathOperator{\Markov}{Markov}
\DeclareMathOperator{\ind}{\rotatebox{90}{$\models$}}

\newcommand{\code}[1]{
	\colorbox{gray!15}{\texttt{\detokenize{#1}}}
}

\theoremstyle{definition}
\newtheorem{theorem}{Theorem}[section]

\newtheorem{lemma}[theorem]{Lemma}
\newtheorem{proposition}[theorem]{Proposition}

\newtheorem{Assumption}{Assumption}

\newtheorem{definition}[theorem]{Definition}
 
\newtheorem*{recall}{Recall}
\newtheorem*{note}{Note}

\theoremstyle{remark}
\newtheorem{remark}[theorem]{Remark}

\newcommand{\norm}[1]{\left\|#1\right\|}

\numberwithin{equation}{section}

\setmarginsrb{33mm}{30mm}{33mm}{40mm} {0mm}{10mm}{0mm}{10mm}

\begin{document}

\begin{titlepage}
	\begin{figure}[!htb]
		\centering
		\includegraphics[keepaspectratio=true,scale=0.5]{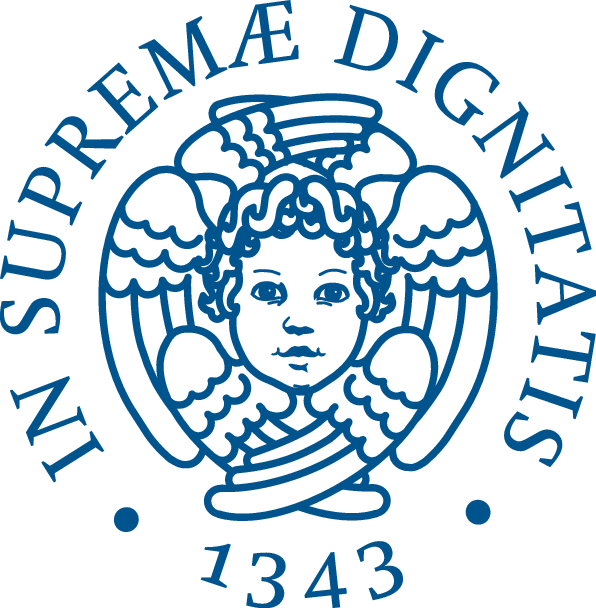}
	\end{figure}
	
	\begin{center}
		\LARGE{UNIVERSITY OF PISA}
		\vspace{5mm}
		\\ \large{DEPARTMENT OF MATHEMATICS}
		\vspace{5mm}
		\\ \large{Bachelor's Degree in Mathematics}
	\end{center}
	
	\vspace{8mm}
	\begin{center}
		{\Large{\bf Approximation Properties of Evolutionary Dynamics in Continuous–Time Finite State Space Games\\ \vspace{5mm} \large From Stochastic Crowds to Deterministic Limits}}
	\end{center}
	\vspace{8mm}
	
	\begin{minipage}[t]{0.47\textwidth}
		{\large{Supervisor}{\normalsize\vspace{3mm}
				\bf\\ \large{Prof. Andrea Agazzi} \normalsize\vspace{6mm}\bf\\}}
		{\large{Revisor}{\normalsize\vspace{3mm}
				\bf\\ \large{Eng. Leonardo Pedroso} \normalsize\vspace{6mm}\bf\\}}
		{\large{Co–Supervisor}{\normalsize\vspace{3mm}
				\bf\\ \large{Prof. Dario Trevisan} \normalsize\vspace{3mm}\bf\\}}
	\end{minipage}
	\hfill
	\begin{minipage}[t]{0.47\textwidth}\raggedleft
		{\large{Candidate}{\normalsize\vspace{3mm} \bf\\ \large{Pietro Grassi}}}
	\end{minipage}
	
	\vspace{8mm}
	\hrulefill
	\\\centering{\large{ACADEMIC YEAR 2024/2025}}
	
\end{titlepage}

\newpage
\thispagestyle{empty}
\vspace*{\stretch{1}}
\begin{flushright}
	\itshape to Daniela, Laura, Gianni and Giordano
\end{flushright}
\vspace{\stretch{2}}

\pagenumbering{gobble}
\tableofcontents
\clearpage

\newpage
\clearpage

\begin{abstract}
	This thesis studies the convergence of finite–population stochastic evolutionary dynamics to their deterministic mean–field limit in continuous–time finite state space games. We first develop refined ergodic theorems for Markov chains with a single positive–recurrent class, guaranteeing the existence of a unique invariant distribution and almost–sure convergence of time averages. Next, we prove that the mean–field model—described by a system of Lipschitz–continuous ordinary differential equations— admits a unique solution that depends continuously on its initial condition and that constitutes the almost–sure limit for the empirical distributions with fixed policy. Furthermore, we show that every Mixed Stationary Nash Equilibrium of the mean–field game is approximated by a Nash equilibrium of the corresponding \(N\)–player game within an error \(\varepsilon\) for sufficiently large \(N\). We finally demonstrate, by Kurtz's theorem, that the empirical state–policy distribution converges in probability to the mean–field trajectory. Numerical simulations conducted in \textsc{matlab} confirm the theoretical \(\mathcal{O}(N^{-1/2})\) convergence rate in both models across a range of population sizes.
\end{abstract}

\pagenumbering{arabic}
\newpage

\chapter*{Introduction}
\addcontentsline{toc}{chapter}{Introduction}

In continuous–time evolutionary mean–field games, a key mathematical question is whether the finite–population stochastic process converges to its deterministic mean– field counterpart as the population size grows. The purpose of this thesis is to rigorously establish this convergence.

We commence with the Ergodic Theorem, which describes the conditions under which a Markov chain reaches a unique stationary distribution, thereby characterizing its long–term behavior. Our approach for proving the convergence relies on both Kolmogorov's Strong Law of Large Numbers and Kurtz's Theorem. The latter is a fundamental result in the theory of stochastic processes, guaranteeing that, under appropriate assumptions, proper empirical distributions in the finite–population system converge to the solutions of deterministic equations in the continuum model. The objective of this work is to apply these results to our specific model and verify, both theoretically and numerically, their implications.

\section*{Model Overview}

We consider a finite population model where each player has an individual state and takes asynchronous actions independently in continuous–time. The state to which each player transitions after taking an action is determined by a Markov kernel. Furthermore, each player follows a policy, i.e., a map from the state space to the set of probability distributions over all available actions.

First of all, we focus on the case in which policies are fixed, then we analyze the evolutionary dynamics, allowing players to change policy according to a revision protocol. In both cases, as the population size grows, the stochastic fluctuations diminish, so the empirical joint state–action and state–policy distributions converge to the solution of the ordinary differential equations (ODEs) which describe the deterministic mean–field model. The focus of this thesis is to rigorously justify this convergence using both Kolmogorov's Strong Law of Large Numbers and Kurtz’s Theorem.

\section*{Main Results}

This work provides detailed proofs of the following key results:\begin{itemize}
	\item Ergodic Theorem with a Unique Recurrent Class: We prove that the continuous–time Markov chain governing the players’ joint state has a unique invariant distribution and that long–run time averages converge almost surely to it. The argument is developed first for chains with one positive–recurrent communicating class in a general countable space, and then specialised to our finite state space.
	\item Well–Posedness of the Mean–Field Model: We establish the existence, uniqueness, and Lipschitz continuity of the solution to the ODE characterizing the mean–field model. Additionally, we show that as the number of players increases, the empirical state–policy distribution of the finite–population system converges almost surely to the deterministic mean–field model. The proof relies on Kolmogorov's Strong Law of Large Numbers and the Continuous Mapping Theorem.
	\item Approximation of the Finite–Population Game by the Mean–Field Model: We show that for a sufficiently large population, the behavior of the finite–population system remains close to that of the mean–field model. This follows from the result above and the Dominated Convergence Theorem, ensuring that as the number of players increases, the finite–population dynamics approximate the mean–field solution.
	\item Convergence of Finite Population Evolutionary Dynamics to the Mean–Field Model: We prove that, under appropriate conditions, the stochastic evolutionary dynamics of the finite–population system converge in probability to the deterministic mean–field evolutionary dynamics as the population size tends to infinity. This result is based on the convergence result above and an application of Kurtz’s theorem to the empirical state–policy distribution over time.
\end{itemize}

\section*{Structure of the Thesis}

The remainder of this thesis is organized as follows:\begin{itemize}
	\item Chapter 1 provides an overview of relevant concepts about Continuous–Time Markov Chains, which are essential for the mathematical analysis, and concludes with the Ergodic Theorem.
	\item Chapter 2 describes the model, detailing the finite–population model and the corresponding mean–field approximation.
	\item Chapter 3 presents the theoretical foundations of stochastic evolution and its deterministic approximation, culminating in the statement of Kurtz’s Theorem.
	\item Chapter 4 introduces the mathematical framework, including the choice of topology and the construction of the Q–matrix that generates the process.
	\item Chapter 5 presents the theoretical results, focusing on the proofs of the ones specified above.
	\item Chapter 6 discusses practical results, including numerical experiments and computational verification of the theoretical findings.
	\item Chapter 7 concludes the thesis, summarizing the main contributions and discussing potential future directions.
\end{itemize}

This thesis systematically applies established mathematical tools to a well–defined problem, ensuring the correctness and clarity of key results in the framework of mean– field evolutionary games.

\section*{Real–World Applicability of This Thesis}

A natural arena for the theory developed here is token–driven resource–sharing schemes in which many participants contend for common assets \cite{Pedroso2} \cite{Pedroso}.

Fair, system–efficient allocation can be induced by a non–transferable–token mechanism: every participant holds a private balance that grows or shrinks whenever they access the resources. In such a setting:\begin{itemize}
	\item \textbf{Real–time decisions}: each participant decides continuously in time when to request the resources that match their needs;
	\item \textbf{Asynchrony}: requests are uncoordinated and sporadic, so players act at different times;
	\item \textbf{Jump dynamics}: each player’s status is captured by an integer–valued token balance that changes in discrete jumps after every access;
	\item \textbf{Congestion coupling}: the payoff from using a resource decreseas with its current load, so players’ decisions are interdependent.
\end{itemize}

This token–economy model is precisely an example of the class of dynamic games studied in the present paper, illustrating its relevance to real–world systems.

\chapter{Continuous-Time Markov Chains}

General references: \cite[Chapter~4]{Norris}, \cite[Chapters~2–3]{Durrett}, \cite[Chapter~5]{Shwartz}.

\paragraph{Notation.}

From this point forward we work in a probability space $(\Omega,\mathcal{F},\mathbb{P})$.

\section{Introductory Concepts}

General reference: \cite[Chapter~2]{Norris}.

\subsection{Q–Matrices and Their Exponentials}

\subsubsection{Definition and Bond with Stochastic Matrices}

\begin{definition}[Q–Matrix]
	Let $I$ be a countable index set. A \textit{Q–matrix} on $I$ is a matrix $Q=(q_{i,j})_{i,j\in I}$ satisfying:\begin{itemize}
		\item $q_{i,i}\in(-\infty,0]\ \forall i\in I$,
		\item $q_{i,j}\ge0\ \forall i\ne j$,
		\item $\sum_{j\in I}q_{i,j}=0\ \forall i\in I$.
	\end{itemize}
\end{definition}

A Q–matrix has the form:\begin{equation*}
	Q=\begin{bmatrix}-\sum_{j\ne 1}q_{1,j} & \ge0 &\cdots &\cdots& \ge0 \\ \ge0 & -\sum_{j\ne 2}q_{2,j} & \ge0 & \cdots & \ge0 \\ \vdots & \ddots & \ddots & \ddots & \vdots \\ \vdots &\ddots &\ddots &\ddots &\ge0\\ \ge0& \cdots &\cdots & \ge0 & -\sum_{j\ne |I|}q_{|I|,j} \end{bmatrix}.
\end{equation*}

Define $q_i:=\sum_{j\ne i}q_{i,j}=-q_{i,i}$, we give to the entries of $Q$ the following meaning:\begin{itemize}
	\item $q_{i,j}$ represents the rate of transition from state $i$ to $j$,
	\item $q_i$ denotes the total rate of leaving state $i$.
\end{itemize}

\begin{definition}[Stochastic Matrix]
	A matrix $P=(p_{i,j})_{i,j}$ is said to be \textit{stochastic} if:\begin{itemize}
		\item $p_{i,j}\ge0\ \forall i,j\in I$,
		\item $\sum_{j\in I}p_{i,j}=1\ \forall i\in I$.
	\end{itemize}
\end{definition}

A matrix $Q$ on a finite set $I$ is a Q–matrix if and only if $P(t):=e^{tQ}$ is a stochastic matrix $\forall t\ge0$.

\subsubsection{Generating Processes with Q–Matrices}

Let $P$ be a stochastic matrix given by $P=e^Q$, where $Q$ is a Q–matrix. Fix $m\in\mathbb{N}$ and let $(X_n^m)_{n\in\mathbb{N}}$ be discrete–time $\Markov\left(\lambda,\ e^{\frac{Q}{m}}\right)$. Now, define $X_{\frac{n}{m}}:=X_n^m$, then $\left(X_{\frac{n}{m}}\right)_{n\in\mathbb{N}}$ is discrete–time $\Markov\left(\lambda,\ \left(e^{\frac{Q}{m}}\right)^m\right)=\Markov(\lambda,P)$.

This constructions allows the formulation of arbitrarily fine discrete–time Markov chains, giving rise to $\Markov(\lambda,P)$ when sampled at integer times. This approach leads to continuous–time processes $(X_t)_{t\ge0}$ with Q–matrix, so that\begin{equation*}
	\mathbb{P}(X_{t_{n+1}}=i_{n+1}\mid X_{t_0}=i_0,\hdots,X_{t_n}=i_n)=p_{i_n,i_{n+1}}(t_{n+1}-t_n),
\end{equation*} where the transition probability satisfies\begin{equation*}
	\mathbb{P}_i(X_t=j):=\mathbb{P}(X_t=j\mid X_0=i)=p_{i,j}(t).
\end{equation*}

\subsection{Continuous–Time Random Processes}

\begin{definition}[Continuous–Time Random Process]
	Let $I$ be a countable set. A \textit{continuous–time random process} $(X_t)_{t\ge0}$ taking values in $I$ is a collection of random variables $X_t: \Omega\rightarrow I$.
\end{definition}

Since, in general, \[\mathbb{P}\left(\bigcup_{t\ge0}A_t\right)\ne\sum_{t\ge0}\mathbb{P}(A_t),\] we restrict our attention to right–continuous processes, which satisfy the property that for every $\omega\in\Omega$ and for all $t\ge0$, there exists $\epsilon>0$ such that: \begin{equation*}
	X_s(\omega)=X_t(\omega),\quad \forall s\in[t,\ t+\epsilon].
\end{equation*}

\begin{theorem}
	The probability of any event that depends on a right–continuous process can be completely determined by its finite–dimensional distributions. Specifically, the event probabilities can be derived from: \[\mathbb{P}(X_{t_0}=i_0,\hdots,X_{t_n}=i_n),\quad n\ge0,\quad 0\le t_0\le\hdots\le t_n,\quad i_0,\hdots,i_n\in I\]
	\begin{proof}
		It follows from the equivalence of probabilities agreeing on a $\pi$–system, as stated in \cite[Theorem~6.6.1]{Norris}.
	\end{proof}
\end{theorem}

We do not consider what happens to a process after potential explosion. Instead, it is convenient to adjoin a new state $\infty$ to $I$, ensuring that the process satisfies $X_t=\infty$ for all $t\ge\zeta$ (see Definition \ref{Explosion time}).

\begin{definition}[Minimal Process]
	A process $(X_t)_{t\ge0}$ is called \textit{minimal} if it satisfies: \[X_t=\infty\quad \forall t\ge\zeta\] for some explosion time $\zeta\in\mathbb{R}$ (see Definition \ref{Explosion time}).
\end{definition}

\subsection{Poisson Processes}

\begin{definition}[Exponential Random Variable]
	A random variable $T: \Omega\rightarrow[0,\infty]$ follows an \textit{exponential distribution} with parameter $\lambda\in[0,+\infty)$ if it satisfies \[\mathbb{P}(T>t)=e^{-\lambda t},\quad \forall t\ge0.\] We denote this by $T\sim E(\lambda)$.
\end{definition}

\begin{proposition}
	A random variable $T:\Omega\rightarrow(0,\infty]$ follows an exponential distribution if and only if it satisfies the memoryless property:\[\mathbb{P}(T>s+t\mid T>s)=\mathbb{P}(T>t),\quad \forall s,t\ge0.\]
	\begin{proof}
		See \cite[Theorem~2.3.1]{Norris}.
	\end{proof}
\end{proposition}

\begin{definition}[Poisson Process]
	A right–continuous process $(X_t)_{t\ge0}$ with values in $\mathbb{N}$ is a \textit{Poisson process} of rate $\lambda$ if its holding times $(S_n)_{n\in\mathbb{N}}$, that is, the durations spent in each state before transitioning, are independent exponentially distributed random variables with parameter $\lambda$ and its jump chain is given by $Y_n=n$.
\end{definition}

A such process can be represented by the following diagram, which highlights its nature of linear Markov process with transition rate $\lambda$.

\begin{center}
	\begin{tikzpicture}[node distance=1.5cm, >=stealth, auto]
		
		% Define nodes explicitly
		\node[circle, fill=black, inner sep=1.5pt, label=below:{0}] (s0) at (0,0) {};
		\node[circle, fill=black, inner sep=1.5pt, label=below:{1}] (s1) at (1,0) {};
		\node[circle, fill=black, inner sep=1.5pt, label=below:{2}] (s2) at (2,0) {};
		\node[circle, fill=black, inner sep=1.5pt, label=below:{3}] (s3) at (3,0) {};
		\node[circle, fill=black, inner sep=1.5pt, label=below:{4}] (s4) at (4,0) {};
		
		% Draw arrows with λ labels
		\draw[->] (s0) -- (s1) node[midway, above] {\(\lambda\)};
		\draw[->] (s1) -- (s2) node[midway, above] {\(\lambda\)};
		\draw[->] (s2) -- (s3) node[midway, above] {\(\lambda\)};
		\draw[->] (s3) -- (s4) node[midway, above] {\(\lambda\)};
		
		% Dashed line indicating continuation
		\draw[dashed] (s4) -++ (1,0);
		
	\end{tikzpicture}
\end{center}

The Q–matrix associated with the process $(X_t)_{t\ge0}$ is:\begin{equation*}
	Q=\begin{bmatrix}-\lambda&\lambda&&&\\&-\lambda&\lambda&&\\&&\ddots&\ddots&\quad \\&&&&\end{bmatrix}.
\end{equation*}

Let $J_n$ be the time in which $X_t$ switches from $n-1$ to $n$. By the strong law of large numbers it holds that $\mathbb{P}(J_n\rightarrow\infty)=1$, so there is no explosion and the law of $(X_t)_{t\ge0}$ is uniquely determined. Thus, the natural construction of a Poisson process of rate $\lambda$ consists of taking a sequence $(S_n)_{n\in\mathbb{N}}$ of independent $E(\lambda)$ random variables, setting $J_0:=0,\ J_n:=S_1+\hdots+S_n$ and then defining:\[X_t:=n\quad\mbox{if}\ J_n\le t<J_{n+1}.\]

\begin{theorem}[Markov Property]
	Let $(X_t)_{t\ge0}$ be a Poisson process of rate $\lambda$. Then the shifted process $(X_{s+t}-X_s)_{t\ge0}$ is also a Poisson process of rate $\lambda$ for all $s\ge0$, and it is independent of $(X_r:\ r\le s)$.
	\begin{proof}
		See \cite[Theorem~2.4.1]{Norris}.
	\end{proof}
\end{theorem}

\begin{theorem}[Strong Markov Property] \label{SMP}
	Let $(X_t)_{t\ge0}$ be a Poisson process of rate $\lambda$, and let $T$ be a stopping time of $(X_t)_{t\ge0}$. Then the shifted process $(X_{T+t}-X_T)_{t\ge0}$ is also a Poisson process of rate $\lambda$, independent of $(X_s:\ s\le T)$.
	\begin{proof}
		See \cite[Theorem~2.4.2]{Norris}.
	\end{proof}
\end{theorem}

\subsection{Jump Chain and Holding Times}

\begin{definition}[Jump Matrix]
	Let $Q$ be a Q–matrix on a countable set $I$, the \textit{jump matrix} $\Pi:=(\pi_{i,j})_{i,j\in I}$ of $Q$ is defined by\begin{equation*}
		\pi_{i,j}=\begin{cases}\frac{q_{i,j}}{-q_{i,i}}&\ \mbox{if}\ j\ne i,\ q_{i,i}\ne0\\ 0&\ \mbox{if}\ j\ne i,\ q_{i,i}=0\\ 1&\ \mbox{if}\ j=i,\ q_{i,i}=0\\0&\ \mbox{if}\ j=i,\ q_{i,i}\ne0	\end{cases},
	\end{equation*} so that $\Pi$ is a stochastic matrix.
\end{definition}

\begin{definition}[Continuous–Time Markov Chain]
	A minimal right–continuous process $(X_t)_{t\ge0}$ on $I$ is a \textit{Markov chain} with initial distribution $\lambda$ and generator matrix $Q$ if its jump chain $(Y_n)_{n\in\mathbb{N}}$ is a discrete–time $\Markov(\lambda,\Pi)$ and if for each $n\ge1$, conditional on $Y_0,\hdots,Y_{n-1}$, its holding times $S_1,\hdots,S_n$ are independent exponential random variables with parameters $q(Y_0),\hdots,q(Y_{n-1})$, respectively.
\end{definition}

\subsubsection{Constructing Markov Chains Based on the Jump Matrix}

Let $(Y_n)_{n\in\mathbb{N}}$ be a discrete–time $\Markov(\lambda,\Pi)$ and let $(T_n)_{n\in\mathbb{N}}$ be independent exponential random variables of parameter $1$, independent of $(Y_n)_{n\in\mathbb{N}}$. Set $S_n:=T_n/q(Y_{n-1})$ and $J_n:=S_1+\hdots+S_n$, then define\begin{equation*}
	X_t:=\begin{cases}Y_n&\ \mbox{if}\ J_n\le t<J_{n+1}\\\infty&\ \mbox{otherwise}\end{cases}.
\end{equation*}

\subsubsection{Constructing Markov Chains Based on Poisson Processes}

Consider an initial state $X_0=Y_0$ with distribution $\lambda$ and a family of independent Poisson processes $\{(N_t^{i,j})_{t\ge0}:\ i,j\in I,\ i\ne j\}$, where $(N_t^{i,j})_{t\ge0}$ has rate $q_{i,j}$. Then define\begin{align*}
	J_0:=&0,\\
	J_{n+1}:=&\inf\{t>J_n:\ N_t^{Y_n,j}\ne N_{J_n}^{Y_n,j}\ \mbox{for some}\ j\ne Y_n\},\\
	Y_{n+1}:=&\begin{cases}j&\ \mbox{if}\ J_{n+1}<\infty\wedge N_{J_{n+1}}^{Y_n,j}\ne N_{J_n}^{Y_n,j}\\i&\ \mbox{if}\ J_{n+1}=\infty\end{cases}.
\end{align*} Thus, the first jump time of $(N_t^{i,j})_{t\ge0}$ is exponential with parameter $q_{i,j}$, so, conditional on $Y_0=i$, $J_1$ is exponential with parameter $q_i=\sum_{j\ne i}q_{i,j}$, while $Y_1$ has distribution $(\pi_{i,j}:\ j\in I)$ and $J_1\ind Y_1$ (independent).

Suppose that $T$ is a stopping time of $(X_t)_{t\ge0}$, conditional on $X_0$ and on the process $(N_t^{k,l})_{t\ge0}$, as $(k,l)\ne(i,j)$, which are independent of $N_t^{i,j}$, then $\{T\le t\}$ depends only on $(N_s^{i,j}:\ s\le t)$, so $\tilde{N}_t^{i,j}:=N_{T+t}^{i,j}-N_T^{i,j}$ is a Poisson process of rate $q_{i,j}$ independent of $(N_s^{i,j}:\ s\le T)$ and of $X_0$ and $(N_t^{k,l})_{t\ge0},\ (k,l)\ne(i,j)$.

Hence, conditional on $T<\infty$ and $X_T=i$, then $(X_{T+t})_{t\ge0}\overset{(d)}{=}(X_t)_{t\ge0}$ (in law) and $(X_{T+t})_{t\ge0}\ind(X_s:\ s\le T)$. In particular, setting $T:=J_n$, we have that, conditional on $J_n<\infty$ and $Y_n=i$, then $S_{n+1}\sim E(q_i)$ and $Y_{n+1}$ has distribution $(\pi_{i,j}:\ j\in I)$ and $S_{n+1}\ind Y_{n+1}$, and independent of $Y_0,\hdots,Y_n$ and $S_1,\hdots,S_n$. 

Hence, $(X_t)_{t\ge0}$ is $\Markov(\lambda,Q)$ and has the strong Markov property.

\subsection{Explosion}

\begin{definition}[Explosion Time]\label{Explosion time}
	Consider a process with jump times $J_0,J_1,\hdots$ and holding times $S_1,S_2,\hdots$. The \textit{explosion time} $\zeta$ is given by\begin{equation*}
		\zeta:=\sup_{n\in\mathbb{N}}J_n=\sum_{n\in\mathbb{N}}S_n.
	\end{equation*}
\end{definition}

The following result provides condition for the non–explosion of a Markov chain.

\begin{theorem}[Non–Explosion]\label{Non-Explosion}
	Let $(X_t)_{t\ge0}\sim\Markov(\lambda,Q)$. Then $(X_t)_{t\ge0}$ does not explode if at least one of the following conditions holds:\begin{itemize}
		\item[(i)] $I$ is finite,
		\item[(ii)] $\sup_{i\in I}\{-q_{i,i}\}<\infty$,
		\item[(iii)] $X_0=i$ and $i$ is recurrent (see Definition \ref{Recurrent}) for the jump chain.
	\end{itemize}
	\begin{proof}
		Define $T_n:=q(Y_{n-1})\cdot S_n$. Then $(T_n)_{n\in\mathbb{N}}$ are independent $E(1)$ random variables and are independent of $(Y_n)_{n\in\mathbb{N}}$.
		
		In cases (i),(ii), it follows that $q:=\sup_{i\in I}q_i<\infty$, leading to\begin{equation*}
			q\cdot\zeta\ge\sum_{n\in\mathbb{N}}T_n=\infty
		\end{equation*} with probability $1$. Hence, $\zeta=\infty$.
	
		In case (iii), since $(Y_n)_{n\in\mathbb{N}}$ visits $i$ infinity often, say at times $N_1,N_2,\hdots$, we obtain\begin{equation*}
			q_i\cdot\zeta\ge\sum_{m\in\mathbb{N}}T_{N_m+1}=\infty
		\end{equation*} with probability $1$. Thus, $\zeta=\infty$.
	\end{proof}
\end{theorem}

\section{Further Properties}

General reference: \cite[Chapter~3]{Norris}.

\subsection{Class Structure}

\begin{definition}
	Let $i,j\in I$, we say that $i$ \textit{leads} to $j$ and write $i\to j$ if\begin{equation*}
		\mathbb{P}_i(X_t=j\ \mbox{for some}\ t\ge0)>0,
	\end{equation*} where $\mathbb{P}_i(\cdot):=\mathbb{P}(\cdot\mid X_0=i)$.
\end{definition}

\begin{definition}[Communication]
	Let $i,j\in I$, we say that $i$ \textit{communicates} with $j$ and write $i\leftrightarrow j$ if both $i\rightarrow j$ and $j\to i$ hold.
\end{definition}

\begin{theorem}\label{Thm Communication}
	Let $i\ne j$ be distinct states in $I$, the following are equivalent:\begin{itemize}
		\item $i\to j$,
		\item $i\to j$ for the jump chain,
		\item There exist states $i_0,\hdots,i_n\in I$ such that $i_0=i,\ i_n=j$ and $q_{i_0,i_1}\cdots~q_{i_{n-1},i_n}>0$.
	\end{itemize}
	\begin{proof}
		See \cite[Theorem~3.2.1]{Norris}.
	\end{proof}
\end{theorem}

\begin{proposition}
	The communication relation is an equivalence relation.
	\begin{proof}
		Reflexivity and symmetry follow directly from the definition of communication. To establish transitivity, we invoke Theorem \ref{Thm Communication}. Let $i,j,k\in I$ such that $i\leftrightarrow j$ and $j\leftrightarrow k$, we aim to show that $i\leftrightarrow k$ as well. To this end, we separately prove that both $i\to k$ and $k\to i$ hold:\begin{itemize}
			\item Since $i\leftrightarrow j$, it follows by definition that $i\to j$. Then, by Theorem \ref{Thm Communication}, there exist states $i_0,\hdots,i_n\in I$ with $i_0=i,\ i_n=j$, such that $q_{i,i_1}\cdots q_{i_{n-1},j}>0$. Similarly, from $j\leftrightarrow k$, there exist $j_0,\hdots,j_m\in I$ with $j_0=j,\ j_m=k$, such that $q_{j,j_1}\cdots q_{j_{m-1},k}>0$. Concatenating the two paths yields \[q_{i,i_1}\cdots q_{i_{n-1},j}\cdot q_{j,j_1}\cdots q_{j_{m-1},k}>0,\] which shows that $i\to k$.
			\item By an analogous argument, we conclude that $k\to i$.
		\end{itemize}
	\end{proof}
\end{proposition}

\subsection{Hitting Times, Absorption Probabilities and Average Times Taken}

Let $(X_t)_{t\ge0}$ be a Markov chain with generator matrix $Q$.

\begin{definition}[Hitting Time]
	The \textit{hitting time} of $A\subseteq I$ is the random variable $D^A$ defined by:\begin{equation*}
		D^A(\omega):=\inf\{t\ge0:\ X_t(\omega)\in A\}.
	\end{equation*}
\end{definition}

Since the chains we deal with are minimal, let $H^A$ be the hitting time of $A$ for the jump chain, then $\{H^A<\infty\}=\{D^A<\infty\}$ and $D^A=J_{H^A}$. The probability, starting from $i$, that $(X_t)_{t\ge0}$ ever hits $A$ is then:\begin{equation*}
	h_i^A:=\mathbb{P}_i(D^A<\infty)=\mathbb{P}_i(H^A<\infty).
\end{equation*}

\begin{definition}[Absorption Probability]
	When $A$ is a closed class of communication, the value $h_i^A$ is called the \textit{absorption probability} of $A$.
\end{definition}

Since the hitting probabilities are those of the jump chain, from the theory of discrete–time Markov chains we inherit the following result.

\begin{proposition}
	The vector of hitting probabilities $h^A:=(h_i^A:\ i\in I)$ is the minimal non–negative solution to:\begin{equation*}
		\begin{cases}
			h_i^A=1&\ \mbox{for}\ i\in A\\
			\sum_{j\in I}q_{i,j}h_j^A=0&\ \mbox{for}\ i\notin A
		\end{cases}.
	\end{equation*}
\end{proposition}

Analogously, we have a result for the average times taken.

\begin{definition}[Average Time Taken]
	The \textit{average time taken}, starting from $i$, for $(X_t)_{t\ge0}$ to reach $A$ is \[k_i^A:=\mathbb{E}_i[D^A],\] where $\mathbb{E}_i[\cdot]=\mathbb{E}[\cdot\mid X_0=i]$.
\end{definition}

\begin{proposition}
	If $-q_{i,i}>0\ \forall i\notin A$, the vector $k^A:=(k_i^A:\ i\in I)$ of expected hitting times is the minimal non–negative solution to\begin{equation*}
		\begin{cases}
			k_i^A=0&\ \mbox{for}\ i\in A\\
			-\sum_{j\in I}q_{i,j}k_j^A=1&\ \mbox{for}\ i\notin A
		\end{cases}.
	\end{equation*}
\end{proposition}

\subsection{Recurrence and Transience}

Let $(X_t)_{t\ge0}$ be a Markov chain with generator matrix $Q$.

\begin{definition}[Recurrent State]\label{Recurrent}
	We say a state $i\in I$ is \textit{recurrent} if\begin{equation*}
		\mathbb{P}_i(\{t\ge0:\ X_t=i\}\ \mbox{is unbounded})=1.
	\end{equation*}
\end{definition}

\begin{definition}[Transient State]
	We say a state $i\in I$ is \textit{transient} if\begin{equation*}
		\mathbb{P}_i(\{t\ge0:\ X_t=i\}\ \mbox{is unbounded})=0.
	\end{equation*}
\end{definition}

Note that if $(X_t)_{t\ge0}$ can explode starting from $i$, then $i$ is certainly not recurrent (so it is transient, as we highlight below).

In the following result, we establish the relation between the notions of recurrence and transience for a continuous–time Markov chain and for its jump chain.

\begin{theorem}[Recurrence and Transience from the Jump Chain] \label{Recurrence and Transience Jump Chain}
	The following hold:\begin{itemize}
		\item[(i)] If $i$ is recurrent for the jump chain $(Y_n)_{n\in\mathbb{N}}$ of $(X_t)_{t\ge0}$, then it is recurrent also for $(X_t)_{t\ge0}$,
		\item[(ii)] If $i$ is transient for the jump chain, then it is transient also for $(X_t)_{t\ge0}$,
		\item[(iii)] Every state is either recurrent or transient,
		\item[(iv)] Recurrence and transience are class properties.
	\end{itemize}
	\begin{proof}
		\color{white} A\color{black}
		\begin{itemize}
			\item[(i)] Suppose $i$ is recurrent for $(Y_n)_{n\in\mathbb{N}}$. If $X_0=i$ then $(X_t)_{t\ge0}$ does not explode and $J_n\to\infty$ by Theorem \ref{Non-Explosion}. Also $X_{J_n}=Y_n=i$ infinitely often, so $\{t\ge0:\ X_t=i\}$ is unbounded, with probability $1$.
			\item[(ii)] Suppose $i$ is transient for $(Y_n)_{n\in\mathbb{N}}$. If $X_0=i$ then:\begin{equation*}
				N:=\sup\{n\ge0:\ Y_n=i\}<\infty,
			\end{equation*} so $\{t\ge0:\ X_t=i\}$ is bounded by $J_{N+1}$, which is finite, with probability $1$, because $(Y_n:\ n\le N)$ cannot include an absorbing state.
			\item[(iii),(iv)] Follow from points (i) and (ii) applying the results for discrete–time Markov chains.
		\end{itemize}
	\end{proof}
\end{theorem}

We now give a result about recurrence, transience and the first passage times.

\begin{definition}[First Passage Time]
	The \textit{first passage time} $T_i$ of $(X_t)_{t\ge0}$ to state $i$ is defined by\begin{equation*}
		T_i(\omega):=\inf\{t\ge J_1(\omega):\ X_t(\omega)=i\}.
	\end{equation*}
\end{definition}

\begin{theorem}[Recurrence, Transience and First Passage Time]
	The following dichotomy holds:\begin{itemize}
		\item If $q_i=0$ or $\mathbb{P}_i(T_i<\infty)=1$, then $i$ is recurrent and $\int_0^\infty p_{i,i}(t)\,\mathrm{d}t=\infty$,
		\item If $q_i>0$ and $\mathbb{P}_i(T_i<\infty)<1$, then $i$ is transient and $\int_0^\infty p_{i,i}(t)\,\mathrm{d}t<\infty$.
	\end{itemize}
	\begin{proof}
		See \cite[Theorem~3.4.2]{Norris}.
	\end{proof}
\end{theorem}

The last result of this section shows that recurrence and transience are determined by any discrete–time sampling of a continuous–time Markov chain $(X_t)_{t\ge0}$.

\begin{proposition}
	Let $h>0$ be given and set $Z_n:=X_{n\cdot h}$, then:\begin{itemize}
		\item If $i$ is recurrent for $(X_t)_{t\ge0}$, then it is recurrent for $(Z_n)_{n\in\mathbb{N}}$,
		\item If $i$ is transient for $(X_t)_{t\ge0}$, then it is transient for $(Z_n)_{n\in\mathbb{N}}$.
	\end{itemize}
\end{proposition}

\subsection{Invariant Distributions}

Let $(X_t)_{t\ge0}$ be a Markov chain with generator matrix $Q$.

\begin{definition}[Invariant Distribution]
	We say that $\lambda$ is \textit{invariant} for $(X_t)_{t\ge0}$ if \[\lambda Q=0.\]
\end{definition}

\begin{theorem}[Invariant Distributions and Jump Matrices]
	Let $\Pi$ be the jump matrix of $Q$ and $\lambda$ a measure, then the following are equivalent:\begin{itemize}
		\item $\lambda$ is invariant, i.e., $\lambda Q=0$,
		\item $\mu\Pi=\mu$, where $\mu_i=\lambda_iq_i$.
	\end{itemize}
	\begin{proof}
		It holds that $q_i\cdot(\pi_{i,j}-\delta_{i,j})=q_{i,j}$ for all $i,j$, so:\begin{equation*}
			\left(\mu(\Pi-I)\right)_j=\sum_{i\in I}\mu_i(\pi_{i,j}-\delta_{i,j})=\sum_{i\in I}\lambda_iq_{i,j}=(\lambda Q)_j.
		\end{equation*}
	\end{proof}
\end{theorem}

\begin{theorem}[Existence and Uniqueness of Invariant Distributions for Irreducible and Recurrent Processes]
	Suppose that $Q$ is irreducible and recurrent. Then $Q$ has an invariant measure $\lambda$, which is unique up to scalar multiples.
	\begin{proof}
		See \cite[Theorem~3.5.2]{Norris}.
	\end{proof}
\end{theorem}

\begin{definition}[Positive and Null Recurrence] \label{Positive and Null Recurrence}
	If $q_i=0$ or $m_i=\mathbb{E}_i[T_i]<\infty$ then we say that $i$ is \textit{positive} recurrent. Otherwise, a recurrent state $i$ is called \textit{null} recurrent.
\end{definition}

As in the discrete–time case, positive recurrence is closely related to the existence of an invariant distribution.

\begin{theorem}[Existence of Invariant Distributions and Positive Recurrence]\label{Invariant distributions and positive recurrence}
	Let $Q$ be an irreducible Q–matrix, then the following are equivalent:\begin{itemize}
		\item Every state $i\in I$ is positive recurrent,
		\item There exists a state $i\in I$ which is positive recurrent,
		\item $Q$ is non–explosive and has an invariant distribution $\lambda$. In this case, it holds that \[m_i=\frac{1}{\lambda_iq_i}\quad \mbox{for all}\ i\in I.\]
	\end{itemize}
	\begin{proof}
		See \cite[Theorem~3.5.3]{Norris}.
	\end{proof}
\end{theorem}

The last theorem of this section states an important translation result for irreducible non–explosive Q–matrices with invariant distributions.

\begin{theorem}
	Let $Q$ be an irreducible non–explosive Q–matrix having an invariant distribution $\lambda$. If $(X_t)_{t\ge0}$ is $\Markov(\lambda,Q)$, then so is $(X_{s+t})_{t\ge0}$ for any $s\ge0$.
	\begin{proof}
		See \cite[Theorem~3.5.6]{Norris}.
	\end{proof}
\end{theorem}

\subsection{Convergence to Equilibrium}

We finally investigate the limiting behaviour of $p_{i,j}(t)$ as $t\to\infty$ and its relation to invariant distributions.

\begin{theorem}[Convergence to Equilibrium]
	Let $Q$ be an irreducible non–explosive Q–matrix with semigroup $(P(t):\ t\ge0)$, and having an invariant distribution $\lambda$, then for all states $i,j\in I$, as $t\to\infty$, we have\begin{equation*}
		p_{i,j}(t)\to\lambda_j.
	\end{equation*}
	\begin{proof}
		See \cite[Theorem~3.6.2]{Norris}.
	\end{proof}
\end{theorem}

\begin{theorem}
	Let $Q$ be an irreducible Q–matrix and let $(X_t)_{t\ge0}$ be $\Markov(\nu,Q)$ for some distribution $\nu$, then for all $j\in I$, as $t\to\infty$, we have\begin{equation*}
		\mathbb{P}(X_t=j)\to\frac{1}{q_jm_j},
	\end{equation*} where $m_j=\mathbb{E}_j[T_j]$.
	\begin{proof}
		See \cite[Theorem~3.6.3]{Norris}.
	\end{proof}
\end{theorem}

\subsection{Ergodic Theorem} \label{ET}

We now present with the culmination of this chapter: the Ergodic Theorem. It will be used in Lemma \ref{Lemma 2.1} for proving the almost sure convergence of the empirical distribution to the deterministic one. This result states that long–run averages for continuous–time chains display the same sort of behaviour as in the discrete–time case, and for similar reasons.

\begin{theorem}[Ergodic Theorem]
	Let $Q$ be irreducible and let $\nu$ be any distribution. If $(X_t)_{t\ge0}$ is $\Markov(\nu,Q)$, then\begin{equation}
		\mathbb{P}\left(\frac{1}{t}\int_0^t\mathds{1}_{\{X_s=i\}}\,\mathrm{d}s\xrightarrow[t\to\infty]{}\frac{1}{m_iq_i}\right)=1,
	\end{equation} where $m_i=\mathbb{E}_i[T_i]$, i.e., the expected return time to state $i$.

	Moreover, in the positive recurrent case, consider the unique invariant distribution $(\lambda_i:\ i\in I)$, then for any bounded function $f: I\to\mathbb{R}$ we have\[\mathbb{P}\left(\frac{1}{t}\int_0^tf(X_s)\,\mathrm{d}s\xrightarrow[t\to\infty]{}\bar{f}\right)=1,\] where $\bar{f}:=\sum_{i\in I}\lambda_if_i$.
	\begin{proof}
		See \cite[Theorem~3.8.1]{Norris}.
	\end{proof}
\end{theorem}

\chapter{Model}\label{Model}

This chapter outlines the model used to study the dynamics of a population of players, focusing on both the finite–population system and its corresponding mean–field approximation as $N \to \infty$. We introduce the formal setup of the model, detailing the time evolution of players’ decisions, the states of the population, and the actions taken by each player. Then, we define key components such as state–action distributions, transition processes, and rewards. Finally, we explain the transition from the finite population model to the mean–field approximation, and it establishes the relationship between individual behaviour and collective dynamics.

\paragraph{Notation.}

\begin{itemize}
	\item The set of consecutive positive integer numbers $\{1,\hdots,N\}$, where $N\in\mathbb{N}$, is denoted by $[N]$.
	\item The $i$–th entry of a vector $x\in\mathbb{R}^n$ is denoted by $x_i$.
	\item The $n$–dimensional vector of zeros and ones are denoted by $\mathbb{0}_n$ and $\mathds{1}_n$, respectively. Alternatively, $\mathbb{0}$ and $\mathds{1}$ denote the vectors of zeros and ones of appropriate dimensions, respectively.
	\item The unit simplex of dimension $n$ is the set denoted by $\Delta_n:=\{x\in \mathbb{R}^n_{\ge0}:\ \mathds{1}^Tx\le1\}$.
	\item The indicator function of $a\in \mathcal{X}$ is denoted by $\mathds{1}_a: \mathcal{X}\rightarrow\{0,1\}$ and $\mathds{1}_a(x)=0$ if $x\ne a$ and $\mathds{1}_a(x)=1$ if $x=a$.
	\item The support of a function $f: \mathcal{X}\rightarrow\mathbb{R}$ is denoted by $\supp(f):=\{x\in\mathcal{X}: f(x)\ne0\}$.
	\item The set of all Borel probability measures on $\mathcal{A}$ is denoted by $\mathcal{P}(\mathcal{A})$. If a probability measure $\mu(t)\in\mathcal{P}(\mathcal{A})$ is a function of time, we use $\mu[a](t)$ to denote the mass on $a\in\mathcal{A}$.
	\item The expected value of a random variable $X$ is denoted by $\mathbb{E}[X]$.
\end{itemize}

\section{Model Formulation}

The game is analysed first for a population of $N$ players and then for the continuum mean–field model obtained in the limit $N\to\infty$. The approximation properties of the mean–field model constitute one of the main topics of this thesis.

\subsection{Finite Population Model} \label{Finite Population Model}

The finite population model is described by:\begin{itemize}
	\item \textbf{Time:} Each agent is equipped with an independent Poisson clock of rate $R_d>0$. When the a player's clock rings they take an action. 
	
	\item \textbf{States:} At each time $t$, each player $i\in[N]$ has an individual state $s\in\mathcal{S}$, where $\mathcal S$ is a finite set of states. Each player's state is characterized by the random variable $s^{i,N}(t)\in\mathcal S$. Between two consecutive rings of their clock the state of a player remains constant, so the sample paths of $s^{i,N}(\cdot)$ are piecewise constant with jumps at decision times.
	
	\item \textbf{Actions:} When a player is in state $s\in\mathcal S$, the actions available to them are the non–empty finite set $\mathcal A(s)$. We define the overall action set by $\mathcal A := \bigcup_{s\in\mathcal S}\mathcal A(s)$. The action that player $i\in[N]$ takes at time $t$ if its clock rings is characterized by the random variable $a^{i,N}(t)$. The empirical joint state–action distribution at each time $t$ is characterized by the random variable\[\hat{\mu}^N_{\mathcal{S}\times\mathcal{A}}[s,a](t):=\frac{1}{N}\sum_{i=1}^N\mathds{1}_{s^{i,N}(t)}(s)\cdot\mathds{1}_{a^{i,N}(t)}(a).\]
	
	\item \textbf{State Transitions:} Immediately after performing action $a$ in state $s$, a player jumps to a new state drawn according to the Markov transition kernel $\phi: \mathcal{S}\times\mathcal{A}\rightarrow\mathcal{P}(\mathcal{S})$. We denote by $\phi(\cdot\mid s,a)$ the distribution of the new state.
	
	\item \textbf{Single–Stage Reward:} The instantaneous payoff obtained by a player who takes action $a$ in state $s$ at time $t$, when the population distribution is $\hat\mu^{N}_{\mathcal S\times\mathcal A}(t)$, is given by $r(s,a,\hat{\mu}^N_{\mathcal{S}\times\mathcal{A}}[s,a](t))$, where $r$ is a real–valued function $r: \mathcal{S}\times\mathcal{A}\times\mathcal{P}(\mathcal{S}\times\mathcal{A})\rightarrow \mathbb{R}$.
	$r\bigl(s,a,\hat\mu^{N}*{\mathcal S\times\mathcal A}(t)\bigr)$. Notice that \( N R_d \hat{\mu}^N_{\mathcal{S} \times \mathcal{A}}[s', a'](t) \) corresponds to the expected flow of players taking action \( a' \) from state \( s' \), which can be used to model a decreasing reward upon congestion of a resource, for instance.
	
	\item \textbf{Discounted Infinite–Horizon Reward:}  Let $0<\beta<1$ be a discount factor and let $t_k$ be the time of player $i$'s $k$–th clock ring. Their infinite–horizon discounted reward is given by\[\mathbb{E}\left[\sum_{k=0}^\infty\beta^k\cdot r\left(s^{i,N}(t_k),\ a^{i,N}(t_k),\ \hat{\mu}^N_{\mathcal{S}\times\mathcal{A}}(t_k)\right)\right].\]
\end{itemize}

We impose the following mild continuity assumption on the single stage reward as follows, which is necessary to establish the existence of equilibria.

\begin{Assumption} \label{Assumption1}
	The single–stage reward function $r(s,a,\mu_{\mathcal S\times\mathcal A})$ is continuous with respect to $\mu_{\mathcal{S}\times\mathcal{A}}$.
\end{Assumption}

\subsection{Policies}

Given the information available to them, each player takes an action at each decision instant. A \textit{policy} is a function that maps each state to a probability distribution over the set of actions. We consider the following information structure on a policy of a player:\begin{itemize}
	\item \textbf{Oblivious}: The policy does not depend on the distribution of players' states. The dependence is indirect through the rewards of each action \cite{Adlakha}.
	\item \textbf{Markov}: The policy is present–dependent, in the sense that it depends only on the individual state of a player at the time their clock rings.
	\item \textbf{Stationary}: The policy is time–invariant: when a player chooses a policy, they plan to use it forever.
\end{itemize}

A policy satisfying the properties above can be characterized as a map $u: \mathcal{S}\rightarrow\mathcal{P}(\mathcal{A})$. The set of all such policies is denoted by $\mathcal{U}$ and defined as\[\mathcal{U}:=\{u: \mathcal{S}\rightarrow\mathcal{P}(\mathcal{A})\mid\supp(u(s))\subseteq\mathcal{A}(s),\ \forall s\in\mathcal{S}\}.\]

The subset of \textit{deterministic} policies, which map every state to a randomization of actions that places all the mass in a single action, is denoted by $\mathcal{U}_D\subseteq\mathcal{U}$ and defined as\[\mathcal{U}_D:=\{u\in\mathcal{U}\mid\forall s\in\mathcal{S}\ \exists a\in\mathcal{A}(s):\ \supp(u(s))=\{a\}\}.\]

The policy that player $i\in[N]$ plays at time $t$ is characterized by the random variable $u^{i,N}(t)$. In the sections to come, we introduce evolutionary dynamics to describe the time evolution of $u^{i,N}(t)$. Until then, we admit that the policy played by each player, i.e., the realization of $u^{i,N}(t)$, is constant in time. The empirical joint state–policy distribution is characterized by the random variable\begin{equation}\label{StatePolicy}
	\hat\mu^{N}_{\mathcal S\times\mathcal U}[s,u](t):=
	\frac1N\sum_{i=1}^N \mathds 1_{s^{i,N}(t)}(s)\cdot\mathds 1_{u^{i,N}(t)}(u).
\end{equation}

\subsection{Mean–Field Model} \label{Mean-Field Model}

In the \textit{mean–field model} we consider instead a continuum of players. The assumption on independent Poisson clocks allow to characterize the aggregate evolution of both joint state–action and state–policy distributions with ODEs. Let $\mu_{\mathcal S\times\mathcal A}(t)\in\mathcal{P}(\mathcal{S}\times\mathcal{A})$ be the deterministic joint state–action distribution at time $t$ and $\mu_{\mathcal{S}\times\mathcal{U}}(t)\in\mathcal{P}(\mathcal{S}\times\mathcal{U})$ the deterministic joint state–policy distribution at time $t$. In an infinitesimal interval $\mathrm d t$ the mass placed by $\mu_{\mathcal{S}\times\mathcal{U}}$ in state $s\in\mathcal{S}$\begin{itemize}
	\item increases by the proportion of players in other states whose clocks ring and, after taking action, end up in state $s$,
	\item decreases by the proportion of players in state $s$ whose clocks ring and, after taking action, leave state $s$.
\end{itemize} The resulting balance of inflows and outflows in state $s\in\mathcal{S}$ and policy $u\in\mathcal{U}$ yields
\begin{equation}
	\begin{aligned}
		d\mu_{\mathcal{S} \times \mathcal{U}}[s, u](t) =& \sum_{s' \in \mathcal{S}} \sum_{a' \in \mathcal{A}(s')} R_d\cdot \mu_{\mathcal{S} \times \mathcal{U}}[s', u](t)\,\mathrm{d}t\cdot \phi(s\mid s', a')\cdot u(a'\mid s')+\\&- R_d\cdot \mu_{\mathcal{S} \times \mathcal{U}}[s, u](t) \,\mathrm{d}t\cdot\sum_{s' \in \mathcal{S}} \sum_{a \in \mathcal{A}(s)} \phi(s'\mid s, a)\cdot u(a \mid s).
	\end{aligned}
	\label{Mean-Field dmu}
\end{equation}

When $\mathrm{d}t\to0$, we write\begin{equation}
	\begin{aligned}
		\dot{\mu}_{\mathcal{S} \times \mathcal{U}}[s, u](t) =& R_d\cdot \sum_{s' \in \mathcal{S}} \sum_{a' \in \mathcal{A}(s')} \phi(s \mid s', a')\cdot u(a' \mid s')\cdot \mu_{\mathcal{S} \times \mathcal{U}}[s', u](t)+ \\&- R_d \cdot\mu_{\mathcal{S} \times \mathcal{U}}[s, u](t),
	\end{aligned}
	\label{Mean-Field dotmu}
\end{equation} since \( \sum_{s' \in \mathcal{S}} \phi(s' | s, a) = 1 \).

The joint state–action distribution follows from the solution to \eqref{Mean-Field dotmu} for all \( s \in \mathcal{S} \) and \( a \in \mathcal{A} \) as\[
\mu_{\mathcal{S} \times \mathcal{A}}[s, a](t) = \int_{\mathcal{U}} \mu_{\mathcal{S} \times \mathcal{U}}[s, u](t) \cdot u(a \mid s) \cdot \mu_{\mathcal{U}}(\mathrm{d}u).
\]

\begin{note}
	Contrarily to that of the empirical distribution $\hat{\mu}_{\mathcal{S}\times\mathcal{U}}^N$, the time evolution of $\mu_{\mathcal{S}\times\mathcal{U}}$ is deterministically given by ODE \eqref{Mean-Field dotmu}. In fact, it admits a unique and Lipschitz–continuous solution with respect to the initial condition and, as the population size $N$ grows, the stochastic path of the empirical distribution converges almost surely to the deterministic trajectory (Lemma \ref{Lemma 2.1}).
\end{note}

\section{Equilibria} \label{Equilibria}

In game theory, equilibria constitute one of the most important concepts. Specifically, a Nash Equilibrium (NE) is a set of strategies where no player has an incentive to unilaterally deviate beacause their chosen strategy maximizes their payoff given the strategies of all other players. In this section, we analyze NE–like solution concepts for the class of mean–field games satisfying the properties of our model. The usefuleness of a solution concept is its ability to predict the outcome of a game.

\paragraph{Behavioral vs Mixed Policies.}  We say that the population follows a \textit{behavioral} policy if each player randomizes afresh at every decision instant, whereas under a \emph{mixed} policy each player draws a deterministic policy at the start of the game and follows it forever.

\subsection{Behavioural and Mixed Stationary Nash Equilibria}

Starting from an initial state distribution $\eta_0\in\mathcal P(\mathcal S)$ under a time–invariant joint state–action distribution $\mu_{\mathcal S\times\mathcal A}\in\mathcal{P}(\mathcal{S}\times\mathcal{A})$, the discounted infinite–horizon reward of playing policy $u\in\mathcal U$ is
\begin{equation}
	J(u, \eta_0, \mu_{\mathcal{S} \times \mathcal{A}}) := \mathbb{E} \left[ \sum_{k=0}^{\infty} \beta^k\cdot r(s_k, a_k, \mu_{\mathcal{S} \times \mathcal{A}}) \right],
	\label{DefinitionNE}
\end{equation} where \( s_k \sim \eta_k \), \( a_k \sim u(s_k) \), and \( \eta_k(s) = \sum_{s' \in \mathcal{S}} \sum_{a' \in \mathcal{A}(s')} \phi(s \mid s', a')\cdot u(a' \mid s')\cdot \eta_k(s') \). By abuse of notation, whenever \( \eta_0 \) has all the mass at a state \( s_0 \in \mathcal{S} \), we, alternatively, write \( J(u, s_0, \mu_{\mathcal{S} \times \mathcal{A}}) \). The formal definitions of BSNE and MSNE are as follows.

\begin{definition}[Behavioural Stationary Nash Equilibrium (BSNE)]
	Given a policy $u\in\mathcal{U}$ and a state distribution $\mu_{\mathcal{S}}\in\mathcal{P}(\mathcal{S})$, the pair $(u,\mu_{\mathcal S})$ is called a \textit{BSNE} in the discounted payoff mean–field game if\[
	J(u, \mu_{\mathcal{S}}, \mu_{\mathcal{S} \times \mathcal{A}}) \geq J(v, \mu_{\mathcal{S}}, \mu_{\mathcal{S} \times \mathcal{A}}), \quad \forall v \in \mathcal{U},
	\]	where \( \mu_{\mathcal{S} \times \mathcal{A}}(s, a) = \mu_{\mathcal{S}}(s)\cdot u(a \mid s) \) \(\forall s \in \mathcal{S},\ \forall a \in \mathcal{A} \), and\[
	\mu_{\mathcal{S}}(s) = \sum_{s' \in \mathcal{S}} \sum_{a' \in \mathcal{A}(s')} \phi(s \mid s', a') \cdot u(a' \mid s')\cdot \mu_{\mathcal{S}}(s'), \quad \forall s \in \mathcal{S}.
	\]
\end{definition}

\begin{definition}[Mixed Stationary Nash Equilibrium (MSNE)] \label{MSNE}
	A joint state–policy distribution \( \mu \in \mathcal{P}(\mathcal{S} \times \mathcal{U}_D) \) is said to be a \textit{MSNE} in the discounted payoff mean–field game if $\forall u \in \mathcal{U}_D$\begin{equation}
		\sum_{s \in \mathcal{S}} \mu(s, u) > 0\implies \left( J(u, s, \mu_{\mathcal{S} \times \mathcal{A}}) \geq J(v, s, \mu_{\mathcal{S} \times \mathcal{A}}),\ \forall v \in \mathcal{U}_D,\ \forall s \in \supp(\mu(\cdot, u))\right), 
		\label{MSNE1}
	\end{equation} where $\mu_{\mathcal{S} \times \mathcal{A}}(s, a) = \sum_{u \in \mathcal{U}_D} \mu(s, u) \cdot u(a \mid s)\ \forall s \in \mathcal{S},\ \forall a \in \mathcal{A}$, and\begin{equation}
		\mu(s, u) = \sum_{s' \in \mathcal{S}} \sum_{a' \in \mathcal{A}(s')} \phi(s \mid s', a') \cdot u(a' \mid s') \cdot\mu(s', u), \quad \forall s \in \mathcal{S},\ \forall u \in \mathcal{U}_D.
		\label{MSNE2}
	\end{equation} 
\end{definition}

The MSNE has an intuitive interpretation: in steady–state, each user follows a deterministic policy and has no incentive to switch to any other deterministic policy from any state they visit.

\begin{theorem}[Existence of MSNE]\label{Existence MSNE}
	Under Assumption \ref{Assumption1}, there exists at least one MSNE.
\end{theorem}

The proof of Theorem \ref{Existence MSNE} relies on writing the MSNE as a fixed point of a set–valued map and then using Kakutani’s fixed point theorem.

\subsection{Approximation of Equilibria of Finite Population Games}

In this section we introduce an equilibrium concept for the finite–population game. Theorem \ref{Theorem 3.2} will later show that, as $N \to \infty$, the mean–field game’s MSNE provides an accurate approximation. An analogous result for the BSNE follows by adapting the arguments as in \cite{Altman}.

Let $\mu_0 \in \mathcal{P}(\mathcal{S} \times \mathcal{U}_D)$ be a state–policy distribution, and let $\{u^i\}_{i\in[N]}$ be a collection of players’ policies. Assume that each player $j$'s initial state distribution $s^{j,N}(0)$ is drawn from $\mu_0(\cdot, u^j) / \mu_0(\mathcal{S}, u^j)$. The infinite-horizon discounted payoff of player $i$ is then\[J^{i,N}(\mu_0, u^1, u^2, \dots, u^i, \dots, u^N) = \mathbb{E} \left[ \sum_{k=0}^{\infty} \beta^k	\cdot r \left( s^{i,N}(t_k),\ a^{i,N}(t_k),\ \hat{\mu}^N_{\mathcal{S} \times \mathcal{A}}(t_k) \right) \right].\]

The definition of a weak MSNE in the discounted payoff finite population game is as follows.

\begin{definition}[Weak MSNE of finite–population games] \label{epsilon MSNE}
	Given an initial distribution $\mu_0\in\mathcal{P}(\mathcal{S}\times\mathcal{U}_D)$, the policy profile $\{u^i\}_{i\in[N]}$ is a \textit{weak MSNE} in the discounted payoff finite population game with initial condition $\mu_0$ if for every player $i\in[N]$\begin{equation}
		J^{i,N}(\mu_0, u^1, u^2, \dots, u^i, \dots, u^N) \geq J^{i,N}(\mu_0, u^1, u^2, \dots, v^i, \dots, u^N), \quad \forall v^i \in \mathcal{U}_D.
		\label{Weak MSNE}
	\end{equation} It is a \textit{weak $\epsilon$-MSNE} if the inequality holds within a tolerance $\epsilon > 0$.
	
	Intuitively, a weak MSNE is a policy profile from which no single player can gain by unilaterally switching to another deterministic policy, given that their initial state is sampled according to $\mu_0$.
\end{definition}

\section{Evolutionary Dynamical Model}\label{sec evolution}

Section \ref{Equilibria} characterised long–run outcomes through equilibrium notions. We now shift focus from outcomes to the behavioural path that leads to them. To that end we introduce an \textit{evolutionary–dynamical model} in which players occasionally \textit{revise} the deterministic policies they employ. We first derive a mean–field ODE that tracks the law of motion of the population and that—when $N\to\infty$—approximates the finite–population dynamics arbitrarily well. We then relate the rest points of this ODE to the equilibrium concepts of Section \ref{Equilibria}.

\subsection{Individual Evolutionary Dynamics}

Classical population–game models, with no explicit state dynamics, have been extensively analysed in the evolutionary–game literature (see, e.g., \cite{Sandholm}). Building on that foundation, we model how each player in our dynamic setting occasionally \emph{revises} their deterministic policy. Two behavioural traits motivate the construction:\begin{itemize}
	\item \textbf{Inertia}: policy changes occur only when a personal revision opportunity arises;
	\item \textbf{Myopia}: each revision is made unilaterally, without coordinating with the rest of the population.
\end{itemize} Hence, the process tracks the evolution of \emph{individual} (deterministic) policies rather than collective agreements—which would be implausible in a large population.

At time $t$ the empirical joint state–policy distribution is characterized by the random variable\begin{equation}
	\hat{\mu}^N[s, u](t) := \frac{1}{N} \sum_{i=1}^{N} \mathds{1}_{s^{i,N}(t)}(s)\cdot \mathds{1}_{u^{i,N}(t)}(u).
	\label{StatePolicy Evolutionary Dynamics}
\end{equation}

For any state \( s \in \mathcal{S} \), we denote the vector of the mass on each policy as \[ \hat{\mu}[s, \cdot](t) := \col(\hat{\mu}[s, u](t), u \in \mathcal{U}_D) \in \Delta_{|\mathcal{U}_D|},\] where the dependence on $t$ will often be omitted for conciseness.

The evolutionary model is described by:\begin{itemize}
	\item \textbf{Time}: Each player is equipped with an independent Poisson revision clock with rate $R_r>0$. When their clock rings, a player has the opportunity to change the policy they are currently playing. All revision clocks share the same rate and are independent of one another and of the players’ action clocks.
	\item \textbf{Policy transitions}: Consider the \textit{revision protocol}\[\rho:\mathbb R^{|\mathcal U_D|}\times\mathbb R^{|\mathcal U_D|}
	\rightarrow\mathbb R_{\ge0}^{|\mathcal U_D|\times|\mathcal U_D|}.\] Suppose a player is in state $s\in\mathcal S$ and follows policy $u\in\mathcal U_D$. Conditioned on a revision opportunity, they switch to $v\in\mathcal U_D$ at rate $\rho_{u,v}\left(F^{s}(\hat{\mu}),\ \hat{\mu}[s,\cdot]\right)$, where\begin{equation}
		F^s(\hat{\mu}) := \col \left( J(u_1, s, \hat{\mu}_{\mathcal{S} \times \mathcal{A}}), \dots, J(u_{|\mathcal{U}_D|}, s, \hat{\mu}_{\mathcal{S} \times \mathcal{A}}) \right),
		\label{PolicyTransitions}
	\end{equation} lists the infinite‐horizon discounted reward obtained from each policy starting at $s$. Here\[ \hat{\mu}_{\mathcal{S} \times \mathcal{A}}[s, a] = \sum_{u \in \mathcal{U}_D} \hat{\mu}[s, u] \cdot u(a \mid s),\ \forall s \in \mathcal{S}, \forall a \in \mathcal{A}\] and, with slight abuse of notation, $F^{s}{u}(\hat{\mu})=J(u,s,\hat{\mu}{\mathcal S\times\mathcal A})$.
\end{itemize}

Intuitively, if a player using policy \( u \in \mathcal{U}_D \) receives a revision opportunity, they switch to a policy \( v \in \mathcal{U}_D \) with probability \( \rho_{u,v}(F^s(\hat{\mu}), \hat{\mu}[s, \cdot]) / R_r \), and they continue to use the same policy with probability \( 1 - \sum_{v \neq u} \rho_{u,v}(F^s(\hat{\mu}),\ \hat{\mu}[s, \cdot]) / R_r \).

We make an assumption to ensure that the aforementioned switching probabilities are well defined and continuous as follows.

\begin{Assumption}
	The revision protocol \( \rho \) is Lipschitz continuous and \( \forall s \in \mathcal{S}, \forall u \in \mathcal{U}_D \) it holds:	\[
	R_r \geq \sup_{\mu \in \mathcal{P}(\mathcal{S} \times \mathcal{U}_D)} \sum_{v \in \mathcal{U}_D \setminus \{u\}} \rho_{u,v} \big( F^s(\mu),\ \mu[s, \cdot] \big).
	\]
	\label{Assumption2}
\end{Assumption}

The literature on evolutionary decision dynamics identifies physically meaningful classes of revision protocols. In this thesis, we restrict our attention to deterministic\footnote{Revision protocols are said to be deterministic if they generate unique solutions for the evolution of the aggregate decisions. Lipschitz continuity of $\rho$ in Assumption \ref{Assumption2} ensures that $\rho$ is deterministic.} revision protocols.

\subsection{Mean–Field Evolutionary Dynamics} \label{Mean-Field Evolutionary Dynamics}

When the population size is treated as a continuum, the revision process can be described deterministically through the evolution of the joint state–policy distribution. Let $\mu(t)\in\mathcal P(\mathcal S\times\mathcal U_D)$ denote this distribution at time $t$. For each state $s\in\mathcal S$ we collect the mass assigned to every deterministic policy in the vector\[\mu[s,\cdot]:=\col(\mu[s,\mu](t),\ u\in\mathcal{U}_D)\in\Delta{|\mathcal U_D|},
\] and, as before, we often suppress the explicit time argument for brevity.

Consider an infinitesimal interval $\mathrm dt$. Fix $s\in\mathcal S$ and $u\in\mathcal U_D$. The mass $\mu[s,u]$ changes because\begin{itemize}
	\item players who currently use $u$ may move in or out of state $s$ according to \eqref{Mean-Field dmu} (rate $R_d$),
	\item players in state $s$ may \emph{switch} into or out of policy $u$ when their revision clocks ring (rate $R_r$).
\end{itemize} Formally, $\forall s \in \mathcal{S},\ \forall u \in \mathcal{U}$\begin{equation*}
\begin{aligned}
	d\mu[s, u] =& \sum_{s' \in \mathcal{S}} \sum_{a' \in \mathcal{A}(s')} R_d\cdot \mu[s', u] \,\mathrm{d}t\cdot \phi(s \mid s', a')\cdot u(a' \mid s')+\\& - R_d \cdot \mu[s, u] \,\mathrm{d}t \cdot\sum_{s' \in \mathcal{S}} \sum_{a \in \mathcal{A}(s)} \phi(s' \mid s, a) \cdot u(a \mid s)+\\
	&+ \sum_{u' \in \mathcal{U}_D} R_r\cdot \mu[s, u'] \,\mathrm{d}t\cdot \frac{\rho_{u', u}(F^s(\mu),\ \mu[s, \cdot])}{R_r}+\\& - R_r\cdot \mu[s, u] \,\mathrm{d}t\cdot \sum_{u' \in \mathcal{U}_D}\frac{\rho_{u, u'}(F^s(\mu),\ \mu[s, \cdot])}{R_r} .
\end{aligned}
\end{equation*}

Taking the limit as $\mathrm dt\to0$ yields the ODE\begin{equation}
	\dot{\mu}[s, u] = f^d_{s,u}(\mu) + f^r_{s,u}(\mu), \quad \forall s \in \mathcal{S},\ \forall u \in \mathcal{U},
	\label{ODE8}
\end{equation}where\begin{itemize}
	\item $f^d_{s,u}(\mu) = R_d\cdot \sum_{s' \in \mathcal{S}} \sum_{a' \in \mathcal{A}(s')} \phi(s \mid s', a') \cdot u(a' \mid s')\cdot \mu[s', u](t) - R_d\cdot \mu[s, u]$,
	\item $f^r_{s,u}(\mu) = \sum_{u' \in \mathcal{U}_D} \mu[s, u']\cdot \rho_{u', u}(F^s(\mu),\ \mu[s, \cdot]) - \mu[s, u]\cdot \sum_{u' \in \mathcal{U}_D} \rho_{u ,u'}(F^s(\mu),\ \mu[s, \cdot])$.
\end{itemize}

Equation \eqref{ODE8} is the \emph{mean dynamic} associated with the revision protocol $\rho$. Theorem \ref{Theorem 4.1} states that, as $N\to\infty$, the empirical distribution of strategies converges to the solution of the mean dynamic.

\section{Summary}

The combination of the finite–population model, mean–field approximation, equilibrium analysis, and evolutionary dynamics provides a comprehensive framework for studying strategic interactions in large–population systems. In the next chapters, we delve deeper into the theoretical results and proofs that substantiate these concepts.

\chapter{Stochastic Evolution and Deterministic Approximation}

General reference: \cite[Chapter~10]{Sandholm}, \cite{Kurtz}.

In this chapter we delve into the theory of stochastic evolution and its deterministic approximation. We explore the underlying mathematical framework governing the evolution of strategies in a population game, with a focus on Kurtz’s Theorem. This theorem forms the foundation for further analysis, providing insight into the convergence of stochastic processes to deterministic dynamics as the population size grows large. This chapter introduces the stochastic evolutionary process, the associated Markov process, and a finite–horizon deterministic approximation that captures the expected evolution of the system.

\section{The Stochastic Evolutionary Process}

Consider a population of $N$ players (where $N$ denotes the population size) who repeatedly engage in a population game $F: X\rightarrow\mathbb{R}^n$, with pure strategy set $S=\{1,\hdots,n\}$.

The decision–making process of each player is modeled by a revision protocol\[\rho: \mathbb{R}^n\times X\rightarrow\mathbb{R}_{\ge0}^{n\times n}\] that takes current payoffs and the states of the population as inputs, and provides conditional switch rates $\rho_{i,j}(F(x),x)$ as outputs. Consequently, the feasible social states lie on the discrete grid:\begin{equation*}
	\mathcal{X}^N:=X\cap\frac{1}{N}\mathbb{Z}^n=\{x\in X:\ Nx\in\mathbb{Z}^n\}.
\end{equation*}

Each player is equipped with a Poisson alarm clock that rings at rate $R$, where \begin{equation*}
	\max_{x,i}\sum_{j\ne i}\rho_{i,j}(F(x),x)\le R<\infty.
\end{equation*} The ringing of an player's clock indicates a revision opportunity. If the player is currently playing strategy $i\in S$, they switch to strategy $j\ne i$ with probability $\frac{\rho_{i,j}}{R}$.

The clocks of the players ring independently of each other, and the choice of strategy is made independently of the timing of these rings. As the evolution proceeds, the players are only influenced by the history of process by way of the current value of the social state.

Given these assumptions, $\{X_t^N\}$ forms a continuous–time Markov process on the finite state space $\mathcal{X}^N$. To characterize this process, it is sufficient to specify its jumping rates $\{\lambda_x^N\}_{x\in\mathcal{X}^N}$ and the transition probabilities $\{P_{x,y}^N\}_{x,y\in\mathcal{X}^N}$. If the current state is $x\in\mathcal{X}^N$, then $N\cdot x_i$ players are playing strategy $i\in S$. 

Since players receive revision opportunities independently at an exponential rate $R$, the total arrival rate of revision opportunities in the population is $N\cdot R$.

Since the decisions to switch from one strategy to another are independent of the arrivals of revision opportunities, the probability that the next revision opportunity goes to an player playing strategy $i$ who then switches to strategy $j$ is\begin{equation*}
	\frac{N x_i}{N}\cdot\frac{\rho_{i,j}}{R}=\frac{x_i\rho_{i,j}}{R}.
\end{equation*}

Thus, the process can be described by the following:\begin{itemize}
	\item An initial state $X_0^N=x_0^N$,
	\item Jump rates $\lambda_x^N=N\cdot R$,
	\item Transition probabilities given by: \[P_{x,x+z}^N=\begin{cases}\frac{x_i\rho_{i,j}(F(x),x)}{R}&\ \mbox{if}\ z=\frac{1}{N}(e_j-e_i),\ i\ne j\\ 1-\sum_{i\in S}\sum_{j\ne i}\frac{x_i\rho_{i,j}(F(x),x)}{R}&\ \mbox{if}\ z=0\\0&\ \mbox{otherwise}\end{cases}.\]
\end{itemize}

\section{Finite–Horizon Deterministic Approximation:\\Kurtz's Theorem}

Let $\zeta_x^N$ be a random variable whose distribution describes the stochastic increment of $\{X_t^N\}$ from state $x$, that is:\begin{equation*}
	\mathbb{P}(\zeta_x^N=z)=P^N_{x,\ x+z}.
\end{equation*}

Define:\begin{align*}
	V^N:\mathcal{X}^N\to TX&\quad x\mapsto\lambda_x^N\cdot\mathbb{E}[\zeta_x^N],\\
	A^N:\mathcal{X}^N\to\mathbb{R}&\quad x\mapsto\lambda_x^N\cdot\mathbb{E}[|\zeta_x^N|],\\
	A_\delta^N:\mathcal{X}^N\to\mathbb{R}&\quad x\mapsto\lambda_x^N\cdot\mathbb{E}[|\zeta_x^N\cdot\mathds{1}_{\{|\zeta_x^N|>\delta\}}|].
\end{align*}

\begin{remark} Note that:
	\begin{itemize}
		\item Function $V^N(x)$ represents the expected increment per time unit from $x$ under the process $\{X^N_t\}$, since it is the product of the jump rate at state $x$ and the expected increment per jump at $x$,
		\item Similarly, $A^N(x)$ is the expected absolute displacement per time unit,
		\item And $A^N_\delta(x)$ is the expected absolute displacement per time unit due to jumps traveling further than $\delta$.
	\end{itemize}
\end{remark}

On the following theorem is based Theorem \ref{Theorem 4.1}.

\begin{theorem}[Kurtz]\label{Kurtz}
	Let $V: X\to TX$ be a Lipschitz continuous vector field. Suppose that there exists $\{\delta^N\}_{N\ge N_0}$ converging to $0$ such that:\begin{itemize}
		\item[(i)] $\lim_{N\to\infty}\sup_{x\in\mathcal{X}^N}|V^N(x)-V(x)|=0$,
		\item[(ii)] $\sup_N\sup_{x\in\mathcal{X}^N}A^N(x)<\infty$,
		\item[(iii)] $\lim_{N\to\infty}\sup_{x\in\mathcal{X}^N}A^N_{\delta^N}(x)=0$.
	\end{itemize} If the initial conditions $X_0^N=x_0^N$ converge to $x_0\in X$, let $\{x_t\}_{t\ge0}$ be the solution to\begin{equation}
		\begin{cases}
		\dot{x}=V(x)\\ x(0)=x_0
		\end{cases}.
	\label{M}
	\end{equation} Then:\begin{equation}
		\lim_{N\to\infty}\mathbb{P}\left(\sup_{t\in[0,T]}|X^N_t-x_t|<\epsilon\right)=1\quad \mbox{for all}\ T<\infty,\ \epsilon>0
	\end{equation}
	\label{Kurtz's equation}
	\begin{proof}
		See \cite[Theorem~10.2.1]{Kurtz}.
	\end{proof}
\end{theorem}

Fix a finite time horizon $T<\infty$ and an error bound $\epsilon>0$, Kurtz's Theorem states that as $N\to\infty$, nearly all simple paths of $\{X_t^N\}_{t\ge0}$ stay within $\epsilon$ of a solution of \eqref{M} through time $T$. Thus, if $N$ is large enough, $X_t^N\approx x_t$.

Condition (i) demands that as $N\to\infty$, the expected displacements per time unit $V^N$ converge uniformly to a Lipschitz continuous vector field $V$. Remark that the Lipschitz continuity of $V$ ensures the existence and uniqueness of solutions to \eqref{M}.

Condition (ii) requires that the expected absolute displacement per time unit is bounded.

Finally, condition (iii) demands tha jumps larger then $\delta^N$ make vanishing contributions to the motion of the processes, where $\{\delta^N\}_{N=N_0}^\infty$ approaches zero.

The reasoning behind Kurtz's theorem can be described as follows. At each revision opportunity, the change in the process $\{X_t^N\}$ is stochastic. However, the expected number of revision opportunities that occur during a small time interval $I = [t,\ t + \,\mathrm{d}t]$ is of the order of $\lambda_{\chi}^N \,\mathrm{d}t$. Whenever this quantity is non–zero, it grows without bound as the population size $N$ increases. Conditions (ii) and (iii) guarantee that when $N$ is large, each change in the state is likely to be small. This ensures that the total variation in the state during the interval $I$ remains small, meaning that jump rates and transition probabilities do not change significantly during this period.

During interval $I$, there are numerous revision opportunities, each producing nearly the same expected change. According to the law of large numbers, the overall change in $\{X_t^N\}$ over this interval is predominantly determined by the expected motion of $\{X_t^N\}$. This expected behavior is captured by the limiting mean dynamic $V$, whose solutions approximate the stochastic process $\{X_t^N\}$ over finite periods with high probability.

\chapter{Mathematical Framework} \label{Mathematical Framework}

This chapter introduces the mathematical framework for stochastic processes with jumps. We begin with a discussion of the choice of topology, highlighting the significance of selecting an appropriate topology to establish convergence results. Indeed, we compare the uniform and the Skorokhod topologies, emphasizing the role each topology plays in modelling stochastic processes. We also write and examine the Q–matrix generating the process, which captures the dynamics of the system and describes the state transitions and policy revisions of the players. The section further elaborates on the mathematical constructs, such as the state transition and policy transition Q–matrices, which are pivotal in understanding the evolution of strategies within a population game.

\section{Topology Selection: Uniform vs Skorokhod} \label{Choice of the Topology}

When analyzing stochastic processes with jumps, the choice of topology plays a crucial role in establishing convergence results. A common approach is to work with the Skorokhod topology, which is well–suited for studying processes that exhibit discontinuities. This topology provides flexibility in handling jumps by allowing small perturbations in time to compensate for abrupt changes in state, making it a natural choice in many asymptotic analyses.

In our setting, we instead focus on the asymptotic behavior of a finite–population model as $N \to \infty$ , where the empirical state–action distribution converges to a deterministic mean–field limit governed by an ODE. To rigorously establish this convergence, we leverage Kurtz’s Theorem, which provides strong approximation guarantees for stochastic processes in large–population limits. Given this setting, the uniform topology (i.e., the sup–norm distance) is sufficient to quantify convergence, as it directly measures the maximum deviation between the empirical and mean–field distributions at any fixed time.

Unlike the Skorokhod topology, which is particularly useful when the limiting process retains jumps, our mean–field limit follows a deterministic trajectory, making the stronger uniform norm appropriate. The results in Lemma \ref{Lemma 2.1} confirm that as $N \to \infty$, the empirical joint state–action and state–policy distributions almost surely converge to their deterministic counterparts. The application of Kurtz’s Theorem ensures that the fluctuations in the finite–population dynamics diminish, allowing us to work within the sup–norm framework without loss of generality.

Thus, while the Skorokhod topology remains a powerful tool in stochastic process theory, the uniform sup–norm metric is sufficient—and preferable—for our purposes, as it aligns naturally with the deterministic limit obtained in the mean–field setting.

\subsection{Sup–Norm Distance ($L^\infty$ Norm)}

\begin{definition}[Sup–Norm]
	For a function $f: [0,T]\rightarrow\mathbb{R}^d$, the \textit{sup–norm} (uniform norm) is defined as\begin{equation}
		\norm{f}_\infty:=\sup_{t\in[0,T]}\left|f(t)\right|,
		\label{Sup-norm}
	\end{equation} where $\left|\cdot\right|$ is the Euclidean norm in $\mathbb{R}^d$.
\end{definition}

\subsubsection{Advantages}

\begin{itemize}
	\item[\color{green}$\checkmark$] Strong form of convergence: If $\norm{f_N-f}_\infty\to0$, then $f_N$ converges uniformly to $f$, meaning that the maximum deviation over time vanishes.
	\item[\color{green}$\checkmark$] Simplicity: Works well for deterministic dynamics (ODE models) and ensures well–posedness.
	\item[\color{green}$\checkmark$] Directly applies to the stability of ODEs: The sup–norm is useful in showing well–posedness, uniqueness, and stability of solutions in compact time intervals.
\end{itemize}

\subsubsection{Limitations}

\begin{itemize}
	\item[\color{red}$\times$] Insensitive to small–time oscillations: The sup–norm does not capture cases where a trajectory jumps at different times but still stays close overall.
	\item[\color{red}$\times$] Not well–suited for discontinuous sample paths: If state–action distributions $\mu_N$ exhibit jumps due to Poisson arrivals or discrete updates, the sup–norm may not effectively capture convergence.
\end{itemize}

\subsection{Skorokhod Topology}

General reference: \cite[Chapter~A]{Shwartz}.

When dealing with stochastic processes, the sup–norm distance may not be an appropriate metric: even a slight perturbation in the time of a jump can lead to a significant discrepancy under the sup–norm. To accomodate such small fluctuations, the Skorokhod topology is widely used.

\subsubsection{The Space of \textit{càdlàg} Functions}

\begin{definition}[Space of \textit{càdlàg} Functions]
	Let $D^d[0,T]$ denote the space of \textit{càdlàg} (right–continuous with left limits) functions defined on $t\in[0,T]$.
\end{definition}

Consider the set $\Lambda$ of strictly increasing, continuous $\lambda: [0,T]\rightarrow\mathbb{R}$ satisfying:\begin{itemize}
	\item $\lambda(0)=0,\ \lambda(T)=T$,
	\item $\gamma(\lambda):=\sup_{0\le s\le t\le T}\left|\log\frac{\lambda(s)-\lambda(t)}{s-t}\right|<\infty$.
\end{itemize}

\begin{definition}[Standard Metric on the Space of \textit{càdlàg} Functions]
	The \textit{standard metric} on $D^d[0,T]$ is given by:\begin{equation*}
		d_d(\textbf{x},\textbf{y}):=\inf_{\lambda\in\Lambda}\left\lbrace \max\left\lbrace\gamma(\lambda),\ \sup_{t\in[0,T]}|\textbf{x}(t)-\textbf{y}(\lambda(t))|\right\rbrace\right\rbrace.
	\end{equation*}
\end{definition}

This distance induces the Skorokhod topology on $D^d[0,T]$, which is complete and separable (and obviously metric), thus Polish.

\subsubsection{Convergence in the \textit{càdlàg} functions space}

In $D^d[0,T]$, the following are equivalent:\begin{itemize}
	\item The sequence $(\textbf{x}_n)_{n\in\mathbb{N}}$ converges to $\textbf{x}$,
	\item There exists a sequence $(\lambda_n)_{n\in\mathbb{N}}\subseteq\Lambda$ such that both hold:\begin{itemize}
		\item $\gamma(\lambda_n)\underset{n\rightarrow+\infty}{\longrightarrow}0$,
		\item $\sup_{t\in[0,T]}|\textbf{x}(t)-\textbf{x}_n(\lambda_n(t))|\underset{n\rightarrow+\infty}{\longrightarrow}0$.
	\end{itemize}
	\item Alternatively, there exists a sequence $(\lambda_n)_{n\in\mathbb{N}}\subseteq\Lambda$ such that both hold:\begin{itemize}
		\item $\sup_{t\in[0,T]}|\lambda_n(t)-t|\underset{n\rightarrow+\infty}{\longrightarrow}0$,
		\item $\sup_{t\in[0,T]}|\textbf{x}(t)-\textbf{x}_n(\lambda_n(t))|\underset{n\rightarrow+\infty}{\longrightarrow}0$.
	\end{itemize}
\end{itemize}

\subsubsection{Sup–Norm Distance in $D_d[0,T]$}

Let $d_c$ denote the sup–norm distance on $D^d[0,T]$, as defined in \eqref{Sup-norm}. Then, the metric space $(D^d[0,T],\ d_c)$ has a stronger topology than $d_d$, meaning every open (respectively closed) set in $d_d$ is also open (respectively closed) in $d_c$. However, while $(D^d,\ d_c)$ is complete, it is not separable.

\subsubsection{Advantages}

\begin{itemize}
	\item[\color{green}$\checkmark$] Captures jumps in state–action distributions: Unlike the sup–norm, it allows for small shifts in time, making it more appropriate for jump processes.
	\item[\color{green}$\checkmark$] Stronger notion of convergence for stochastic processes: If $\mu_N$ is approximating a random process, the Skorokhod topology ensures that discontinuities align correctly in the limit.
	\item[\color{green}$\checkmark$] More flexible than sup–norm for stochastic systems: Useful when state–action distributions are driven by discrete events (e.g., Poisson updates).
\end{itemize}

\subsubsection{Limitations}

\begin{itemize}
	\item[\color{red}$\times$] More complex: Requires dealing with homeomorphisms of time, making analysis harder than with the sup–norm.
	\item[\color{red}$\times$] Not always necessary for ODE–based models: If the model follows deterministic ODEs, as in our case, the sup–norm is sufficient.
\end{itemize}

\subsection{A Comparison Between Sup–Norm Distance and Skorokhod Topology}

In Table \ref{SupNorm vs Skorokhod} we make a summarized comparison of the properties of interest of the two topologies above.

\begin{table}[H]
	\centering
	\caption{Comparison between Sup–Norm Distance and Skorokhod Topology}
	\begin{tabular}{l||c|c}
		\textbf{Property$\mathbf{\backslash}$Topology} & \textbf{Sup–Norm} & \textbf{Skorokhod Topology} \\
		\hline\hline
		Best for ODEs? & \color{green}$\checkmark$\color{black}\ Yes & \color{red}$\times$\color{black}\ No need \\
		Best for jump processes? &\color{red}$\times$\color{black}\ No & \color{green}$\checkmark$\color{black}\ Yes \\
		Handles time shifts? & \color{red}$\times$\color{black}\ No & \color{green}$\checkmark$\color{black}\ Yes \\
		Easier to use? & \color{green}$\checkmark$\color{black}\ Yes & \color{red}$\times$\color{black}\ More complex \\
		Supports Lipschitz continuity results? & \color{green}$\checkmark$\color{black}\ Yes & \color{red}$\times$\color{black}\ Harder \\
	\end{tabular}
	\label{SupNorm vs Skorokhod}
\end{table}

\subsection{Which Topology Best Fits Our Model}

Given that fluctuations in our model are ruled by the $\frac{1}{N}$ term in \eqref{StatePolicy} and \eqref{StatePolicy Evolutionary Dynamics}, meaning they vanish as $N \to \infty$, they do not significantly impact the limiting behavior of the system. This has direct implications for the choice of topology:\begin{itemize}
	\item Fluctuations are negligible in the limits thanks to the $\frac{1}{N}$ terms, so the Skorokhod topology is unnecessary, as it is primarily useful for handling discontinuities and large fluctuations.
	
	\item Mean–field model follows a deterministic ODE:
	\begin{itemize}
		\item The ODE representation in \eqref{ODE8} shows that the state–policy distribution evolves continuously in time.
		\item Since we want to prove Lipschitz–continuity of the solutions, the sup–norm topology naturally aligns with stability analysis and equilibrium approximation.
	\end{itemize}
	
	\item Convergence proofs support sup–norm topology:
	\begin{itemize}
		\item Lipschitz continuity of the solutions.
		\item The Picard–Lindel\"of theorem guarantees uniqueness and well–posedness in the sup–norm.
		\item Gr\"onwall’s inequality supports sup–norm–based error bounds.
		\item Kurtz’s Theorem establishes the convergence of empirical processes to their deterministic limit both in sup–norm distance and Skorokhod topology.
	\end{itemize}
\end{itemize}

Thus, the sup–norm topology ($L^\infty$) is the best choice for our model:\begin{itemize}
	\item It aligns with the deterministic mean–field ODE framework.
	\item It simplifies stability analysis and ensures uniform convergence as $N \to \infty$.
	\item Skorokhod topology is unnecessary, since fluctuations are bounded, meaning they do not introduce significant jumps or irregularities.
\end{itemize}

\section{The Q–Matrices Involved in the Processes}

\begin{recall}
	Let $I$ be a countable index set. A Q–matrix on $I$ is a matrix $Q=(q_{i,j})_{i,j\in I}$ satisfying:\begin{itemize}
		\item $q_{i,i}\in(-\infty,0]\ \forall i\in I$,
		\item $q_{i,j}\ge0\ \forall i\ne j$,
		\item $\sum_{j\in I}q_{i,j}=0\ \forall i\in I$.
	\end{itemize}
\end{recall}

Define $q_i:=\sum_{j\ne i}q_{i,j}=-q_{i,i}$, we give to the entries of $Q$ the following meaning:\begin{itemize}
	\item $q_{i,j}$ represents the rate of transition from state $i$ to $j$,
	\item $q_i$ denotes the total rate of leaving state $i$.
\end{itemize}

The model under consideration is characterized by transitions between joint state–policy distributions, driven by individual players’ decisions. To capture these dynamics, we define:\begin{itemize}
	\item A state transition Q–matrix governing the evolution of player states, useful when studying situations in which policy transitions are not allowed,
	\item A collection of $|\mathcal{S}|$ policy transition Q–matrices, one for each state, governing changes in players' policies,
	\item Then, we merge them together to obtain the joint state–policy transition Q–matrix, which we eventually use in \ref{Theorem 4.1}.
\end{itemize}

\subsubsection{State Transition Q–Matrix}

\begin{recall}
	Each player is equipped with a Poisson alarm clock of rate $R_d>0$. Every time a player's clock rings, the player takes an action. The clocks of different players are independent and identically distributed with rate $R_d$.
\end{recall}

\begin{recall}
	Upon taking an action, a player's state evolves according to a Markov transition kernel $\phi: \mathcal{S}\times \mathcal{A}\rightarrow\mathcal{P}(\mathcal{S})$, where:\begin{itemize}
		\item $\mathcal{S}$ is the set of states,
		\item $\mathcal{A}$ is the set of actions,
		\item $\mathcal{P}(\mathcal{S})$ denotes the space of Borel probability measures on $\mathcal{S}$.
	\end{itemize} We denote by $\phi(\cdot\mid s,a)$ the distribution of the new state of a player in state $s\in \mathcal{S}$ taking action $a\in \mathcal{A}(s)$, where $\mathcal{A}(s)$ is the set of all the available actions to a player in state $s\in \mathcal{S}$.
\end{recall}

\begin{recall}
	A policy is a map $u: \mathcal{S}\rightarrow\mathcal{P}(\mathcal{A})$, specifying a probability distribution over actions in each state. The probability of selecting action $a$ in state $s$ under policy $u$ is written as $u(a\mid s)$.
\end{recall}

\begin{definition}[State Transition Q–Matrix]\label{StateQ}
	The \textit{state transition Q–matrix} $Q_{\mathcal{S}}=(q_{s,s'})_{s,s'\in \mathcal{S}}$ is defined by:\begin{itemize}
		\item $q_{s,s'}:=R_d\cdot\sum_{a\in \mathcal{A}(s)}u(a\mid s)\cdot\phi(s'\mid s,a),\ \forall s\ne s'$,
		\item $q_{s,s}:=-\sum_{s'\ne s}q_{s,s'},\ \forall s\in \mathcal{S}$.
	\end{itemize} Here:\begin{itemize}
		\item $q_{s,s'}$ represents the transition rate from state $s$ to state $s'$,
		\item $-q_{s,s}$ represents the total departure rate from state $s$.
	\end{itemize}
\end{definition}

By construction, $Q_{\mathcal{S}}$ satisfies the conditions of a Q–matrix.

\subsubsection{Policy Transition Q–Matrices}

\begin{recall}
	Each player is equipped with a Poisson alarm clock of rate $R_r>0$. Every time a player's clock rings, they have the opportunity to revise the policy they are currently playing. These revision times are independent across players and independent of the state–transition process.
\end{recall}

\begin{recall}
	A player in state $s\in \mathcal{S}$ playing policy $u\in\mathcal{U}_D$ (where $\mathcal{U}_D$ is the space of deterministic policies, i.e., the policies that map each state to a single action) switches to policy $v\in\mathcal{U}_D$ at a rate given by:\begin{equation*}
		\rho_{u,v}(F^s(\hat{\mu}),\ \hat{\mu}[s,\cdot]).
	\end{equation*}
\end{recall}

\begin{definition}[Policy Transition Q–Matrix]\label{PolicyQ}
	For each state $s\in \mathcal{S}$, the \textit{policy transition Q–matrix} $Q_{\mathcal{U}_D}(s)=(q_{u,v}(s))_{u,v\in\mathcal{U}_D}$ is defined by:\begin{itemize}
		\item $q_{u,v}(s):=R_r\cdot\rho_{u,v}(F^s(\hat{\mu}),\ \hat{\mu}[s,\cdot]),\ \forall u\ne v$,
		\item $q_{u,u}(s):=-\sum_{v\ne u}q_{u,v},\ \forall u\in\mathcal{U}_D$.
	\end{itemize} Here:\begin{itemize}
		\item $q_{u,v}(s)$ represents the rate at which a player in state $s$ switches from policy $u$ to policy $v$,
		\item $-q_{u,u}(s)$ represents the total departure rate from policy $u$ in state $s$.
	\end{itemize}
\end{definition}

By construction, the matrix $Q_{\mathcal{U}_D}(s)$ is a Q–matrix for each state $s\in \mathcal{S}$.

\subsubsection{Joint State–Policy Transition Q–Matrix}

The Q–matrix for the joint state–policy distribution should combine both the state transition Q–matrix (Definition \ref{StateQ}) and the policy transition Q–matrices (Definition \ref{PolicyQ}). The structure must respect the fact that:\begin{itemize}
	\item State transitions occur according to the Markovian dynamics governed by the state transition Q–matrix.
	\item Policy transitions occur independently within each state according to the policy transition Q–matrices.
\end{itemize}

\begin{definition}[Joint State–Policy Transition Q–Matrix]\label{Joint Q-matrix}
	Let $Q_{(s,u),(s',u')}$ be the transition rate from state–policy pair $(s,u)$ to $(s',u')$. The \textit{joint state policy transition Q–matrix} is given by:\begin{equation*}
		Q_{(s,u), (s',u')}:=\begin{cases}
			R_d\cdot \sum_{a\in\mathcal{A}(s)}u(a\mid s)\cdot\phi(s'\mid s,a) & \mbox{if}\ s\ne s',\ u=u'\\
			\rho_{u,u'}(F^s(\hat{\mu}),\ \hat{\mu}[s,\cdot]) & \mbox{if}\ s=s',\ u\ne u'\\
			0 & \mbox{if}\ s\ne s',\ u\ne u'\\
			-\sum_{(s',u')\ne(s,u)}Q_{(s,u),(s',u')}&\mbox{if}\ s=s',\ u=u'
		\end{cases}.
	\end{equation*} Here:\begin{itemize}
		\item The first case represents state transitions: A player in state $s$ playing policy  $u$ transits to state $s'$ based on the state transition Q–matrix $Q_{\mathcal{S}}$ (Definition \ref{StateQ}), while their policy remains unchanged.
		\item The second case represents policy transitions: A player in state $s$ switches from policy $u$ to policy $u'$ based on the policy transition Q–matrix $Q_{\mathcal{U}_D}$ (Definition \ref{PolicyQ}), while their state remains unchanged.
		\item The third case represents both state and policy transitions: A player cannot change both state and policy at the same time because the state transition process and policy revision process are independent. Instead, these two types of transitions happen sequentially but never at the exact same instant.
		\item The fourth case ensures that the rows of the Q–matrix sum to zero, a necessary property for a valid generator matrix.
	\end{itemize}
\end{definition}

\begin{recall}
	Policy transitions are allowed only for deterministic policies. Therefore, the joint state–policy transition Q–matrix defined above is a square matrix of dimension $|\mathcal{S}\times\mathcal{U}_D|$.
\end{recall}

\section{Notions of Convergence for Random Variables}

Convergence of random variables is a central concept in probability theory and statistics. It describes how a sequence of random variables $(X_n)_{n\in\mathbb{N}}$ behaves as $n\to\infty$. There are several distinct notions of convergence, each with its own properties and implications. In this section, we explore the four most common modes of convergence for random variables: in probability, almost surely, in distribution (or in law), and in $L^p$, highlighting their definitions, differences, and relationships. At the end of this section, we state and prove the Continuous Mapping Theorem, which plays a crucial role in understanding the behavior of functions of converging random variables.

\subsection{Convergence in Probability}

\begin{definition}[Convergence in Probability]
	A sequence of random variables $(X_n)_{n\in\mathbb{N}}$ converges to a random variable $X$ \textit{in probability}, and we write $X_n\xrightarrow{\mathbb{P}}X$, if\[
	 \lim_{n \to \infty} \mathbb{P}(|X_n - X| \geq \epsilon) = 0\quad \text{ for all } \epsilon > 0.
	\]
\end{definition}

This type of convergence means that for large enough $n$, the random variables $X_n$ are—loosely speaking—close to $X$ with high probability, though the closeness may not occur for every outcome.

\subsection{Almost Sure Convergence}

\begin{definition}[Almost Sure Convergence]
	A sequence $(X_n)_{n\in\mathbb{N}}$ converges to $X$ \textit{almost surely}, and we write $X_n\xrightarrow{\text{a.s.}}X$, if:\[\mathbb{P}\left( \lim_{n \to \infty} X_n = X \right) = 1.\]
\end{definition}

Almost sure convergence is a stronger condition than convergence in probability. It means that the sequence $X_n$ eventually converges to X for nearly every possible outcome in the sample space, i.e., for all outcomes except for a set of outcomes with probability zero.

\subsection{Convergence in Distribution, or in Law}

\begin{definition}[Convergence in Distribution, or in Law]
	A sequence of random variables $(X_n:\ (\Omega_n,\mathcal{F}_n,\mathbb{P}_n)\to\mathbb{R}^d)_{n\in\mathbb{N}}$ converges \textit{in distribution} (or in law) to the random variable $X:\ (\Omega,\mathcal{F},\mathbb{P})\to\mathbb{R}^d$, and we write \( X_n \xrightarrow{(d)} X \), if \[
		\mathbb{E}_{\mathbb{P}_n}[g(X_n)]\to\mathbb{E}_{\mathbb{P}}[g(X)]\quad\text{for all}\ g\in C_b(\mathbb{R}^d,\mathbb{R}),
	\] where $C_b(\mathbb{R}^d,\mathbb{R})$ denotes the set of continuous and bounded functions from $\mathbb{R}^d$ to $\mathbb{R}$.
\end{definition}

Convergence in distribution concerns the convergence of the distributions (or probabilities) of the random variables, rather than their actual values. It is the weakest form of convergence and does not require the random variables to converge pointwise in any specific way.

\subsection{Convergence in $L^p$, for $p\ge1$}

\begin{definition}[Convergence in $L^p$, for $p\ge1$]
	A sequence of random variables $(X_n)_{n\in\mathbb{N}}$ converges to $X$ in $L^p$, and we write $X_n \xrightarrow{L^p} X$, if the $p$–th power of the absolute difference between $X_n$ and $X$ has expected value that goes to zero. Specifically, for $p \ge 1$, this means:\[\lim_{n \to \infty} \mathbb{E}[|X_n - X|^p] = 0.\]
\end{definition}

Convergence in $L^p$ implies that the expected size of the difference between $X_n$ and $X$, raised to the power of $p$, becomes arbitrarily small as $n$ grows. This form of convergence provides a way to control the magnitude of deviations between $X_n$ and $X$, particularly in terms of their higher moments.

\subsection{Implications Between the Notions of Convergence}

The theorem below highlights the relationships between the four notions of convergence defined above, which are resumed as follows.

\begin{table}[H]
	\centering
	\begin{tabular}{lcccc}
		a.s. & \rotatebox{340}{$\Rightarrow$} & & & \\
		& & $\mathbb{P}$ & $\Rightarrow$ & $(d)$ \\
		$L^p$ & \rotatebox{30}{$\Rightarrow$} & & &
	\end{tabular}
\end{table}

Simple counterexamples can be constructed to demonstrate that the reverse implications are false, and also to establish that no implications hold between almost sure convergence and convergence in $L^p$.

\begin{theorem}
	The following hold:\begin{itemize}
		\item[(i)] If $X_n\xrightarrow{\text{a.s.}}X$, then $X_n\xrightarrow{\mathbb{P}}X$,
		\item[(ii)] If $X_n\xrightarrow{L^p}X$, then $X_n\xrightarrow{\mathbb{P}}X$,
		\item[(iii)] If $X_n\xrightarrow{\mathbb{P}}X$, then $X_n\xrightarrow{(d)}X$.
	\end{itemize}
	\begin{proof}
		\begin{itemize}
			\item[(i)] Let $\epsilon>0$, by Fatou's Lemma \cite[Theorem~1.5.5]{Lemmas} we have that\[\limsup_{n\to\infty}\mathbb{P}(|X_n-X|>\epsilon)\le\mathbb{P}\left(\limsup_{n\to\infty}\{|X_n-X|>\epsilon\}\right)=0,\] where the equality follows from the definition of almost sure convergence.
			\item[(ii)] Let $\epsilon>0$, by Markov Inequality we have that\[\mathbb{P}(|X_n-X|\ge\epsilon)\le\frac{1}{\epsilon^p}\mathbb{E}[|X_n-X|^p]\xrightarrow[n\to\infty]{}0,\] where the convergence follows from the definition of convergence in $L^p$.
			\item[(iii)] By Borel–Cantelli's Lemma \cite[Theorem~2.3.1]{Lemmas} it is possible to demonstrate that if $X_n\xrightarrow{\mathbb{P}}X$, then there exists a subsequence $(X_{n_k})\subseteq(X_n)$ such that $X_{n_k}\xrightarrow{\text{a.s.}}X_n$. We conclude the proof by applying the Dominated Convergence Theorem \cite[Theorem~9.1.2]{Rosenthal}.
		\end{itemize}
	\end{proof}
\end{theorem}

\subsection{Continuous Mapping Theorem}

The notions of convergence above form the foundation for understanding the limiting behavior of random variables and are essential tools in probability theory. The last result of this section is the Continuous Mapping Theorem (also known as Mann–Wald's Theorem), which describes how the convergence of random variables is preserved under continuous functions. The theorem provides a powerful tool for analyzing the convergence of transformations of converging random variables.

\begin{theorem}[Continuous Mapping Theorem] \label{CMP}
	Let $(X_n)_{n\in\mathbb{N}},\ X$ be random variables defined on a metric space $S$, and let $S'$ be another metric space. Consider a continuous function $g: S\to S'$ (it is sufficient for $g$ to be continuous everywhere except fot a set $D_g$ of null measure). The following hold:\begin{itemize}
		\item[(i)] If $X_n\xrightarrow{(d)}X$, then $g(X_n)\xrightarrow{(d)}g(X)$,
		\item[(ii)] If $X_n\xrightarrow{\mathbb{P}}X$, then $g(X_n)\xrightarrow{\mathbb{P}}g(X)$,
		\item[(iii)] If $X_n\xrightarrow{\text{a.s.}}X$, then $g(X_n)\xrightarrow{\text{a.s.}}g(X)$.
	\end{itemize}
	\begin{proof}
		See \cite[Theorem~5]{CMP}.
	\end{proof}
\end{theorem}

\chapter{Theoretical Results} \label{Theoretical Results}

\section{Ergodic Theorem in our Framework} \label{Ergodic Theorem Model}

The Ergodic Theorem we stated in Section \ref{ET} requires the Markov chain to be irreducible. However, in our model irreducibility is too strong. Instead, in the following assumption we enforce that a unique recurrent communicating class exists, while all other states are transient.

\begin{Assumption}\label{A3}
	For each deterministic policy $u\in\mathcal{U}_D$, the Markov kernel $\phi^u$ on $S$ defined by \[\phi^u(s\mid s'):=\sum_{a'\in\mathcal{A}(s')}\phi(s\mid s',a')\cdot u(a'\mid s')\] has exactly one recurrent communicating class, and all other states are transient.
\end{Assumption}

This section is then dedicated to proving that the Ergodic Theorem still holds under Assumption \ref{A3}. There are two possible approaches:\begin{itemize}
	\item Prove the result in a slightly different but more general framework, with possibly infinitely many states and assuming the unique recurrent class is positive–recurrent\footnote{If the unique recurrent class is null–recurrent, then the theorem does not hold, so the statement is strictly more general than the one needed.} and hit almost surely, and then using the fact that a finite recurrent class is positive recurrent (see Definition \ref{Positive and Null Recurrence}),
	\item Prove the result in our framework, which includes assuming a finite state space, hence a unique positive–recurrent communicating class.
\end{itemize} We will present both approaches for completeness.

\subsection{Ergodic Theorem for a Markov Chain with Unique Positive Recurrent Class}

To commence with, we state and prove the Ergodic Theorem for a Markov chain defined on a general state space and having exactly one positive–recurrent communicating class, which is entered almost surely, while all other states are transient. This result is strictly more general than the Ergodic Theorem we need in our model, in which the state space is finite, implying the unique recurrent class—enforced in Assumption~\ref{A3}—to be positive–recurrent.

\begin{theorem}[Ergodic Theorem for a Markov Chain with Unique Positive–Recurrent Class] \label{Ergodic Theorem Positive Recurrent}
	Let $(X_t)_{t\ge0}$ be a continuous–time Markov chain with generator matrix $Q$ on a countable state space $I$. Suppose:\begin{itemize}
		\item[(i)] There is exactly one communicating class $C\subseteq I$ that is positive–recurrent,
		\item[(ii)] Every state $i\in I\setminus C$ is transient and $\mathbb{P}_i\bigl(T_C<\infty\bigr)=1$, where $T_C$ is the first passage time in $C$, i.e., $T_C=\inf\{t\ge0:\ X_t\in C\}$,
		\item[(iii)] Let $\lambda$ be the unique invariant distribution supported on the positive–recurrent communicating class $C$ (see the remark below).
	\end{itemize} Then, for any initial distribution $\nu$ on $I$ and for each $j\in I$,\[
		\frac{1}{t}\int_{0}^{t} \mathds{1}_{\{X_s=j\}}\,\mathrm{d}s\xrightarrow[t\to\infty]{\mbox{a.s.}}		\begin{cases}
		\lambda_j=\frac{1}{m_jq_j} & \mbox{if}\ j\in C\\
		0 &\mbox{if}\ j\notin C
		\end{cases},
	\] where $m_j=\mathbb{E}_j[T_j]$ is the expected return time to state $j\in C$ and $q_j=-Q_{j,j}$. Furthermore, for any bounded function $f: I\to\mathbb{R}$,\[
		\frac{1}{t} \int_{0}^{t} f(X_s) \,\mathrm{d}s \xrightarrow[t\to\infty]{\mbox{a.s.}} \sum_{i\in C} \lambda_i f(i).
	\]
	\begin{proof}
		The core of our proof is showing that the proportion of time spent outside $C$ goes to $0$, then restricting to the unique positive‐recurrent class $C$ and applying (actually, going over the proof of) the standard irreducible‐class argument, i.e., the Ergodic Theorem \ref{ET}.
		
		By hypothesis all states in $I\setminus $C are transient and with probability $1$ the chain visits $I\setminus C$ only finitely many times (or, equivalently, that the holding times outside $C$ have finite total sum almost surely). Concretely:\begin{itemize}
			\item Almost surely\[\int_0^\infty\mathds{1}_{\{X_s\in I\setminus C\}}\,\mathrm{d}s<\infty,\]
			\item Consequently, for large $t$, the chain spends negligible time outside $C$. Formally, \[
			\frac{1}{t} \int_{0}^{t}\mathds{1}_{\{X_s\in I\setminus C\}}\,\mathrm{d}s
			\xrightarrow[t\to\infty]{\mbox{a.s.}} 0,
			\] showing that the proportion of time outside $C$ goes to $0$ almost surely.
		\end{itemize}
		
		Once the chain is in $C$, it behaves like an irreducible, positive‐recurrent continuous‐time chain restricted to $C$.  Indeed, we can consider the trace of $X_t$ on $C$, i.e., look only at the times and states in $C$.  Because $C$ is irreducible and positive‐recurrent, there is a return‐time argument exactly as in the usual Ergodic Theorem proof: for each $i\in C$, denote \[m_i=\mathbb{E}_i[T_i],\quad \lambda_i=\frac{1}{m_iq_i},\] where $T_i$ is the first return time to $i$ starting from $i$, and $q_i=-Q_{i,i}$ is the holding rate at $i$. It holds that $(\lambda_i:\ i\in C)$ is the unique invariant distribution on $C$:\begin{itemize}
			\item By the strong Markov property and the law of large numbers \cite[Theorem~3.8.1]{Norris}, for any fixed $i\in C$,\[
			\frac{1}{t}\int_{0}^{t}\mathds{1}_{\{X_s=i\}}\,\mathrm{d}s\xrightarrow[t\to\infty]{\mbox{a.s.}} \lambda_i,
			\]
			\item Because the chain visits $C$ infinitely often with probability $1$ (once it reaches $C$, it is a positive‐recurrent class and cannot escape to a transient set forever, by definition), we obtain that the limit holds almost surely whenever the chain stops staying in a transient set, i.e., enters the recurrent class.
		\end{itemize}
	
		The first convergence of the theorem follows by putting the two steps above together:\begin{itemize}
			\item Almost surely, the chain spends vanishingly little time in $I\setminus C$ as $t\to\infty$
			\item Inside $C$, the standard Ergodic Theorem for irreducible positive‐recurrent chains applies and yields\begin{equation}
				\frac{1}{t}\int_{0}^{t}\mathds{1}_{\{X_s=i\}}\,\mathrm{d}s\xrightarrow[t\to\infty]{\mbox{a.s.}} \lambda_i.
				\label{ET 1}
			\end{equation}
		\end{itemize} Hence, for each $i\in C$, the proportion of time spent at $i$ converges to $\lambda_i$, and for $i\notin C$ it converges to $0$.
	
		By linearity, applying the same argument to a bounded function $f$ instead of the indicator function, we get\begin{equation}
			\frac{1}{t}\int_{0}^{t} f(X_s)\,\mathrm{d}s \xrightarrow[t\to\infty]{\mbox{a.s.}}\sum_{i\in C} \lambda_if(i),
			\label{ET 2}
		\end{equation}
		completing the proof (see Remark \ref{Linearity}).
	\end{proof}
	\begin{note}
		This is precisely the same conclusion as in the purely irreducible case, except that the limit distribution is supported on the single recurrent class $C$. All points outside $C$ receive limit‐weight $0$, reflecting the transience of those states.
	\end{note}
	\begin{remark}[Existence and Uniqueness of an Invariant Distribution for a Markov Chain with Exactly One Positive–Recurrent Class] \label{Invariant Distribution Ergodic}
		Hypothesis (iii) of the theorem above claims that in the framework of the theorem there exists a unique invariant distribution supported on the positive–recurrent communicating class $C$, which is unique for our Markov chain. Here is why this statement holds:\begin{itemize}
			\item The restriction of the generator matrix $Q$ to class $C$ gives an irreducible positive–recurrent Q–matrix. By Theorem \ref{Invariant distributions and positive recurrence}, such a matrix has a unique invariant distribution $\lambda=(\lambda_i:\ i\in C)$, given by\[\lambda_i=\frac{1}{m_iq_i}\quad\text{for all}\ i\in C.\] Extending the invariant distribution $\lambda$ to the entire state space $I$ by placing null mass on each state $i\in I\setminus C$ results in an invariant distribution (since each state $i\in I\setminus C$ is transient, then the mass placed on $i$ by any invariant distribution must be null).
			\item Uniqueness follows directly by the uniqueness of $\lambda$ over $C$: let $\nu$ be an invariant distribution for our process, then $\nu(i)=0$ for all $i\in I\setminus C$, which is a set of transient states. Thus, $\nu$ is supported on $C$ and since the restriction of $Q$ to $C$ has a unique invariant distribution, which is $\lambda$, then $\nu=\lambda$.
		\end{itemize}
	\end{remark}
	\begin{remark}[From the Convergence of Indicators to that of Bounded Functions] \label{Linearity}
		Partition the time interval $[0,t]$ through the disjoint sets $E_i(t):=\{s\in[0,t]:\ X_s=i\}$, which satisfy\[|E_i(t)|=\int_0^t\mathds{1}_{\{X_s=i\}}\mathrm{d}s.\] Then,\[\int_0^tf(X_s)\mathrm{d}s=\sum_{i\in I}\int_{E_i(t)}f(X_s)\mathrm{d}s=\sum_{i\in I}f(i)\cdot|E_i(t)|.\] Since $0\le\frac{1}{t}\int_0^t\mathds{1}_{\{X_s=i\}}\mathrm{d}s\le1$, by the Dominated Convergence Theorem the limit as $t\to\infty$ and the sum over $i\in I$ are exchangeable and the convergence obtained in \eqref{ET 1} implies that in \eqref{ET 2}.
	\end{remark}
\end{theorem}

\paragraph{The Null–Recurrent Case.}

In Appendix \ref{Ergodic Theorem Null Recurrent} we state and prove that the usual Ergodic Theorem does not hold for a Markov chain with unique null–recurrent class, showing that Theorem \ref{Ergodic Theorem Positive Recurrent} not only implies Theorem \ref{Ergodic Theorem Finite Space}, but is also the only possible generalization of the finite case.

\subsection{Ergodic Theorem for a Markov Chain on a Finite State Space with Unique Recurrent Class}

When trying to prove the Ergodic Theorem for a finite‐state continuous‐time Markov chain having exactly one recurrent class, the key point is a well‐known fact.

\begin{proposition} \label{Finite State Space and Recurrence}
	In a finite‐state Markov chain, every recurrent class is positive‐recurrent.
	\begin{proof}
		In the discrete–time framework, let $\lambda$ be an invariant distribution, then\[\mathbb{E}_i[T_i]=\frac{1}{\lambda_i}\quad\mbox{for all}\ i\in C.\] Note that $\lambda_i>0$ for all $i\in C$ since $C$ is irreducible and finite. Thus, $\mathbb{E}_i[T_i]<\infty$ for all $i\in C$, i.e., $C$ is positive recurrent.
		
		We conclude by applying this to the jump chain of our continuous–time Markov chain, since recurrence can be controlled on the jump chain (see Theorem \ref{Recurrence and Transience Jump Chain}).
	\end{proof}
\end{proposition}

Hence:\begin{itemize}
	\item If there is exactly one positive‐recurrent class and the rest of the states are transient, the Ergodic Theorem \ref{Ergodic Theorem Positive Recurrent} holds exactly as in the general case. The proof is even simpler in the finite‐state setting, since positive recurrence in a finite chain implies an exponential ergodic convergence (see Lemma \ref{Exponential ergodic convergence}) to the unique invariant distribution supported on the recurrent class,
	\item A scenario of exactly one null‐recurrent class cannot arise in a finite‐state Markov chain, so there is no need for a separate null‐recurrent Ergodic Theorem (which actually does not hold in the way we shall like to, see Appendix \ref{Ergodic Theorem Null}) in the finite‐state setting.
\end{itemize}

\begin{theorem}[Ergodic Theorem for a Markov Chain on a Finite State Space with Unique Recurrent Class] \label{Ergodic Theorem Finite Space}
	Suppose $I$ is a finite set of states, and $(X_t)_{t\ge0}$ is a continuous‐time Markov chain with generator matrix $Q$ on $I$. Suppose:\begin{itemize}
		\item[(i)] There is exactly one communicating class $C\subseteq I$ that is recurrent, thus positive–recurrent since $I$ is finite (see Proposition \ref{Finite State Space and Recurrence}),
		\item[(ii)] Every state in $I\setminus C$ is transient,
		\item[(iii)] Let $\lambda$ be the unique invariant distribution supported on the recurrent communicating class $C$ (same argument of Remark \ref{Invariant Distribution Ergodic}).
	\end{itemize} Then for any initial distribution $\nu$,\begin{equation}
		\frac{1}{t}\int_0^t\mathds{1}_{\{X_s=i\}}\,\mathrm{d}s\xrightarrow[t\to\infty]{\mbox{a.s.}}\begin{cases}\lambda_i&\ \mbox{if}\ i\in C\\0 &\ \mbox{if}\ i\notin C\end{cases}. \label{Ergodic Theorem Equation}
	\end{equation} Furthermore, for any bounded function $f: I\to\mathbb{R}$,\[
		\frac{1}{t}\int_{0}^{t} f(X_s)\,\mathrm{d}s
		\xrightarrow[t\to\infty]{\text{a.s.}}
		\sum_{i\in C} \lambda_if(i).
	\]
	\begin{proof}
		 Let $T_C:=\inf\{t\ge0:\ X_t\in C\}$, by hypothesis\[\mathbb{P}(T_C<\infty\mid X_0=j)=1\quad \text{for all } j\in I\setminus C.\] The Strong Markov Property (Theorem \ref{SMP}) states that, conditional on $T_C<\infty$:\begin{itemize}
			\item $(X_{T_C+t})_{t\ge 0}$ is a Markov chain with generator matrix $Q^{|C}$, i.e., the restriction of $Q$ to $C$,
			\item it is independent of the history before $T_C$.
		 \end{itemize} Hence we can study the long‑time averages after $T_C$ in total isolation.
	 
	 	Because $C$ is finite and recurrent, Proposition \ref{Finite State Space and Recurrence} guarantees it is positive recurrent. Finite state space immediately implies irreducibility of $Q^{|C}$. Therefore all assumptions of Theorem \ref{ET} are satisfied for the process $Y_t:=X_{T_C+t},\ t\ge 0$.
	 	
	 	Thus:\begin{itemize}
			\item If $i\in C$, Theorem \ref{ET} and Remark \ref{Invariant Distribution Ergodic} give \[\frac{1}{t}\int_0^{t}\mathds{1}_{\{Y_s=i\}}\,\mathrm{d}s\xrightarrow[t\to\infty]{\text{a.s.}}\lambda_i.\] Since $Y_t$ is just a time‑shift of $X_t$, the same limit holds for $X_t$ itself once $t$ is large enough,
			\item If $i\in I\setminus C$, then\[\mathbb{P}\left(\lim_{t\to\infty}\frac{1}{t}\int_0^t\mathds{1}_{\{X_s=i\}}\,\mathrm{d}s>0\right)\le\mathbb{P}\left(\lim_{t\to\infty}\int_0^t\mathds{1}_{\{X_s=i\}}\,\mathrm{d}s=\infty\right)=0,\] since it represents the total time spent in $i$, which is null because of $i$ being transient. 
			
			Thus,\[\mathbb{P}\left(\lim_{t\to\infty}\frac{1}{t}\int_0^t\mathds{1}_{\{X_s=i\}}\,\mathrm{d}s=0\right)=1.\]
		 \end{itemize}
	 
		 Combining these facts shows\[\frac{1}{t}\int_0^t\mathds{1}_{\{X_s=i\}}\,\mathrm{d}s\xrightarrow[t\to\infty]{\mbox{a.s.}}\begin{cases}\lambda_i&\ \mbox{if}\ i\in C\\0 &\ \mbox{if}\ i\notin C\end{cases}\] and by linearity (same argument of Remark \ref{Linearity}, made even simpler here since the state space is finite) we obtain the same result for $\frac{1}{t}\int_0^t f(X_s)\,\mathrm{d}s$.
	\end{proof}
	\begin{lemma}[Exponential Convergence of a Finite, Irreducible and Positive‐Recurrent Markov Chain to its Unique Invariant Distribution] \label{Exponential ergodic convergence}
		Consider a Markov chain on a finite state space $I=\{1,\hdots,N\}$. Assume the Markov chain to be irreducible and aperiodic (in the discrete–time case) or non–cyclic (in the continuous–time framework), ensuring that there is a unique invariant distribution $\lambda$, where $\lambda_i>0$ for all states $i\in I$.
		
		Then, not only convergence holds\[
		P^n(i,j)=\mathbb{P}(X_n = j \mid X_0 = i)
		\xrightarrow[n\to\infty]{}\lambda_j,
		\] but also exponential‐rate (geometric) convergence, i.e., there exist constants $0<r<1$ and $M<\infty$ such that\[
		\bigl|P^n(i,j) - \lambda_j\bigr|\le M\cdot r^n,
		\quad \text{for all }n\text{ and all states }i,j.\]
		It is written in the discrete–time case, but as we said it also holds for continuous–time Markov chains and the result is that $\|P_t - \Lambda\|$ decays exponentially fast in $t$, where $P_t$ is the transition matrix for time $t$ and $\Lambda$ is the invariant projection $\lambda\mathds{1}^T$.
		\begin{proof} We split the proof between the discrete and the continuous–time cases:
			\begin{itemize}
				\item \textbf{Discrete–Time}: Since the chain is finite, irreducible and aperiodic, there is a nonzero minimum probability of transitioning from any state to some specific state in a fixed number of steps. Equivalently, there exist $\delta>0$ and an integer $m$ such that for all states $i,j\in I$,	\[P^m(i,j) \ge \delta.\]
				
				From that, with probability at least $\delta$, the chain from any two initial states coalesces within $m$ steps. Once they coalesce, they evolve identically and remain coupled. This implies that for any two initial states $i$ and $k$, the total variation distance between the laws of $(X_{n,m}^{(i)}, X_{n,m}^{(k)})$ shrinks by a factor of $(1-\delta)$ at each multiple of $m$.
				
				Thus, we get geometric decay in distance to stationarity: for each $n$,\[\max_{i,j}\bigl|P^n(i,j)-\lambda_j\bigr|
				\le M\cdot r^n\quad\text{for some }M<\infty\text{ and }r<1.\]
				\item \textbf{Continuous–Time}: Firstly, we look at the jump chain: between jumps, the chain holds at its current state for an $E(\alpha)$ amount of time. The jump chain is discrete‐time on the same finite set, irreducible and aperiodic (since the original chain is not locked into a deterministic cycle), hence by the discrete‐time argument, the jump chain converges exponentially fast to its invariant distribution $\lambda$. We conclude by translating that back to the continuous‐time process:\[\bigl\|P_t(i,\cdot)-\lambda(\cdot)\bigr\|=O(e^{-\beta\,t})\quad\mbox{for some}\ \beta>0.\]
			\end{itemize}
		\end{proof}
		\begin{note}
			This lemma also applies to the chains in our model, i.e., consistent with Assumption \ref{A3}, which have a finite, irreducible and positive–recurrent subchain.
		\end{note}
	\end{lemma}
\end{theorem}

\section{Convergence of the Empirical Joint State–Policy Distribution to the Solution of the Mean–Field Model}

It follows, from the Picard–Lindelöf Theorem \cite[Theorem~5.7]{Smirnov} \cite[Theorem~4.A.5]{Sandholm}, that a solution trajectory \( \mu_{\mathcal{S} \times \mathcal{U}}(t) \) to \eqref{Mean-Field dotmu} exists and is unique and, from Kolmogorov's Strong Law of Large Numbers \cite[Theorem~2.4.1]{Lemmas}, that it approximates arbitrarily well the evolution of the empirical joint state–policy distribution \( \hat{\mu}^N_{\mathcal{S} \times \mathcal{U}}(t) \) as \( N \to \infty \), as shown in the following result.

\begin{lemma}[Convergence of the Empirical Joint State–Policy Distribution to the Solution of the Mean–Field Model]\label{Lemma 2.1}
	A solution to \eqref{Mean-Field dotmu} with initial condition \( \mu_{\mathcal{S} \times \mathcal{U}}(0) \) exists in \( t \in [0, \infty) \), is unique, and is Lipschitz continuous w.r.t. \( \mu_{\mathcal{S} \times \mathcal{U}}(0) \). Furthermore, if \( (s^{i,N}(0), u^{i,N}(0)) \sim \mu_{\mathcal{S} \times \mathcal{U}}(0) \) for all \( i \in [N] \), then \( \hat{\mu}^N_{\mathcal{S} \times \mathcal{U}}(t) \) and \( \hat{\mu}^N_{\mathcal{S} \times \mathcal{A}}(t) \) converge almost surely to \( \mu_{\mathcal{S} \times \mathcal{U}}(t) \) and \( \mu_{\mathcal{S} \times \mathcal{A}}(t) \), respectively, as \( N \to \infty \) for all \( t \in [0, \infty) \).
	\begin{proof}
		As we saw in Section \ref{Mean-Field Model}, when $\mathrm{d}t\rightarrow0$, the balance equation \eqref{Mean-Field dmu} takes the form of the ODE \eqref{Mean-Field dotmu}. Let $V(x):=x\cdot Q_{\mathcal{S}}$, where $Q_{\mathcal{S}}=(q_{s,s'})_{s,s'\in\mathcal{S}}$ denotes the state transition Q–matrix\footnote{In the context of this result, the state transition Q–matrix is sufficient, because no policy switches are taken into account, as it is expressed in the balance equation \eqref{Mean-Field dmu}.}, as defined in Definition \ref{StateQ}. Then the Cauchy problem we are studying can be written as:\begin{equation}
			\begin{cases}
				\dot{\mu}_{\mathcal{S}}^u(t)=V(\mu_{\mathcal{S}}^u(t))=\mu_{\mathcal{S}}^u(t)\cdot Q_{\mathcal{S}}\\
				\mu_{\mathcal{S}}^u(0)=\mu(u,0)
			\end{cases},
		 	\label{Equation proof 5.1.1}
		\end{equation} where for each policy $u\in\mathcal{U}$, we define \[\mu_{\mathcal{S}}^u(t):=\mu_{\mathcal{S}\times\mathcal{U}}[\cdot,u](t)\in\mathcal{P}_u(\mathcal{S}),\] where $\mathcal{P}_u(\mathcal{S})$ denotes the space of marginal distributions over the state space $\mathcal{S}$ for a fixed policy $u$. Hence, $\mu_{\mathcal{S}}^u(t)$ is a vector of dimension $|\mathcal{S}|$ such that $\mu_{\mathcal{S}}^u[s](t)=\mu_{\mathcal{S}\times\mathcal{U}}[s,u](t)$. By doing this, we split the ODE into $|\mathcal{U}|$ different ODEs, one for each policy, to perceive the full context of this lemma, in which policies are held fixed and players only change their states.
	
		Thus, the vector field $V$ lies on the tangent space of $\mathcal{P}(\mathcal{S})$, the space of probability distributions over $\mathcal{S}$ (for an in–depth analysis of this statement see the dedicated remark below). Since $\mathcal{P}_u(\mathcal{S})$ is a compact convex space and $V$ is Lipschitz continuous (for thorough proofs of both see the remarks below), we apply an extension of the Picard–Lindel\"of Theorem to compact convex spaces \cite{Smirnov}\cite{Sandholm} and obtain that a unique solution $\mu_{\mathcal{S}}^u(t)$ exists for all $t\in[0,\infty)$.
		
		To establish Lipschitz continuity of $\mu_{\mathcal{S}}^u(t)$ with respect to the initial condition $\mu_{\mathcal{S}}^u(0)$, we commence with observing that the vector field $V$ satisfies a Lipschitz condition (see the remark below):\begin{equation*}
			\norm{V(\mu_1)-V(\mu_2)}\le L\cdot\norm{\mu_1-\mu_2}
		\end{equation*} for some constant $L$ and for all $\mu_1,\mu_2\in\mathcal{P}_u(\mathcal{S})$. Applying Gronwall's Inequality \cite[Chapter~5]{Norris}\cite{Sandholm}, we obtain:\begin{equation*}
			\norm{\mu_1(t)-\mu_2(t)}\le e^{Lt}\cdot\norm{\mu_1(0)-\mu_2(0)}.
		\end{equation*} Thus, the constant $e^{Lt}$ establishes that $\mu_{\mathcal{S}}^u(t)$ is Lipschitz continuous with respect to $\mu_{\mathcal{S}}^u(0)$.
		
		Finally, to show that $\hat{\mu}_{\mathcal{S}\times\mathcal{U}}^N(t)$ and $\hat{\mu}_{\mathcal{S}\times\mathcal{A}}^N(t)$ converge almost surely to $\mu_{\mathcal{S}\times\mathcal{U}}(t)$ and $\mu_{\mathcal{S}\times\mathcal{A}}(t)$, respectively, we split the proof into separate convergences, one for each policy, showing that $\hat{\mu}_{\mathcal{S}}^{u,N}(t)$ converges almost surely to $\mu_{\mathcal{S}}^u(t)$, where $\hat{\mu}_{\mathcal{S}}^{u,N}(t):=\hat{\mu}_{\mathcal{S}\times\mathcal{U}}^N[\cdot,u](t)$.
		
		Because the players evolve independently under a fixed policy $u$, the argument of Remark \ref{Law of Large Numbers limit at fixed time} applies: for every $t$ the empirical vector $\hat{\mu}^{u,N}_{\mathcal{S}}(t)$ is an average of i.i.d. variables and converges almost surely to $\mu^{u}_{\mathcal{S}}(t)$ as $N\to\infty$.
		
		Therefore $\hat{\mu}^{N}_{\mathcal{S}\times\mathcal{U}}(t)\rightarrow\mu_{\mathcal{S}\times\mathcal{U}}(t)$ almost surely. Then, by the Continuous Mapping Theorem\footnote{The applicability of Theorem \ref{CMP} in this context is discussed below in a dedicated remark.} (see Theorem \ref{CMP}), we obtain the same type of convergence for the joint state–action distribution:\[	\hat{\mu}^N_{\mathcal{S} \times \mathcal{A}}(t)[s,a] = \int_{\mathcal{U}} \hat{\mu}^N_{\mathcal{S} \times \mathcal{U}}(t)[s,u] \cdot u(a \mid s) \,\mathrm{d}\mu(u)\xrightarrow[N\to\infty]{\mbox{a.s.}} \mu_{\mathcal{S} \times \mathcal{A}}(t)[s,a],\] completing the proof.
	\end{proof}

	\begin{remark}[Key Reasons Why $V$ Lies on the Tangent Space of  $\mathcal{P}(\mathcal{S})$]
		The vector field $V(\mu)=\mu\cdot Q_{\mathcal{S}}$ lies in the tangent space of $\mathcal{P}(\mathcal{S})$ because it respects the probability conservation law\begin{equation*}
			\sum\limits_{s'\in\mathcal{S}}q_{s,s'}=0.
		\end{equation*} This follows directly from the structure of the Q–matrix, which ensures that probability mass is redistributed but never created or destroyed.
	\end{remark}
	
	\begin{remark}[Convexity and compactness of the space $\mathcal{P}_u(\mathcal{S})$]
		Firstly, we prove convexity and compactness of $\mathcal{P}(\mathcal{S})$. Then, we deduce that $\mathcal{P}_u(\mathcal{S})$ inherits these properties.
		
		Let $\mu_1,\mu_2\in\mathcal{P}(\mathcal{S})$ and consider the convex combination $\mu:=\lambda\mu_1+(1-\lambda)\mu_2$ for some $\lambda\in[0,1]$, then:\begin{itemize}
			\item $\mu(s)\ge0\ \forall\ s\in\mathcal{S}$,
			\item $\sum\limits_{s\in\mathcal{S}}\mu(s)=\lambda\cdot\sum\limits_{s\in\mathcal{S}}\mu_1(s)+(1-\lambda)\cdot\sum\limits_{s\in\mathcal{S}}\mu_2(s)=\lambda\cdot1+(1-\lambda)\cdot1=1$.
		\end{itemize} Thus, $\mathcal{P}(\mathcal{S})$ is convex.
		
		Since $\mathcal{S}$ is a finite space, let $n:=|\mathcal{S}|$, it is possible to identify the space $\mathcal{P}(\mathcal{S})$ with the unit simplex $\Delta_n$ of dimension $n$, which is a closed and bounded subset of $\mathbb{R}^n$, thus compact.
		
		Finally, we conclude that what above implies convexity and compactness of the set $\mathcal{P}_u(\mathcal{S})$ of marginal distributions for a fixed policy $u$:\begin{itemize}
			\item Convexity: Since policy $u$ is fixed, the convexity of $\mathcal{P}(\mathcal{S})$ implies the one of $\mathcal{P}_u(\mathcal{S})$,
			\item Compactness: Since $\mathcal{P}_u(\mathcal{S})\subseteq\mathcal{P}(\mathcal{S})$, which is compact, to ensure the compactness of $\mathcal{P}_u(\mathcal{S})$, it is sufficient to prove its closedness. Taken any convergent sequence of elements in $\mathcal{P}_u(\mathcal{S})$, its limit must belong to $\mathcal{P}_u(\mathcal{S})$, since policy $u$ is fixed.
		\end{itemize}
	\end{remark}
	
	\begin{remark}[Key Reasons Why $V$ Satisfies a Lipschitz Condition]
		The vector field $V$ satisfies
			\begin{align*}
				V(\mu_1)-V(\mu_2)=&R_d\cdot\sum\limits_{s'\in\mathcal{S}}\sum\limits_{a'\in\mathcal{A}(s')}\phi(s\mid s',a')\cdot u(a'\mid s')\cdot(\mu_1[s']-\mu_2[s'])+\\&-R_d\cdot(\mu_1[s]-\mu_2[s]).
			\end{align*}We conclude by using the fact that both $\phi(s\mid s',a')$ and $u(a'\mid s')$ are bounded.
	\end{remark}

	\begin{remark}[Law‑of‑Large‑Numbers limit at fixed time] \label{Law of Large Numbers limit at fixed time}
		Fix a policy $u\in\mathcal{U}$. For each player $i\in[N]$ the process $(s^{N}_{i}(t))_{t\ge0}$ is an independent copy of the continuous‑time Markov chain with generator $Q_\mathcal{S}$.
		
		Hence, for every state $s\in\mathcal{S}$ and every $t\ge0$ the random variables\[X_i(t):=\mathds{1}_{\{s^{i,N}(t)=s,\ u_i=u\}},\ i=1,\hdots,N\] are i.i.d.\footnote{The random pairs $(s_1^{N}(t),u_1),\ldots,(s_N^{N}(t),u_N)$ are i.i.d.\ (see Section \ref{Finite Population Model}).	Each $X_i(t)$ is the indicator of the Borel event $\{s_i^{N}(t)=s,\ u_i=u\}$; hence the independence carries over and $\mathbb P\left(X_i(t)=1\right)=\mathbb P\left(s_i^{N}(t)=s,\ u_i=u\right)=\mu^{u}_{\mathcal S}[s]$, so $X_i(t)\sim\text{Bernoulli}\left(\mu^{u}_{\mathcal S}[s]\right)$, showing that $X_1(t),\hdots,X_N(t)$ are i.i.d. at every fixed time $t$.} with expected value\[\mathbb{E}[X_i(t)]=\mathbb{P}(s^{i,N}(t)=s)\overset{(\star)}{=}\mu^{u}_{\mathcal{S}}[s](t)<\infty.\]
		
		By Kolmogorov’s Strong Law of Large Numbers \cite[Theorem~2.4.1]{Lemmas},\[\hat{\mu}^{u,N}_{\mathcal{S}}[s](t)=\frac{1}{N}\sum_{i=1}^{N}X_i
		\xrightarrow[N\to\infty]{\mbox{a.s.}}\mu^{u}_{\mathcal{S}}[s](t).\]
		
		\begin{note}
			The equality ($\star$) holds because, by definition, $\mu^{u}_{\mathcal S}[s]=\mu_{\mathcal S\times \mathcal U}[s,u]$. Since the policy component $u_i$ is frozen over time, the slice $\mu^{u}_{\mathcal S}(t)=\mu_{\mathcal S\times \mathcal U}[\cdot,u](t)$ evolves as $\mu^{u}_{\mathcal S}(t)=\mu^{u}_{\mathcal S}(0)\cdot e^{tQ_\mathcal S}$,	and its $s$-th component is exactly the joint probability $\mathbb P\left(s_i^{N}(t)=s,\ u_i=u\right)$.
		\end{note}
	\end{remark}

	\begin{remark}[Applicability of the Continuous Mapping Theorem]
		Fix a policy $u\in\mathcal{U}$, we want to demonstrate that the map\[F_{s,a}: \hat{\mu}\mapsto\int_{\mathcal{U}}\hat{\mu}[s,u]\cdot u(a\mid s)\,\mathrm{d}\mu(u)\] is continuous for all $s\in\mathcal{S},\ a\in\mathcal{A}(s)$. Specifically, we show that if $\hat{\mu}_n\to\hat{\mu}$, then $F_{s,a}(\hat{\mu}_n)\to F_{s,a}(\hat{\mu})$.
		
		Since both $\hat{\mu}_n[s,u]$ and $u(a\mid s)$ take values in $[0,1]$, then $F_{s,a}$ is well–defined (the integrand is always bounded by $1$ and positive) and\[\bigl|		\hat{\mu}_n[s,u]\cdot u(a\mid s)-\hat{\mu}[s,u]\cdot u(a\mid s)\bigr|=\bigl|\hat{\mu}_n[s,u] - \hat{\mu}[s,u]\bigr|\cdot u(a\mid s)\le\bigl|\hat{\mu}_n[s,u] - \hat{\mu}[s,u]\bigr|.\] Integrating over $\mathcal{U}$, we obtain\[\left|\int_{\mathcal{U}}\hat{\mu}_n\cdot u(a\mid s)\,\mathrm{d}\mu(u)-\int_{\mathcal{U}}\hat{\mu}\cdot u(a\mid s)\,\mathrm{d}\mu(u)\right|\le\int_{\mathcal{U}}\left|\hat{\mu}_n-\hat{\mu}\right|\,\mathrm{d}\mu(u)\xrightarrow{}0,\] ensuring continuity.
	\end{remark}
\end{lemma}

\begin{note}[Convergence Rate through Kolmogorov's Strong Law of Large Numbers]
	Since $0\le X_i\le1$, we have\[\mathbb{E}[|X_i|^k]\le1,\quad \mbox{for all}\ k\ge 1.\] Thus, we obtain a convergence rate of $\mathcal{O}(N^{-1/2})$ \cite{SLLNRate}.
\end{note}

\section{Approximation of Nash Equilibria of the Finite Population Game by the Mean–Field Model}

In this section we rigorously state and prove that an MSNE in the mean–field game approximates arbitrarily well, for large enough $N$, an MSNE in the finite population game.

\begin{theorem}[Approximation of Nash Equilibria of the Finite Population Game by the Mean–Field Model] \label{Theorem 3.2}
	Let \( \mu \) be a MSNE in the discounted payoff mean–field game according to Definition \ref{MSNE}. Then, for any \( \epsilon > 0 \) there is \( N_{\epsilon} \in \mathbb{N} \) such that for any \( N > N_{\epsilon} \), a collection of policies \( \{ u^i \}_{i \in [N]} \) with\begin{equation*}
		\left| \frac{1}{N} \sum_{i \in [N]} \mathds{1}_{u^i}(u) - \mu(\mathcal{S}, u) \right| \leq \frac{1}{N}, \quad \forall u \in \mathcal{U}_D
	\end{equation*} is a weak \( \epsilon \)–MSNE (see Definition \ref{epsilon MSNE}) in the finite–population game with initial condition \( \mu \).

	\begin{proof}
		Consider the expected payoff function for player \( i \):\[	J^{i,N}(\mu, u^1, \hdots, u^N) = \mathbb{E} \left[ \sum_{k=0}^{\infty} \beta^k\cdot r \left( s^{i,N}(t_k),\ a^{i,N}(t_k),\ \hat{\mu}_{\mathcal{S} \times \mathcal{A}}^N(t_k) \right) \right] \]
	
		From Lemma \ref{Lemma 2.1}, we know that the empirical state–action distribution \( \hat{\mu}_{\mathcal{S} \times \mathcal{U}}^N(t) \) converges almost surely to \( \mu_{\mathcal{S} \times \mathcal{U}}(t) \) as \( N \to \infty \). Thus, we can apply the Dominated Convergence Theorem \cite[Theorem~9.1.2]{Rosenthal} and switch the limit as $N\to\infty$ with the expectation:\begin{align*}
			\lim_{N\to\infty}J^{i,N}(\mu, u^1, \hdots, u^N)&=\lim_{N\to\infty}\mathbb{E} \left[ \sum_{k=0}^{\infty} \beta^k\cdot r \left( s^{i,N}(t_k),\ a^{i,N}(t_k),\ \hat{\mu}_{\mathcal{S} \times \mathcal{A}}^N(t_k) \right) \right] \\&\overset{(\star)}{=} \mathbb{E} \left[ \sum_{k=0}^{\infty} \beta^k\cdot r \left( s^{i,N}(t_k),\ a^{i,N}(t_k),\ \mu_{\mathcal{S} \times \mathcal{A}}(t_k) \right) \right] \\ &=J(u^i, \mu^\infty_{\mathcal{S}\times\mathcal{A}}),
		\end{align*} where $\mu^\infty_{\mathcal{S}\times\mathcal{A}}(s,a)=\sum_{u\in\mathcal{U}_D}\mu(s,u)\cdot u(a\mid s)$.
	
		Because $\mu$ is an MSNE, every deterministic policy that receives positive mass is a best response. We built the profile $\{u_i\}{i\in[N]}$ by sampling from $\mu$, so with probability $1$ it holds $\mu(\mathcal S,u_i):=\sum_{s\in \mathcal S}\mu_{\mathcal S\times\mathcal U}[s,u^i]>0$ for every realised $u_i$. Consequently, for each such $i$\[J\left(u^i,\mu_{\mathcal S\times\mathcal A}^{\infty}\right)=\max_{v\in \mathcal U_D}J\left(v,\mu_{S\times A}^{\infty}\right),\] and using the same procedure as before through the Dominated Convergence Theorem, but now when player $i$ uses policy $v\in\mathcal{U}_D$, yields\[ \lim_{N\to\infty}J^{i,N}(\mu,u^1,\hdots,u^{i-1},v,u^{i+1},\hdots,u^N)=J(v,\mu^\infty_{\mathcal{S}\times\mathcal{A}})\le\max_{v\in\mathcal{U}_D}J(v, \mu^\infty_{\mathcal{S}\times\mathcal{A}}). \]
		
		Thus,\[\lim_{N\to\infty}J^{i,N}(\mu,u^1,\hdots,u^{i-1},v,u^{i+1},\hdots,u^N)\le\lim_{N\to\infty}J^{i,N}(\mu, u^1, \hdots, u^N).\] By the definition of limit we obtain that for any $\epsilon>0$ there exists $N_\epsilon\in\mathbb{N}$ such that for all $N>N_\epsilon$\[J^{i,N}(\mu, u^1, \hdots, u^N)>J^{i,N}(\mu,u^1,\hdots,u^{i-1},v,u^{i+1},\hdots,u^N)-\epsilon,\quad\forall v\in\mathcal{U}_D,\] establishing the weak \( \epsilon \)–MSNE condition.
	\end{proof}

	\begin{remark}[Interchangeability of the Limit as $N\to\infty$ and the Infinite Discounted Sum]
		 In this Remark we give a step–by–step justification of the equality $(\star)$.
		 
		 Let $\beta\in(0,1)$ and $R:=\sup_{s,a,\mu}|r(s,a,\mu)|<\infty$ (both $\mathcal{S}$ and $\mathcal{A}$ are finite, so $\mathcal{P}(\mathcal{S}\times\mathcal{A})$ is compact). For player $i\in[N]$ we define\begin{align*}
			Z_k^{i,N}&:=r(s^{i,N}(t_k),\ a^{i,N}(t_k),\ \hat{\mu}^N_{\mathcal{S}\times\mathcal{A}}(t_k)),\\ Z_k^i&:=r(s^{i,N}(t_k),\ a^{i,N}(t_k),\ \mu_{\mathcal{S}\times\mathcal{A}}(t_k)).
		 \end{align*} Lemma \ref{Lemma 2.1} tells us that \(\widehat\mu^{N}_{\mathcal S\times \mathcal A}(t_k)\xrightarrow[N\to\infty]{\text{a.s.}}\mu_{\mathcal S\times \mathcal A}(t_k)\) for every fixed $k$. Because $r$ is continuous in its third argument (see Assumption \ref{Assumption1}) and bounded, we get\[Z_{k}^{i,N}\xrightarrow[N\to\infty]{\text{a.s.}} Z_{k}^i\quad\text{for every fixed }k\text{ and }i\in[N].\]
	 
	 	For all $N,k$ we have $|Z_{k}^{i,N}|\le R$.	Therefore the random variables\begin{align*}
			Y^{i,N}&:=\sum_{k=0}^{\infty}\beta^{k}\cdot Z_{k}^{i,N},\\
			Y^i&:=\sum_{k=0}^{\infty}\beta^{k}\cdot Z_{k}^i.
	 	\end{align*} satisfy $|Y^{i,N}|\le R/(1-\beta)$, independently of $N$.
 	
 		Now, fix $K\in\mathbb{N}$. Split the series:\[Y^{i,N}=\sum_{k=0}^{K-1}\beta^{k}\cdot Z_{k}^{i,N}+\sum_{k=K}^{\infty}\beta^{k}\cdot Z_{k}^{i,N}.\]\begin{itemize}
			\item For each fixed $k<K$ the Dominated Convergence Theorem \cite[Theorem~9.1.2]{Rosenthal} gives\[\mathbb{E}\left[\beta^{k}\cdot Z_{k}^{i,N}\right]\xrightarrow[N\to\infty]{}\mathbb{E}\left[\beta^{k}\cdot Z_{k}^i\right].\]
			\item Regardless of $N$,\[\left|\mathbb{E}\left[\sum_{k=K}^{\infty}\beta^{k}\cdot Z_{k}^{i,N}\right]\right|\le\frac{R\cdot\beta^{K}}{1-\beta}\xrightarrow[K\to\infty]{}0.\] Because the right‑hand side is $N$–independent, we have uniform convergence of the tails.
 		\end{itemize}
 	
 		Thus, define the partial–sum function\[S^{i,N}(K):=\sum_{k=0}^{K-1}\beta^k\cdot\mathbb{E}\left[Z_k^{i,N}\right]\]\begin{itemize}
			\item Because $|Z^{i,N}_{k}|\le R$ for all $k,N$, the series $\sum_{k\ge 0}\beta^{k}\cdot \mathbb{E}[Z^{i,N}_{k}]$ converges absolutely, hence $A^{i,N}:=\lim_{K\to\infty}S^{i,N}(K)=\mathbb E[Y^{i,N}]$ exists for every $N$.
			\item The geometric bound \[\left|A^{i,N}-S^{i,N}(K)\right|=\left|\mathbb{E}\left[\sum_{k=K}^{\infty}\beta^{k}\cdot Z^{i,N}{k}\right]\right|\le\frac{R\cdot\beta^{K}}{1-\beta}\] is independent of $N$. Hence $S^{i,N}(K)\xrightarrow[K\to\infty]{}A^{i,N}$ uniformly in $N$.
 		\end{itemize}
 	
 		With these two facts in place, both assumptions of \cite[Theorem~7.11]{Rudin} (uniform convergence	and existence of the outer limits $A^{i,N}$) are satisfied, so we may invoke the theorem	to conclude that\[\lim_{N\to\infty}A^{i,N}=\lim_{N\to\infty}\lim_{K\to\infty}S^{i,N}(K)=\lim_{K\to\infty}\lim_{N\to\infty}S^{i,N}(K),\]i.e.,\[\lim_{N\to\infty}\mathbb{E}\left[Y^{i,N}\right]=\mathbb{E}\left[Y^i\right].\] Putting the steps together yields exactly the equality marked by $(\star)$.
	\end{remark}
\end{theorem}

\section{Convergence of Finite Population Joint State–Policy Distributions to the Evolutionary Dynamics Mean– Field Model}

In this section we encounter the last theoretical results, which regard the mean–field evolutionary dynamics. We begin with Lemma \ref{Lemma 4.1}, which is an existence and uniqueness result based on Picard–Lindel\"of Theorem. This lemma forms the foundation for our last result, Theorem \ref{Theorem 4.1}, which demonstrates that the joint state–policy distributions in the finite population game converge to the ones of the mean–field evolutionary dynamics as the population size grows.

\begin{lemma}[Well–Definition of the Mean–Field Evolutionary Dynamics]\label{Lemma 4.1}
	Under Assumptions \ref{Assumption1} and \ref{Assumption2}, a solution to the mean dynamic of revision protocol \( \rho \), characterized by \eqref{ODE8}, with initial condition \( \mu(0) \in \mathcal{P}(\mathcal{S} \times \mathcal{U}) \) exists in \( t \in [0, \infty) \), is unique, and is Lipschitz continuous w.r.t. \( \mu(0) \).
	\begin{proof}
		First, notice that \eqref{ODE8} can be written for all states \( s \in \mathcal{S} \) and policies \( u \in \mathcal{U}_D \) in vector form as an ODE with a vector field \( V : \mathcal{P}(\mathcal{S} \times \mathcal{U}_D) \to T\mathcal{P}(\mathcal{S} \times \mathcal{U}_D) \), where \( T\mathcal{P}(\mathcal{S} \times \mathcal{U}_D) \) denotes the tangent space of \( \mathcal{P}(\mathcal{S} \times \mathcal{U}_D) \).\\	Second, notice that for all \( s \in \mathcal{S} \) and all \( u \in \mathcal{U}_D \), \( J(u, s, \mu_{\mathcal{S} \times \mathcal{A}}) \), as defined in \eqref{DefinitionNE}, can be written as a linear combination of a finite number of single state reward functions. \\ Therefore, due to Assumption \ref{Assumption1}, \( J(u, s, \mu_{\mathcal{S} \times \mathcal{A}}) \) is Lipschitz continuous w.r.t. \( \mu \).\\ Hence, for any \( s \in \mathcal{S} \), \( F^s(\mu) \), defined in \eqref{PolicyTransitions}, is Lipschitz continuous w.r.t. \( \mu \).\\ Furthermore, due to Assumption \ref{Assumption2}, \( V(\mu) \) is Lipschitz continuous w.r.t. \( \mu \).
		
		Under these conditions, since \( \mathcal{P}(\mathcal{S} \times \mathcal{U}_D) \) is convex and compact, existence and uniqueness follows from an extension of the Picard–Lindelöf Theorem to compact convex spaces \cite[Theorem~5.7]{Smirnov} \cite[Theorem~4.A.5]{Sandholm} and Lipschitz continuity follows from Gr\"onwall’s Inequality \cite[Theorem 4.A.3]{Sandholm}.
	\end{proof}
\end{lemma}

\begin{theorem}[Convergence of Finite Population Joint State–Policy Distribution to the Evolutionary Dynamics Mean–Field Model]\label{Theorem 4.1}
	If \( (s^{i,N}(0), u^{i,N}(0)) \sim \mu(0) \) for all \( i \in [N] \), then \( \hat{\mu}^N(t) \) converges in probability to \( \mu(t) \) for all \( t \in [0, \infty) \) as \( N \to \infty \).
	
	\begin{proof}
		We aim to show that as the population size $N$ grows large, the empirical distribution of joint state–policy $\hat{\mu}^N(t)$ converges in probability to the solution $\mu(t)$ of the mean–field evolutionary dynamics, for all $t\in[0,\infty)$.
		
		In Lemma \ref{Lemma 4.1} we showed that the joint state–policy distribution $\mu(t)$ is the unique solution of the ODE \eqref{ODE8}. This equation can be written in vector form as an ODE with a vector field $V$, which thus satisfies $V(x):=x\cdot Q$, where $Q$ denotes the joint state–policy transition Q–matrix (see Definition \ref{Joint Q-matrix}). Then the Cauchy problem studied in Lemma \ref{Lemma 4.1} can be written as:\begin{equation}
			\begin{cases}
				\dot{\mu}(t)=V(\mu(t))=\mu(t)\cdot Q\\
				\mu(0)
			\end{cases}.
			\label{Cauchy proof 5.3.2}
		\end{equation}
		
		To show that $\hat{\mu}^N(t)$ converges in probability to $\mu(t)$ we apply Kurtz's Theorem \cite[Chapter~5]{Norris}\cite{Sandholm} (see the remark below) as presented in Section \ref{Kurtz}, which, in our framework, states that the empirical joint state–policy distribution $\hat{\mu}^N(t)$ converges in probability to the deterministic measure $\mu(t)$ as $N\rightarrow\infty$.
	\end{proof}
	\begin{remark}[Applicability of Kurtz's Theorem to the Evolutionary Dynamics]
		Let $X_t^N=\hat{\mu}^N(t)\in\mathcal{P}(\mathcal{S}\times\mathcal{U}_D)$ and $\mathcal{X}^N:=\mathcal{P}(\mathcal{S}\times\mathcal{U}_D)$. The process $(X_t^N)_{t\ge0}$ evolves according to the joint state–policy transition Q–matrix $Q$ defined in \eqref{Joint Q-matrix}. The limiting mean–field distribution $\mu(t)$ solves the deterministic Cauchy problem \ref{Cauchy proof 5.3.2}.
		
		To apply Kurtz's Theorem with the notation of Section \ref{Kurtz}, we define the following in our context, with $x\in\mathcal{X}^N$:\begin{itemize}
			\item Total jump rate: \[\lambda_x^N:=N\cdot\sum_{(s,u)\in\mathcal{S}\times\mathcal{U}_D}x[s,u]\cdot q(s,u),\] where $q(s,u)=-Q_{(s,u),(s,u)}$.
			\item Increment: a single jump in the process changes the empirical distribution by \[\zeta_x^N:=\frac{1}{N}(\textbf{e}_{(s',u')}-\textbf{e}_{(s,u)})\] with probability\[\mathbb{P}\left(\zeta_x^N=\frac{1}{N}(\textbf{e}_{(s',u')}-\textbf{e}_{(s,u)})\right)=\frac{x[s,u]\cdot Q_{(s,u),(s',u')}}{\lambda_x^N},\] where $\textbf{e}_i$ denotes the $i$–th elementary basis vector.
		\end{itemize}
		
		Now, we verify the conditions of Kurtz's Theorem:\begin{itemize}
			\item[(i)] Let $V^N(x):=\lambda_x^N\cdot\mathbb{E}[\zeta_x^N]$, to apply Kurtz's Theorem we need to show that the empirical process converges uniformly to the mean–field limit as $N\to\infty$, i.e.,\[\lim_{N\to\infty}\sup_{x\in\mathcal{X}^N}|V^N(x)-V(x)|=0.\tag{$\star$}\]
			
			It holds (see \cite[Section~10.2.2]{Sandholm}):\begin{align*}
				V^N(x)&=\lambda_x^N\cdot\mathbb{E}[\zeta_x^N]\\&=\sum_{(s,u)\in\mathcal{S}\times\mathcal{U}_D}x[s,u]\cdot\sum_{(s',u')\ne(s,u)}Q_{(s,u),(s',u')}e_{(s',u')}+\\&\hspace{16pt}-\sum_{(s,u)\in\mathcal{S}\times\mathcal{U}_D}x[s,u]\cdot q(s,u)e_{(s,u)}\\&=xQ=V(x).
			\end{align*}Thus, $V^N$ is independent of $N$ and $V^N\equiv V$, so ($\star$) trivially holds.
			\item[(ii)] The expected displacement $A^N(x):=\lambda_x^N\cdot\mathbb{E}[|\zeta_x^N|]$ is uniformly bounded in $N$ since each individual jump changes the distribution by at most $1/N$:\[\sup_N\sup_{x\in\mathcal{X}^N}A^N(x)<\infty.\]
			\item[(iii)] The expected displacement due to large jumps vanishes:
			\[
			\lim_{N\to\infty}\sup_{x\in\mathcal{X}^N}A^N_\delta(x) = 0 \quad \text{for any } \delta > 0,
			\]
			since $|\zeta_x^N| \leq 2/N$ implies that for large $N$, no jump exceeds any fixed $\delta > 0$, where $A^N_\delta(x)=\lambda_x^N\cdot\mathbb{E}[|\zeta_x^N\cdot\mathds{1}_{\{|\zeta_x^N|>\delta\}}|]$.
		\end{itemize}
		
		Therefore, by \eqref{Kurtz's equation} we obtain that
		\[
		\lim_{N\to\infty}\mathbb{P}\left(\sup_{t \in [0,T]} \left| X_t^N - x_t \right|<\epsilon\right) = 1 \quad \mbox{for all}\ T<\infty,\ \epsilon>0,
		\] i.e.,\[\hat{\mu}^N(t)\xrightarrow[N\to\infty]{\mathbb{P}}\mu(t)\quad\mbox{for all}\ t\in[0,\infty).\]
	\end{remark}
\end{theorem}

\begin{note}[Convergence Rate through Kurtz's Theorem]
	We proved the convergence in probability by an application of Kurtz's Theorem, which guarantees a convergence rate of $\mathcal{O}(N^{-1/2})$ \cite{KurtzRate}.
\end{note}

\chapter{Practical Results} \label{Practical Results}

The objective of this chapter is to confirm experimentally the two key convergence statements proved in Chapter \ref{Theoretical Results}, namely\begin{itemize}
	\item Lemma \ref{Lemma 2.1}: almost–sure convergence of the empirical joint state–policy distribution to the solution of the mean–field model,
	\item Theorem \ref{Theorem 4.1}: convergence in probability of the joint state–policy distribution to the evolutionary dynamics mean–field model.
\end{itemize}

We first outline the practical model and set–up, analyze the code and then present the results observed for the empirical distributions both with fixed policy and evolutionary dynamics. Throughout the chapter we measure the convergence as a discrepancy in the sup–norm (see Section \ref{Choice of the Topology}).

\section{Experimental Model and Set–Up}

The model we use for the simulations, that follows the one presented in \cite{Wiecek}, is a concrete, battery–driven instance of the continuous–time mean–field game developed in Chapters \ref{Model}, \ref{Mathematical Framework} and \ref{Theoretical Results}.

Each player is a mobile device whose state $s\in\{0,1,\hdots,S\}$ records its remaining battery charge ($0=\text{empty},\ S=\text{full}$). When its individual Poisson decision clock rings, the device chooses an energy level $a\in\mathcal A(s)=\{1,\hdots ,a_s\}=:q$ at which to transmit. Higher energy levels speed up data transfer but drain the battery faster, creating a tension that is typical of many wireless applications. The state then jumps according to a simple Markov rule: with probability\[p(a)=1-\alpha\cdot q_a-\gamma\] the battery stays where it is, otherwise it decreases by one unit (unless it is already empty, in which case it is recharged to S with probability $p_{0,S}$). Because $p(a)$ depends explicitly on the chosen action, individual decisions shape the population–level transition kernel $\phi$ and hence the mean–field dynamics.

\paragraph{What the model represents.}

\begin{itemize}
	\item \textbf{Energy–constrained communication}: The battery variable stands for any exhaustible local resource (energy, tokens, credits). The action set represents a menu of operating intensities.
	\item \textbf{Asynchronous interaction}: Agents act at independent random times, mirroring the lack of global coordination in large–scale networks.
	\item \textbf{Congestion feedback}: Rewards (defined in Section \ref{Finite Population Model}) can penalise widely chosen high–energy actions, capturing the way collective behaviour feeds back on individual pay–offs.
	\item \textbf{Reset mechanism}: The probabilistic jump from state $0$ to state $S$ models sporadic recharges—for instance plugging a phone in or harvesting solar energy—and guarantees a single positive–recurrent class, a key assumption for the ergodic theorems proved earlier.
\end{itemize}

\paragraph{Relevance of the model.}

\begin{itemize}
	\item \textbf{Faithful yet minimal}: By abstracting battery discharge to one–step losses proportional to the transmission power, the model keeps the state space small while preserving the core non–linear coupling between actions and dynamics. This parsimony lets us compute invariant distributions with reduced errors, making discrepancies with the mean–field solution easy to measure.
	\item \textbf{Validates the theory}: The set–up satisfies all structural hypotheses required for Lemma \ref{Lemma 2.1} and Theorem \ref{Theorem 4.1}—finite state–action space, unique recurrent class, Lipschitz vector field—so any divergence we observe must stem from finite–population noise, not from violated assumptions. Confirming the theoretical $\mathcal{O}(N^{-1/2})$ convergence given by both Kolmogorv's Strong Law of Large Numbers and Kurtz's Theorem in this controlled environment therefore provides compelling evidence that the analytical results are tight.
	\item \textbf{Practical relevance}: Variants of this battery–aware decision problem appear in mobile off–loading, sensor scheduling, token–based resource–sharing and edge– cloud orchestration, all discussed as application domains in Section \ref{sec:applied_implications}.
\end{itemize}

Insights gleaned from our model thus carry over to real–world systems where thousands of autonomous players compete for scarce energy or bandwidth resources.\newpage

\subsection{Modelling with Fixed Policy}

We introduce:\begin{itemize}
	\item \code{S}, number of states, so that $\mathcal{S}=\{1,\hdots,S\}$, which represent battery levels, with an extra state $0$ that means \textit{no action available}. The more natural interpretation is that $0$ corresponds to empty battery and $S$ to full.
	\item \code{q}, vector of all possible actions, so that $\mathcal{A}=q$. We denote with \code{A} the cardinality of \code{q}, so $A=|\mathcal{A}|$.
\end{itemize} At state $s$ the set of actions available to a player is given by\code{As(S, A)}, which assigns to each state $s\in\mathcal{S}$ an increasing integer $a_s\in\mathcal{A}$ so that \code{As(s) = 1:a_s} and $s<s'$ implies $a_s<a_{s'}$.

The role of the Markov transition kernel $\phi: \mathcal{S}\times\mathcal{A}\to\mathcal{P}(\mathcal{S})$ (see Section \ref{Finite Population Model}) is taken by function \code{phi(S, A, alpha, gamma, q, p0S)}, which generates an \code{AxSxS} matrix. We call this matrix \code{Q}. If we look at \code{Q(a,:,:)}, it represents the \code{SxS} matrix whose entry $(s,s')$ is the probability of going from state $s$ to state $s'$ when playing action $a$. In our practical model this probability does not depend on the states and we have the following:\begin{itemize}
	\item The probability of remaining in state $s>0$ when playing action $a$ is given by\begin{equation}
		p(a)=1-\alpha\cdot q(a)-\gamma,
	\end{equation}
	\item The probability of transitioning from state $s>0$ to state $s-1$ when playing action $a$ is $1-p(a)$,
	\item The probability of transitioning from state $s>0$ to state $s'\ne s-1$ when playing action $a$ is always $0$,
	\item The probability of transitioning from state $0$ to state $S$ is always $p_{0,S}\in[0,1]$, whatever the action played,
	\item The probability of remaining in state $0$ when playing action $a$ is always $1-p_{0,S}$.
\end{itemize} Here $\alpha$ and $\gamma$ are positive fixed coefficients such that\[\alpha\cdot q_A+\gamma\le1.\] In particular:\begin{itemize}
	\item If $\alpha=0$, then the lifetime of the battery is independent of the energy level used for transmission,
	\item If $\gamma=0$, then the expected lifetime is inversely proportional to the energy level used for transmission.
\end{itemize} Furthermore, the constant probability $p_{0,S}$ translates the following assumption.

\begin{Assumption}
	At any given time a mobile whose battery is empty may have it recharged.
\end{Assumption}

This modelling yields to the transition matrix\[\Phi(a)=\begin{bmatrix}1-p_{0,S}&1-p(a)&&&p_{0,S}\\&p(a)&1-p(a)&&\\&&p(a)&\ddots&\\&&&\ddots&1-p(a)\\&&&&p(a)\end{bmatrix},\] which is exactly what \code{phi(a,:,:)} returns.

\subsection{Modelling in the Evolutionary Dynamics}

The model in \cite{Wiecek} affirms that the reward at time $t$ for a player in state $s$ playing action $a$ when the vector of proportions of players using different actions is $w$ is given by\begin{equation}
	R(a,s,w)=\frac{a}{\sigma^2+C\cdot\sum_{l=1}^Aq_l\cdot w_l}-\beta\cdot a.
\end{equation} Here:\begin{itemize}
	\item $C$ is a constant that captures the interference of other mobiles,
	\item $\sigma^2$ is the noise power,
	\item $\beta$ is the energy cost.
\end{itemize} Thus, $R$ is the difference between the global signal to interference and noise ratio (SINR) and the energy cost $\beta\cdot a$.

From \cite[Example~5.1.1~(5.3)]{Sandholm} we choose the imitative protocol\begin{equation}
	\rho_{i,j}(\pi,x)=\frac{x_j}{m}\cdot[\pi_j-\pi_i]_+,
\end{equation} where:\begin{itemize}
	\item $i,j$ are two deterministic policies in $\mathcal{U}_D$,
	\item $\pi=F$ (see Equation \eqref{PolicyTransitions}),
	\item $x$ is the input distribution,
	\item $m$ is the mass of the population,
	\item $[\pi_j-\pi_i]_+=\max\{\pi_j-\pi_i,\ 0\}$.
\end{itemize} The individual probabilities of switching from one policy to another is given by \cite[Section~4.1.2]{Sandholm}:\begin{leftbarquote}
	a player playing strategy $i$ switches to strategy $j\ne i$ with probability $\rho_{i,j}/R$ and keeps playing strategy $i$ with probability $1-\sum_{j\ne i}\rho_{i,j}/R$,
\end{leftbarquote}where $R=R_r$.

\subsection{Choice of Parameters}

Here is the list of parameters chosen for the experimentation, accompanied by an explanation of their choice.\begin{lstlisting}
	% Number of players, taken in logspace for plotting purposes (20 log-linear values between 10^2 and 10^4)
	N     = round(logspace(2,4,20));
	
	% Number of states, representing 3 active transmission states plus a no-trasmission state
	S     = 3;
	
	% Number of available actions, rich enough to display trade-offs, but small enough to keep U_d numerically feasible.
	A     = 5;
	
	% The fact that Rr<Rd reflects players adapting less frequently than channels change
	Rd    = 1.0;
	Rr    = 0.2;
	
	% Total simulation horizon, long enough for the ODE to get near its quasi-steady-state, yet short enough to keep repeated runs computationally reasonable
	T     = 5.0; 
	
	alpha = rand();
	q = sort([0;rand(A-2,1);1])';
	gamma = 0.1*(1-alpha*q(end))*rand();
	p0S = rand();
	sigma2 = 1e-4*rand(); % Low-noise regime
	
	% Moderate interference strength to balance noise vs. multi-user coupling
	C = 4*rand(); 
	
	% Energy-cost coefficient, that penalizes higher transmit power
	beta = 4*rand(); 
	
	R = 5; % Number of independent repetitions
\end{lstlisting}

\section{Implementation Details and Simulation Outcomes}

\subsection{Coding Approach}

The deterministic limit trajectories $\mu(t)$ are obtained by integrating the mean–field ODEs \ref{Mean-Field dotmu} and \ref{ODE8} with the Runge–Kutta–based \textsc{matlab} function \code{ode45}. All empirical distributions are simulated player–by–player with Poisson distributed inter–event times, exactly matching the generator $Q$ defined in \ref{StateQ} and \ref{PolicyQ}.

For every population size $N$ and repetition $r=1,\hdots,R$ we compute the sup–norm errors\[e^{r,N}=\max_{t\in[0,T]}\norm{\hat{\mu}^{r,N}(t)-\mu(t)}_{\infty}.\] The values $e^{r,N}$ and their mean over $r$ are then plot in a log–log\footnote{Since the theoretical convergence is $\mathcal{O}\left(\frac{1}{\sqrt{N}}\right)$.} scale as functions of $N$, showing their convergence to $0$ (see Section \ref{Computational Results}).

\subsection{State Evolution Function}

In the script aimed to show the convergence of Lemma \ref{Lemma 2.1}, we simulate the state evolution through the following function, which samples player $i$'s path under a fixed policy \code{u}.\begin{lstlisting}
function [si, jumpTimes] = state_evolution(S, A, s0, T, f, u, Rd)
	% State s jumps at rate Rd,  
	% Action in state s is drawn from the fixed policy u(s,:),  
	% Next state sampled via phi.
	maxJumps = ceil(Rd * T * 2);
	si        = zeros(1, maxJumps+1);
	jumpTimes = zeros(1, maxJumps+1);
	
	idx = 1; t = 0; s = s0;
	si(idx)        = s;
	jumpTimes(idx) = t;
	
	while true
		\,\mathrm{d}t = -log(rand) / Rd;
		t  = t + \,\mathrm{d}t;
		if t > T, break; end
		
		% sample action according to u(s,:)
		a     = randsample(1:A, 1, true, u(s,:));
		probs = squeeze(f(a, s, :))';
		s     = randsample(1:S, 1, true, probs);
		
		idx = idx + 1;
		si(idx)        = s;
		jumpTimes(idx) = t;
		end
		
		si        = si(1:idx);
		jumpTimes = jumpTimes(1:idx);
	end
end
\end{lstlisting}

The fixed policy $u_i$ assigned to player $i$ is a randomized distribution over $\mathcal{S}\times\mathcal{A}$ such that, as the model requires:\begin{itemize}
	\item $u_i(a\mid s)=0$ for all $a\notin \mathcal{A}(s)$ and for all $s\in\mathcal{S}$,
	\item $\sum_{a\in\mathcal{A}}u(a\mid s)=1$ for all $s\in\mathcal{S}$.
\end{itemize} This is translated in the following simple piece of code:\begin{lstlisting}
u = rand(S,A);
for s = 1:S
	u(s, actions(s)+1:A) = 0;
	u(s,:) = u(s,:) / sum(u(s,:));
end
\end{lstlisting} where \code{actions(s)} is the implementation of $\mathcal{A}(s)$.

An analogous but more complex function named \code{joint_evolution_dynamic} is designed to show the convergence of Theorem \ref{Theorem 4.1}, in which not only states but also policies change for a single player throughout the process, so both the Markov transition kernel $\phi$ and the revision protocol $\rho$ must be included to model the evolutionary dynamics. The rest of the script is just an adaptation that take into account this modification of the problem's structure.

\subsection{Computational Results} \label{Computational Results}

The two scripts plot the sup–norm errors, their mean and the \code{polyfit} of the mean. Even with a small number of independent repetitions (\code{R = 5}) and population size (\code{N(end) = 10^4}), we obtain convergence of order $\mathcal{O}(N^{-1/2})$, as given by both Kolmogorov's Strong Law of Large Numbers and Kurtz's Theorem.

Both Figures \ref{Lemma Convergence} and \ref{Theorem Convergence} show this behaviour in the two cases.

\begin{figure}[H]
	\centering
	\includegraphics[scale=0.3]{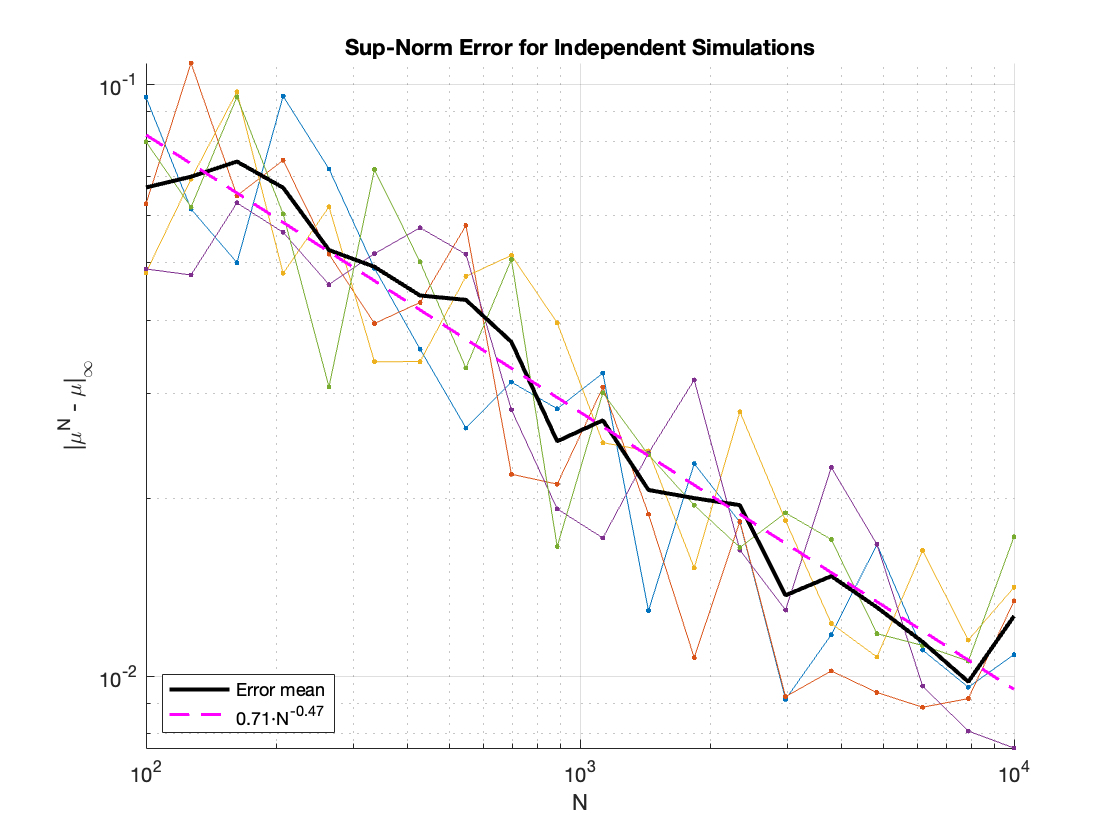}
	\caption{Convergence of Lemma \ref{Lemma 2.1}}
	\label{Lemma Convergence}
\end{figure}

\begin{figure}[H]
	\centering
	\includegraphics[scale=0.3]{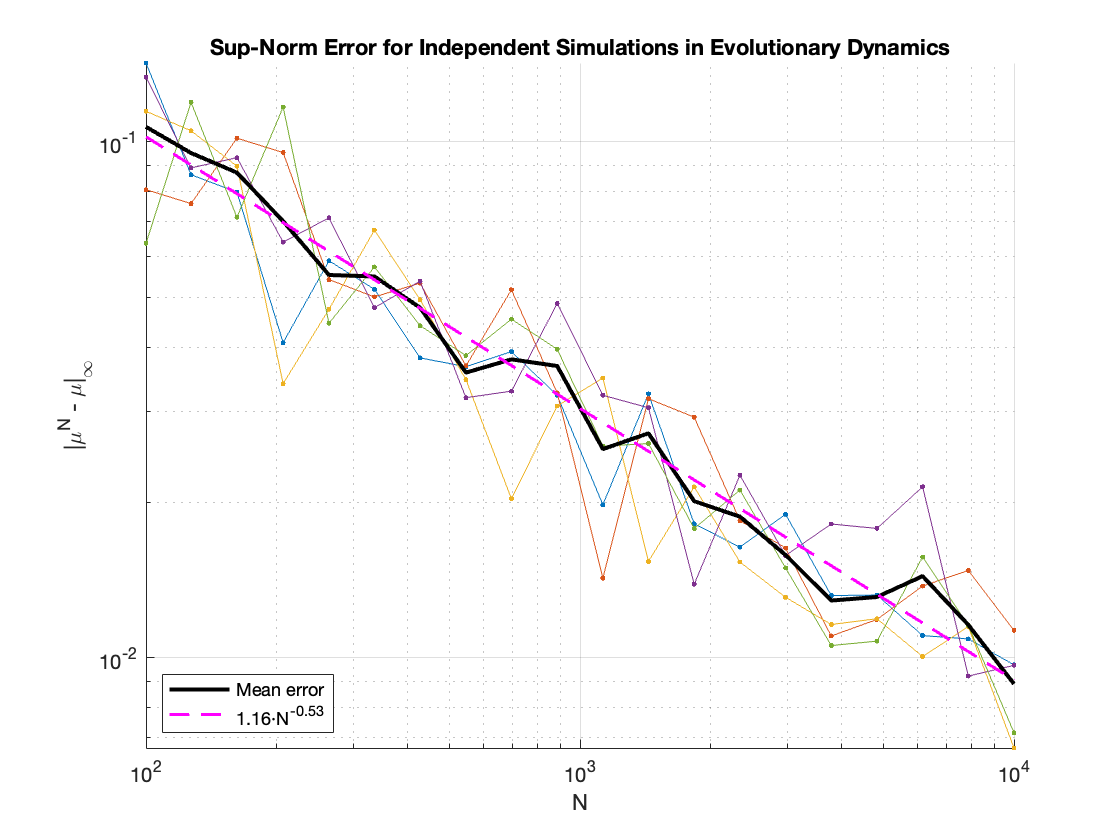}
	\caption{Convergence of Theorem \ref{Theorem 4.1}}
	\label{Theorem Convergence}
\end{figure}

Repeated executions of the two scripts give satisfactory results, as shown in Tables \ref{Table Lemma} and \ref{Table Theorem}.

\begin{table}[H]
	\centering
	\caption{Results of 10 Executions of the Script for Lemma \ref{Lemma 2.1}}
	\begin{tabular}{|l||c|c|c|c|c|c|c|c|c|c|}
		\hline
		Coefficient & 0.74 & 0.64 & 0.88&0.99&0.88&0.87&0.56&0.87&0.78&0.85\\
		\hline
		Exponent & -0.46 & -0.46 &-0.47&-0.48&-0.49&-0.49&-0.39&-0.49&-0.48&-0.46\\
		\hline
	\end{tabular}
	\label{Table Lemma}
\end{table}

\begin{table}[H]
	\centering
	\caption{Results of 10 Executions of the Script for Theorem \ref{Theorem 4.1}}
	\begin{tabular}{|l||c|c|c|c|c|c|c|c|c|c|}
		\hline
		Coefficient & 0.80 & 0.93 & 0.85&0.86&0.86&0.84&0.97&0.89&0.98&1.06\\
		\hline
		Exponent & -0.46 & -0.50 &-0.47&-0.48&-0.47&-0.48&-0.50&-0.48&-0.49&-0.51\\
		\hline
	\end{tabular}
	\label{Table Theorem}
\end{table}

 In fact, the means of the exponents of the collected data are:\begin{itemize}
	\item For Lemma \ref{Lemma 2.1}: $-0.47\pm0.05$
	\item For Theorem \ref{Theorem 4.1}: $-0.48\pm0.03$
\end{itemize} matching the theoretical value of $-0.50$.

\section{Source Code Access}

The \textsc{matlab} scripts are free to acess here (digital version only):\begin{itemize}
	\item \href{https://github.com/pietrograssiunipi/scripts/blob/68d4ebc5af694b99906f3b33397af77a6ad45fa3/Convergence_Lemma_Def.m}{\underline{Convergence of Lemma \ref{Lemma 2.1}}}
	\item \href{https://github.com/pietrograssiunipi/scripts/blob/68d4ebc5af694b99906f3b33397af77a6ad45fa3/Convergence_Theorem_Def.m}{\underline{Convergence of Theorem \ref{Theorem 4.1}}}
\end{itemize}

\chapter{Conclusions}

\section{Summary of Results}\label{sec:summary_results}

This thesis aimed to address a central question in the theory of continuous–time mean–field evolutionary games: to what extent does the finite–population stochastic process approximate the deterministic mean–field model as the number of players grows?

To tackle this question, three specific research objectives were pursued:\begin{itemize}
	\item \textbf{Establish a solid probabilistic framework} based on Markov chains with a single positive–recurrent class and an ergodic theorem tailored to that setting.
	\item \textbf{Prove the existence, uniqueness, and Lipschitz continuity} of the mean– field model, ensuring its analytical well–posedness.
	\item \textbf{Prove the convergence} of the empirical distributions to the deterministic limit of the mean–field model, both in the fixed policy and evolutionary dynamics frameworks.
\end{itemize}

\subsection{Main Theoretical Findings}

\begin{itemize}
	\item \textbf{Ergodic theorems for finite Markov chains.} Section \ref{Ergodic Theorem Model} introduces two refined versions of the Ergodic Theorem that guarantee the existence of a unique invariant distribution and almost–sure convergence of time averages whenever the chain possesses a single positive–recurrent class.  These results underpin the treatment of empirical state frequencies in the long–run limit.
	
	\item \textbf{Well–posedness of the mean–field model and convergence to the deterministic limit.} The vector field that generates the mean–field ODE system (Equation \eqref{Mean-Field dotmu}) is shown to be Lipschitz on a convex, compact domain. An extension of the Picard–Lindelöf theorem then ensures existence and uniqueness of solutions together with continuous dependence on initial conditions. This deterministic solution is then proved to be the limit of the empirical distributions (Lemma \ref{Lemma 2.1}).
	
	\item \textbf{Approximation of equilibria.} Theorem \ref{Theorem 3.2} demonstrates that a Mixed– Strategy Nash equilibrium (MSNE) of the mean–field game can be approximated, within an arbitrary $\varepsilon>0$, by a proper collection of strategies of the corresponding $N$–player game for sufficiently large $N$, which is proved to be—in a slightly weaker way—a MSNE itself. The proof hinges on almost–sure convergence of empirical state–action distributions and the dominated convergence theorem.
	
	\item \textbf{Convergence in evolutionary dynamics.}  
	By applying Kurtz’s Theorem (Theorem \ref{Kurtz}) to the empirical state–policy process, Theorem \ref{Theorem 4.1} establishes convergence in probability towards the deterministic ODE solution at each finite time horizon, thus formalising the subtitle \textit{From Stochastic Crowds to Deterministic Limits}.
\end{itemize}

\subsection{Computational Validation}

Chapter \ref{Practical Results} reports an extensive simulation campaign carried out in \textsc{matlab} that faithfully reproduces the Markov dynamics introduced in Chapter \ref{Mathematical Framework}.  Key highlights are:\begin{itemize}
	\item \textbf{Experimental set–up}: population sizes $N$ ranging from $10^{2}$ to $10^{4}$, time horizon $T=5$, and parameters calibrated to reflect moderate congestion and low exogenous noise.
	\item \textbf{Sup–norm error on the distribution}: for each $N$ and replication, the error is measured as
	\[
	e_{r,N} \;=\; \max_{t \in [0,T]}\,\bigl\lVert\hat{\mu}_{r,N}(t)\;-\;\mu(t)\bigr\rVert_{\infty}.
	\]
	\item \textbf{Empirical rate of convergence}: a log–log plot of the mean errors exhibits a slope compatible with the theoretical estimate $\mathcal{O}(N^{-1/2})$, confirming the Law– of–Large–Numbers scaling.
\end{itemize}

\section{Applied Implications}\label{sec:applied_implications}

While the results are primarily theoretical, they resonate across several application domains in which large populations of players interact asynchronously. This section outlines four emblematic scenarios and explains how the findings of this thesis translate into actionable design principles.

\subsection{Token–Based Resource–Sharing}\label{subsec:token_resource_sharing}

General references: \cite{Pedroso}, \cite{Fan}, \cite{Narander}.

\paragraph{Problem Setting.}
Token economies are increasingly adopted to regulate access to scarce resources—from compute cycles in edge clouds to electricity in micro–grids. Each player holds a portable budget of digital tokens and spends them to obtain service. The designer's objective is to keep the system \emph{incentive compatible} while guaranteeing fairness and high utilisation.

\paragraph{Implications of the Mean–Field Limit.}
The approximation results in Chapter~\ref{Theoretical Results} imply that a pricing mechanism optimised in the deterministic mean–field limit retains near–optimal properties in the corresponding finite system, provided the population is sufficiently large. In practice, this allows designers to:\begin{itemize}
	\item compute equilibrium token prices via low–dimensional ODEs rather	than high– dimensional Markov chains;
	\item exploit the $\mathcal{O}(N^{-1/2})$ error bound to decide how many tokens must circulate in order to keep price volatility below a desired	threshold;
	\item predict congestion episodes by monitoring the deterministic trajectory—a real– time proxy for aggregate token demand.
\end{itemize}

\subsection{Decentralised Market Platforms}\label{subsec:decentralised_markets}

General references: \cite{Bayraktar}, \cite{Chen}, \cite{Sequeira}.

\paragraph{Use Case.}

Blockchain–based exchanges and NFT marketplaces routinely host thousands of autonomous trading bots submitting orders asynchronously. Strategic behaviour is influenced by aggregate market states such as price trends and order–book imbalances.

\paragraph{Leveraging Equilibrium Approximation.}

Because any mean–field Nash equilibrium is $\varepsilon$–close to a finite $N$–player equilibrium (see Theorem \ref{Theorem 3.2}), platform designers can:\begin{itemize}
	\item analyse \emph{front–running risks} and liquidity fragmentation via the tractable mean– field game;
	\item implement fee schedules or market–making incentives that target the deterministic equilibrium, confident that these incentives remain effective in the finite player economy;
	\item evaluate the welfare impact of policy changes without resorting to costly player– based simulations.
\end{itemize}

\subsection{Cloud Service Orchestration}\label{subsec:cloud_orchestration}

General references: \cite{Anselmi}, \cite{Cloud}.

\paragraph{Scenario.}

Modern cloud platforms orchestrate micro–services over geo–distributed resources. Autoscalers and schedulers act as players that asynchronously request CPU, memory, and storage according to local load signals.

\paragraph{Practical Benefits of the Thesis Results.}

The mean–field framework delivers:\begin{itemize}
	\item closed–form scaling laws that relate population size, service demand, and latency;
	\item error bounds that inform how many redundant pods or replicas are required to ensure probabilistic \emph{service–level objectives};
\end{itemize}

\section{Limitations of the Study}\label{sec:limitations}

The limitations of this work are grouped and analyzed in this section into two categories: modelling assumptions and methodological and computational constraints.

\subsection{Modelling Assumptions}\label{subsec:modelling_assumptions}

\paragraph{Finite State Space.}

The analysis is confined to a \emph{finite} set of states and actions. Although this choice simplifies the proofs, many real systems may require countable or continuum state spaces. Extending the current results to such settings would involve more sophisticated theoretical arguments.

\paragraph{Single Positive–Recurrent Class.}

All Markov chains considered possess exactly one positive–recurrent class. This precludes models in which players can permanently drift into \emph{absorbing} or \emph{null–recurrent} regions.  As a consequence, the ergodic theorems developed in Section \ref{Ergodic Theorem Model} might fail to hold when multiple recurrent classes coexist or when the chain exhibits phase transitions.

\paragraph{Homogeneous Revision Protocols.}

The revision mechanism that drives strategy updates is identical across players and time. Heterogeneous revision rates—a common feature in behavioural economics—could alter convergence speeds and equilibrium selection, thereby demanding new analytical tools.

\paragraph{Behavioural Realism.}

Agents are assumed to be perfectly rational within their revision protocols. Bounded rationality and social influence are not considered, yet they play a crucial role in many large–scale socio–technical systems.

\paragraph{Omniscience.}\label{subsec:learning_algorithms}

In practical systems, players rarely know the game structure a~priori and must learn payoffs and transition probabilities online.

\subsection{Methodological and Computational Constraints}\label{subsec:methodological_constraints}

\paragraph{Global Lipschitz Condition.}

Well–posedness relies on a \emph{global} Lipschitz constant for the mean–field vector field, while we expect real–world applications to have guarantees of only \emph{local} Lipschitz conditions, leading to finite–time blow–ups or non–unique solutions. The current proofs would need substantial revision to accommodate such cases.

\paragraph{Finite Time Horizon and Limited Population Size in Simulations.}

The empirical study in Chapter~\ref{Practical Results} considers a fixed horizon $T=5$. Long–horizon or steady–state properties are inferred indirectly from theoretical results rather than observed. A comprehensive numerical analysis over extended horizons would provide stronger evidence but raises computational challenges.

Furthermore, simulations are limited to $N \leq 10^{4}$. Although sufficient to illustrate the $\mathcal{O}(N^{-1/2})$ convergence rate, larger populations could reveal second–order effects.

\section{Closing Remarks}\label{sec:closing_remarks}

This final section offers a concise \emph{take–home message}. It revisits the guiding question posed at the outset and distils the main insights gained along the way.

\subsection{From Stochastic Crowds to Deterministic Limits}\label{subsec:stochastic_to_deterministic}

At its heart, the thesis demonstrates how a population of asynchronously updating players—each governed by microscopic randomness—can give rise to a predictable macroscopic law. In doing so, the work reinforces a central theme in applied mathematics: order emerges from randomness when viewed at the right scale.

\subsection{Ethical and Social Considerations}\label{subsec:ethical_considerations}

Any model that deals with large populations carries ethical reflections. The mean–field framework must be deployed with an awareness of:\begin{itemize}
	\item \textbf{Fairness}: ensuring that aggregated control policies do not systematically disadvantage minority sub–populations that may be invisible in the limit.
	\item \textbf{Transparency}: communicating model assumptions and approximation errors to stakeholders who rely on the resulting decisions.
	\item \textbf{Robustness}: guarding against model misspecification and adversarial manipulation, especially in economic or security contexts.
\end{itemize}

Addressing these aspects will be as important as advancing the mathematics itself.

\subsection{Final Outlook}\label{subsec:final_outlook}

Whether applied to climate-‐driven resource allocation, large‐-scale autonomous mobility, or decentralised finance, the core message of mean–field methodologies endures:\begin{leftbarquote}
	The whole is greater than the sum of its parts.
\end{leftbarquote}

\appendix

\chapter{Trivial Limiting Behavior of a Markov Chain with Unique Null-Recurrent Class} \label{Ergodic Theorem Null Recurrent}

\section{Introduction}

In Section \ref{Ergodic Theorem Model} we discussed variants of the Ergodic Theorem at first stated in Section~\ref{ET} for irreducible Markov chains.

\begin{recall}
	In our model, irreducibility is too strong, so we opted for the weaker Assumption \ref{A3}, in which we forced our Markov chain to have a unique recurrent class $C$. Since our model considers a finite state space, $C$ results in being positive-recurrent. Thus, we proved two variants of the Ergodic Theorem, both fitting our framework:\begin{itemize}
		\item Theorem \ref{Ergodic Theorem Positive Recurrent} refers to Markov chains with general state space (possibly infinite) and unique positive-recurrent class, while all other states are transient,
		\item Theorem \ref{Ergodic Theorem Finite Space} regards Markov chains on a finite state space and having a unique recurrent class, which is then positive-recurrent by definition.
	\end{itemize}
	The first theorem is stronger than the second, in the sense that there is an obvious implication, but the proof for the second theorem is much more simple, so we decided to present it too.
\end{recall}

In this appendix chapter, we analyze the remaining case of this setup: a Markov chain on a general state space having a unique null-recurrent (see Definition \ref{Positive and Null Recurrence}) communicating class and other transient states. Specifically, we demonstrate why the usual Ergodic Theorem does not hold in this framework, but just a trivial limiting distribution arises.

\section{Ergodic Theorem for a Markov Chain with Unique Null-Recurrent Class}

In short, there is no finite stationary distribution in the null‐recurrent case, so we cannot have\[
\frac{1}{t}\int_{0}^{t} f(X_s)\,\mathrm{d}s \xrightarrow[t\to\infty]{}\sum_{i}\lambda_if(i)
\quad(\text{for any probability measure }\lambda).
\] In fact, the chain still spends all of its infinite time in the single recurrent class $C$, but within $C$ it keeps returning to every state infinitely often while not settling down to spend a positive fraction of time at any single state.

A rigorous argument shows that, for each individual state $i \in C$,\[
\frac{1}{t}\int_{0}^{t}\mathds{1}_{\{X_s = i\}}\,\mathrm{d}s \xrightarrow[t\to\infty]{\mbox{a.s.}} 0.
\] Hence, no nontrivial limiting distribution can arise.

The proof of the theorem below mirrors the style used in the positive‐recurrent case, highlighting the crucial points where null recurrence differs from positive recurrence.

\begin{theorem}[Ergodic Theorem for a Markov Chain with Unique Null-Recurrent Class] \label{Ergodic Theorem Null}
	Consider a continuous‐time Markov chain $(X_t)_{t\ge0}$ on a countable state space $I$ with generator matrix $Q$. Suppose:\begin{itemize}
		\item[(i)] There is exactly one communicating class $C\subseteq I$ that is recurrent and null (i.e., it is recurrent but the expected return times to each state in C are infinite),
		\item[(ii)] Every other state in $I\setminus C$ is transient.
	\end{itemize} Then:\begin{itemize}
		\item The chain does not admit a stationary distribution that is a finite probability measure.
		\item Nonetheless, the chain, with probability $1$, visits $C$ infinitely often (in fact, eventually stays in $C$ except for possibly finitely many returns to transients, which must occur only finitely many times since those states are transient).
		\item Within $C$, the chain’s limiting behavior does not concentrate on any single state or subset. In fact, for each $i\in C$, the fraction of time spent at $i$ converges to $0$.
	\end{itemize} Therefore, in the sense of time averages,\[\frac{1}{t}\int_0^t\mathds{1}_{\{X_s=i\}}\,\mathrm{d}s\xrightarrow[t\to\infty]{\mbox{a.s.}}0\quad\mbox{for each}\ i\in I.\] Equivalently, for any bounded function $f: I\to\mathbb{R}$ whose support is finite,\[\frac{1}{t}\int_0^tf(X_s)\,\mathrm{d}s\xrightarrow[t\to\infty]{\mbox{a.s.}}0.\] In other words, there is no nontrivial limit mimicking the Ergodic Theorem for positive‐recurrent chains (Theorem \ref{Ergodic Theorem Positive Recurrent}).
	\begin{proof}
		We begin with observing that the chain eventually stays in $C$ forever, up to finitely many excursions: since all states outside $C$ are transient, we know that with probability $1$ the process $X_t$ can only visit those transient states finitely many times. Thus, from some random time onward, $X_t$ lies entirely in $C$. Hence for large $t$, the chain is effectively a null‐recurrent chain on $C$ alone (we say that the chain’s trace on $C$ is a null‐recurrent Markov chain).
		
		Having noted that, we now face null recurrence, which precludes a positive fraction of time in any single state. Here is the crucial difference from positive recurrence:\begin{itemize}
			\item In a positive‐recurrent class, the expected return time to a given state $i\in C$ is finite. As a result, the fraction of time the chain spends in i converges to a positive constant $\lambda_i$ (see Theorem \ref{Ergodic Theorem Positive Recurrent}),
			\item In a null‐recurrent class, that same expected return time to $i$ is infinite.  Intuitively, although the chain returns to $i$ infinitely often (recurrence), the returns become so spread out over long times that the overall fraction of time in $i$ goes to $0$.
		\end{itemize} Formally, let $T_i$ be the first return time to $i$ starting from $i$. In the null‐recurrent case,\[\mathbb{E}_i[T_i]=\infty.\] By applying the strong Markov property at each return to $i$, we have\[
		\frac{N_i(t)}{t}= \frac{\text{total time spent in state }i \text{ up to }t}{t}
		\xrightarrow[t\to\infty]{\mbox{a.s.}}0,
		\] since the times between successive returns to $i$ have infinite average. Here is a sketch of the argument:\begin{itemize}
			\item Number of returns grows slowly: Let $R_k$ be the time of the $k$-th return to $i$. Since $\mathbb{E}[T_i] = \infty$, the strong law of large numbers cannot apply in a standard way, and in fact\[
			R_k\xrightarrow[k\to\infty]{}\infty,
			\] and grows faster than linearly in $k$, in a loose sense. In fact, $R_k/k$ cannot be bounded, because that would imply a finite mean return time,
			\item Fraction of visits in $i$ is negligible in large time: On each visit to $i$, the chain remains there for an $E(q_i)$ holding time of finite mean $1/q_i$. Summing those small holding times infinitely often still remains negligible compared to $R_k$, which grows faster. Hence for large $t$, the total amount of time the chain has spent in $i$ up to time $t$ is small compared to $t$, so\[\frac{1}{t}\int_0^t\mathds{1}_{\{X_s=i\}}\,\mathrm{d}s=\frac{N_i(t)}{t}\xrightarrow[t\to\infty]{\mbox{a.s.}}0.\] That is exactly the statement that no single state can capture a positive fraction of the long‐run time.
		\end{itemize}
	
		Combining the above with the transience of $I\setminus C$:\begin{itemize}
			\item By transience, the chain can only be in states outside $C$ for a finite total amount of time, so\[
			\frac{1}{t}\int_{0}^{t}\mathds{1}_{\{X_s\in I\setminus C\}}\,\mathrm{d}s \xrightarrow[t\to\infty]{\mbox{a.s.}}0.
			\]
			\item For each $i\in C$,\[
			\frac{1}{t}\int_{0}^{t}\mathds{1}_{\{X_s=i\}}\,\mathrm{d}s \xrightarrow[t\to\infty]{\mbox{a.s.}}0.
			\] Summing these indicators over all $i\in C$ still gives $1$ (after time is large enough that the chain stays in $C$), but individually they are vanishing.
		\end{itemize} Hence for any bounded function $f: I\to\mathbb{R}$ that has finite support (or at least is bounded and small enough outside a finite set, see the remark below), the time average\[\frac{1}{t}\int_0^tf(X_s)\,\mathrm{d}s\] can be made arbitrarily small for large $t$. Indeed, if $\supp(f)\subseteq\{i_1,\hdots,i_n\}$, then each $i_k$ occupies vanishing fraction of time, so the linear combination also goes to zero. Thus, we conclude that\[
			\frac{1}{t}\int_{0}^{t} f(X_s)\,\mathrm{d}s \xrightarrow[t\to\infty]{\mbox{a.s.}}0.
		\]
	\end{proof}
	\begin{remark}[Bounded Functions which Are Small Enough Outside a Finite Set]
		If $f$ is a general bounded function on a countably infinite $I$, it can be approximated by truncated versions $f_N$ that vanish outside a finite set and let $N\to \infty$. It holds that\[\limsup_{t\to\infty}\left|\frac{1}{t}\int_0^tf(X_s)\,\mathrm{d}s\right|\le\sup_N\limsup_{t\to\infty}\left|\frac{1}{t}\int_0^tf(X_s)-f_N(X_s)\,\mathrm{d}s\right|\] and by Dominated Convergence Theorem \cite[Theorem~9.1.2]{Rosenthal} and the fact that $f-f_N$ is supported outside the first $N$ states, we again obtain the a.s. convergence to zero if $f$ grows slowly enough. Either way, the key point is that no finite mass can accumulate on any given state in the long run.
	\end{remark}
\end{theorem}

\bibliographystyle{amsalpha}
\bibliography{Ref}

@book {Shwartz,
	AUTHOR = {A. Shwartz},
	TITLE = {Large Deviations for Performance Analysis},
	PUBLISHER = {Routledge, Nokia of America Corporation},
	YEAR = {2018},
	ISBN = {978-1-138-31577-8},
}

@book {Norris,
	AUTHOR = {J. R. Norris},
	TITLE = {Markov Chains},
	PUBLISHER = {Cambridge University Press},
	YEAR = {2009},
	ISBN = {0-521-48181-3},
	SERIES = {Cambridge Series on Statistical and Probabilistic Mathematics},
}

@book {Adlakha,
	AUTHOR = {S. Adlakha and R. Johari and G. Y. Weintraub},
	TITLE = {Equilibria of dynamic games with many players: Existence, approximation, and market structure},
	VOLUME = {156},
	PAGES = {269-316},
	PUBLISHER = {Journal of	Economic Theory},
	YEAR = {2015},
	DOI = {10.1016/j.jet.2013.07.002},
}

@book {Altman,
	AUTHOR = {P. Wi\c ecek and E. Altman},
	TITLE = {Stationary anonymous sequential games with undiscounted rewards},
	VOLUME = {166},
	PAGES = {686-710},
	PUBLISHER = {Journal of Optimization Theory and Applications},
	YEAR = {2015},
	DOI = {10.1007/s10957-014-0649-9},
}

@article {Wiecek,
	AUTHOR = {P. Wi\c ecek and E. Altman and Y. Hayel},
	TITLE = {Stochastic state dependent population games in wireless communication},
	VOLUME = {56(3)},
	PAGES = {492-505},
	PUBLISHER = {IEEE Transactions on Automatic Control},
	YEAR = {2011},
}

@notes {Durrett,
	AUTHOR = {R. Durrett},
	TITLE = {Essentials of Stochastic Processes},
	PUBLISHER = {Cambridge University Press},
	YEAR = {2010},
}

@book {Smirnov,
	AUTHOR = {G. V. Smirnov},
	TITLE = {Introduction to the Theory of Differential Inclusions},
	PUBLISHER = {American Mathematical Society},
	YEAR = {2002},
}

@book {Sandholm,
	AUTHOR = {W. H. Sandholm},
	TITLE = {Population Games and Evolutionary Dynamics},
	PUBLISHER = {MIT Press},
	YEAR = {2010},
}

@book {Rosenthal,
	AUTHOR = {J. S. Rosenthal},
	TITLE = {A First Look at Rigorous Probability Theory},
	PUBLISHER = {World Scientific Publishing Co.},
	YEAR = {2006},
}

@article {Kurtz,
	AUTHOR = {T. G. Kurtz},
	TITLE = {Strong Approximation Theorems for Density Dependent Markov Chains},
	VOLUME = {6},
	PAGES = {223-240},
	PUBLISHER = {North-Holland Publishing Company},
	YEAR = {1978},
}

@article {CMP,
	AUTHOR = {H. B. Mann and A. Wald},
	TITLE = {On Stochastic Limit and Order Relationships},
	VOLUME = {14},
	PAGES = {217-226},
	PUBLISHER = {Institute of Mathematical Statistics},
	YEAR = {1943},
}

@book {Lemmas,
	AUTHOR = {R. Durrett},
	TITLE = {Probability: Theory and Examples},
	PUBLISHER = {Cambridge University Press},
	YEAR = {2019},
}

@inproceedings{Pedroso,
	AUTHOR = {Leonardo Pedroso and Andrea Agazzi and W. P. M. H. Heemels and Mauro Salazar},
	TITLE = {Fair Artificial Currency Incentives in Repeated Weighted Congestion Games: Equity vs. Equality},
	BOOKTITLE = {63rd IEEE Conference on Decision and Control},
	YEAR = {2024},
	PAGES = {954-959},
	DOI = {10.1109/CDC56724.2024.10886786},
}

@inproceedings{Pedroso2,
	AUTHOR = {Leonardo Pedroso and W. P. M. H. Heemels and Mauro Salazar},
	TITLE = {Urgency-aware Routing in Single Origin-destination Itineraries through Artificial Currencies},
	BOOKTITLE = {62nd IEEE Conference on Decision and Control},
	YEAR = {2023},
	PAGES = {4142-4149},
	DOI = {10.1109/CDC49753.2023.10383739},
}

@book {Rudin,
	AUTHOR = {W. Rudin},
	TITLE = {Principles of Mathematical Analysis},
	PUBLISHER = {McGraw-Hill},
	YEAR = {1976},
}

@article{KurtzRate,
	ISSN = {00219002},
	URL = {http://www.jstor.org/stable/3212147},
	AUTHOR = {T. G. Kurtz},
	JOURNAL = {Journal of Applied Probability},
	NUMBER = {1},
	PAGES = {49–58},
	PUBLISHER = {Applied Probability Trust},
	TITLE = {Solutions of Ordinary Differential Equations as Limits of Pure Jump Markov Processes},
	URLDATE = {2025-06-15},
	VOLUME = {7},
	YEAR = {1970}
}

@article{SLLNRate,
	ISSN = {00029947},
	URL = {http://www.jstor.org/stable/1994170},
	AUTHOR = {Leonard E. Baum and Melvin Katz},
	JOURNAL = {Transactions of the American Mathematical Society},
	NUMBER = {1},
	PAGES = {108--123},
	PUBLISHER = {American Mathematical Society},
	TITLE = {Convergence Rates in the Law of Large Numbers},
	URLDATE = {2025-06-15},
	VOLUME = {120},
	YEAR = {1965}
}

@article {Fan,
	AUTHOR = {S. Fan and J. Zhao and R. Zhao and Z. Wang and W. Cai},
	TITLE = {CryptoArcade: A Cloud Gaming System with Blockchain-based Token Economy},
	VOLUME = {11 (3)},
	PAGES = {2445-2458},
	PUBLISHER = {IEEE Transactions on Cloud Computing},
	YEAR = {2023},
}

@article {Narander,
	AUTHOR = {K. Narander and S. Swati},
	TITLE = {Token-based Predictive Scheduling of Tasks in Cloud Data-centers},
	VOLUME = {4},
	PAGES = {29-33},
	PUBLISHER = {Research Journal of Recent Sciences},
	YEAR = {2015},
}

@book {Bayraktar,
	AUTHOR = {E. Bayraktar and A. Cohen and A. Nellis},
	TITLE = {DEX Specs: A Mean Field Approach to DeFi Currency Exchanges},
	PUBLISHER = {SSRN Electronic Journal},
	YEAR = {2024},
	DOI = {10.2139/ssrn.4796297},
}

@book {Chen,
	AUTHOR = {Y. Chen and L. Wu and R. Xu and R. Zhang},
	TITLE = {Periodic Trading Activities in Financial Markets: Mean-field Liquidation Game with Major-Minor Players},
	YEAR = {2024},
}

@book {Sequeira,
	AUTHOR = {J.I. Sequeira and A.M. González and R.O. Illera},
	TITLE = {Liquidity Pools as Mean Field Games: A New Framework},
	YEAR = {2024},
}

@article {Anselmi,
	AUTHOR = {J. Anselmi},
	TITLE = {Asynchronous Load Balancing and Auto-Scaling: Mean-Field Limit and Optimal Design},
	VOLUME = {32 (4)},
	PAGES = {2960-2971},
	PUBLISHER = {IEEE/ACM Transactions on Networking},
	YEAR = {2024},
}

@book {Cloud,
	AUTHOR = {V. Sachidananda and A. Sivaraman},
	TITLE = {Collective Autoscaling for Cloud Microservices},
	YEAR = {2022},
}

\chapter*{Acknowledgments}
\addcontentsline{toc}{chapter}{Acknowledgments}

Here we are, with the most read chapter of the thesis, one of the few pages that will no longer be untouched but will bear a few small signs of wear.

First and foremost thanks to God, who guides my path.

A special thank you to Prof. Agazzi, Eng. Leonardo Pedroso and Prof. Dario Trevisan for your support, which made it possible for me to complete this work. I also want to thank all the professors at the University of Pisa and those at Liceo Enriques who contributed to my education.

Thanks to all the dear people I care about, who have been part of the person I am today, whether for a long time or a short while: to Vio, my parents, and all my family, to Nunzio, Andre, and Edo, to Gabry, Bracas, Giacomo, Eli, Ale, Jacopone, Mattia, Anna, and Chiara, to Eva, Gaia, Cristi, Lore, Fili, Jacopo, Virgi, Rache, Gioele, Mariam, Eli and Andre, to Albe, Giuba, Bienti, Berto, Sofi, Sese, Costy, and Lena, to Don Giovanni and Don Andrea, to Patri, Stefano, Lino, and Laura, to Daddy Gaiffi, the Tracotanti, and all the Sezione AIA of Livorno.

I am grateful to all of you for your presence, support, and encouragement throughout this journey, which you made special.

\end{document}